\newif\ifsubmit     
\newif\ifllncs      
\newif\ifexabs      
\newif\ifblind  

\submittrue

\ifllncs
\documentclass[runningheads,a4paper]{llncs}

\else \documentclass[letterpaper,11pt,pdfa]{article}
  \usepackage[in]{fullpage}
\fi

\usepackage{iftex}
\ifPDFTeX
  \usepackage[utf8]{inputenc}
  \usepackage[noTeX]{mmap}
  \usepackage[T1]{fontenc}
\fi
\ifLuaTeX
\usepackage{luatex85}
\usepackage[noTeX]{mmap}
\fi

\usepackage{mdframed}
\usepackage{amsmath}
\usepackage{amsfonts}
\usepackage{amssymb}
\usepackage{amsthm}
\usepackage{color}
\usepackage[bookmarks]{hyperref}
\usepackage[nameinlink]{cleveref}

\usepackage{mdframed}
\usepackage{amsmath}
\usepackage{amsfonts}
\usepackage[nottoc,numbib]{tocbibind}
\usepackage{amssymb}
\usepackage{amsthm}
\usepackage{color}
\usepackage[bookmarks]{hyperref}
\usepackage[nameinlink]{cleveref}
\ifllncs

  \spnewtheorem{claim}{Claim}{\bfseries}{\rmfamily}
  \crefname{claim}{claim}{claims}
  \Crefname{claim}{Claim}{Claims}
\else
  \newtheorem{theorem}{Theorem}[section]
  \newtheorem{definition}[theorem]{Definition}
  \newtheorem{remark}[theorem]{Remark}
  \newtheorem{lemma}[theorem]{Lemma}
  \newtheorem{corollary}[theorem]{Corollary}
  
  \newtheorem{claim}[theorem]{Claim}
  \newtheorem*{remark*}{Remark}
  
  \newtheorem{fact}[theorem]{Fact}
  
  \newtheorem{observation}[theorem]{Observation}
 \newtheorem{conjecture}[theorem]{Conjecture}
  \newtheorem{assumption}[theorem]{Assumption}
  \newtheorem*{theorem*}{Theorem}
  \newtheorem*{lemma*}{Lemma}

\fi

\usepackage{appendix}
\usepackage{algorithm}
\usepackage{algpseudocode}
\usepackage{braket}
\usepackage{comment}
\usepackage{url}
\usepackage{multicol}
\usepackage{tikz}
\usetikzlibrary{arrows,automata,positioning}
\usepackage{caption}
\usepackage[caption=false]{subfig}
\usepackage[font=small,labelfont=bf]{caption}
\usepackage{comment}
\usepackage{pifont}
\usetikzlibrary{patterns}

\usepackage{enumitem}
\setlist[description]{noitemsep}
\setlist[enumerate]{noitemsep}
\setlist[itemize]{noitemsep}

\usepackage{soul, xcolor, xparse}
\makeatletter
  \ExplSyntaxOn
    \cs_new:Npn \white_text:n #1
    {
      \fp_set:Nn \l_tmpa_fp {#1 * .01}
      \llap{\textcolor{white}{\the\SOUL@syllable}\hspace{\fp_to_decimal:N \l_tmpa_fp em}}
      \llap{\textcolor{white}{\the\SOUL@syllable}\hspace{-\fp_to_decimal:N \l_tmpa_fp em}}
    }
    \NewDocumentCommand{\whiten}{ m }
    {
      \int_step_function:nnnN {1}{1}{#1} \white_text:n
    }
  \ExplSyntaxOff
  
  \NewDocumentCommand{ \varul }{ D<>{5} O{0.2ex} O{0.1ex} +m } {%
    \begingroup
    \setul{#2}{#3}%
    \def\SOUL@uleverysyllable{%
      \setbox0=\hbox{\the\SOUL@syllable}%
      \ifdim\dp0>\z@
      \SOUL@ulunderline{\phantom{\the\SOUL@syllable}}%
      \whiten{#1}%
      \llap{%
        \the\SOUL@syllable
        \SOUL@setkern\SOUL@charkern
      }%
      \else
      \SOUL@ulunderline{%
        \the\SOUL@syllable
        \SOUL@setkern\SOUL@charkern
      }%
      \fi}%
    \ul{#4}%
    \endgroup
  }
\makeatother

\usepackage{braket}

\newcommand{\revise}[1]{#1}

\ifsubmit
    \newcommand{\aayush}[1]{}
    \newcommand{\vipul}[1]{}
    \newcommand{\jiahui}[1]{}
    \newcommand{\orestis}[1]{}
\else
    \newcommand{\aayush}[1]{{\color{magenta} Aayush: #1}}
    \newcommand{\vipul}[1]{{\color{red} Vipul: #1}}
    \newcommand{\jiahui}[1]{{\color{purple} Jiahui: #1}}
    \newcommand{\orestis}[1]{{\color{blue} Orestis: #1}}
\fi

\newcommand{\Enc}{\mathsf{Enc}}
\newcommand{\KeyGen}{\mathsf{KeyGen}}

\newcommand{\InKeyGen}{\mathsf{InKeyGen}}
\newcommand{\Dec}{\mathsf{Dec}}
\newcommand{\Del}{\mathsf{Del}}
\newcommand{\VerDel}{\mathsf{VerDel}}
\newcommand{\Setup}{\mathsf{Setup}}
\newcommand{\samp}{\xleftarrow{\$}}
\newcommand{\td}{\mathsf{td}}
\newcommand{\pk}{\mathsf{pk}}
\newcommand{\mpk}{\mathsf{mpk}}
\newcommand{\sk}{\mathsf{sk}}
\newcommand{\ct}{\mathsf{ct}}
\newcommand{\cert}{\mathsf{cert}}

\newcommand{\obligate}{{\sf Obligate}}
\newcommand{\solve}{{\sf Solve}}

\newcommand{\negl}{\mathsf{negl}}
\newcommand{\QPT}{\mathsf{QPT}}
\newcommand{\A}{\mathcal{A}}

\newcommand{\poly}{{\sf poly}}

\DeclareMathOperator{\Tr}{Tr}

\newcommand{\cD}{{\mathcal{D}}}
\newcommand{\cN}{{\mathcal{N}}}
\newcommand{\cM}{{\mathcal{M}}}
\newcommand{\cE}{{\mathcal{E}}}
\newcommand{\cR}{{\mathcal{R}}}
\newcommand{\cI}{{\mathcal{I}}}
\newcommand{\cQ}{{\mathcal{Q}}}

\newcommand{\cJ}{{\mathcal{J}}}

\newcommand{\cH}{{\mathcal{H}}}

\newcommand{\sX}{\mathcal{X}}
\newcommand{\sY}{\mathcal{Y}}
\newcommand{\sR}{\mathcal{R}}

\newcommand{\E}{\mathbb{E}}
\newcommand{\Z}{\mathbb{Z}}
\newcommand{\N}{\mathbb{N}}

\newcommand{\bA}{\mathbf{A}}
\newcommand{\bB}{\mathbf{B}}

\newcommand{\bb}{\mathbf{b}}

\newcommand{\bR}
{\mathbf{R}}

\newcommand{\bE}
{\mathbf{E}}

\newcommand{\bG}
{\mathbf{G}}

\newcommand{\bC}
{\mathbf{C}}

\newcommand{\bI}
{\mathbf{I}}

\newcommand{\by}
{\mathbf{y}}
\newcommand{\bx}
{\mathbf{x}}

\newcommand{\bs}{\mathbf{s}}
\newcommand{\be}{\mathbf{e}}
\newcommand{\bd}{\mathbf{d}}

\newcommand{\bc}{\mathbf{c}}

\newcommand{\bu}{\mathbf{u}}

\newcommand{\bsk}
{\mathbf{sk}}

\newcommand{\bw}
{\mathbf{w}}

\newcommand{\api}{{\sf API}}
\newcommand{\ati}{{\sf ATI}}

\newcommand{\ti}{{\sf TI}}
\newcommand{\projimp}{{\sf ProjImp}}
\newcommand{\cproj}{{\sf CProj}}

\newcommand{\cP}{{\mathcal{P}}}

\newcommand{\isuniform}{{\sf IsUniform}}

\newcommand{\unif}{{\sf unif}}

\newcommand{\supp}{\mathsf{SUPP}}

\newcommand{\qsk}{{\sf \rho_\sk}}

\newcommand{\Ext}{\mathop{\mathsf{Ext}}}

\newcommand{\lwe}{{\sf LWE}}

\newcommand{\valid}{\mathsf{Valid}}
\newcommand{\invalid}{\mathsf{Invalid}}

\newcommand{\delete}{\mathsf{Delete}}

\newcommand{\INV}{\mathsf{INV}}

\newcommand{\oneoftwo}{\texorpdfstring{$1$-of-$2$}{1-of-2} }

\usepackage{authblk}

\title{Quantum Key Leasing for PKE and FHE with a Classical Lessor}
\ifblind
\else
\author[1]{Orestis Chardouvelis}
\author[2]{Vipul Goyal}
\author[3]{Aayush Jain}
\author[4]{Jiahui Liu}
\affil[1,2,3]{Carnegie Mellon University}
\affil[2]{NTT Research}
\affil[4]{Massachusetts Institute of Technology}
\fi
\date{today}

\begin{document}

\maketitle

In this work, we consider the problem of secure key leasing, also known as revocable cryptography (Agarwal et. al. Eurocrypt' 23, Ananth et. al. TCC' 23), as a strengthened security notion to its predecessor put forward in Ananth et. al. (Eurocrypt' 21). This problem aims to leverage unclonable nature of quantum information to allow a lessor to lease a quantum key with reusability for evaluating a classical functionality. Later, the lessor can request the lessee to provably delete the key and then the lessee will be completely deprived of the capability to evaluate.
In this work, we construct a secure key leasing scheme to lease a decryption key of a (classical) public-key, homomorphic encryption scheme from standard lattice assumptions. Our encryption scheme is exactly identical to the (primal) version of Gentry-Sahai-Waters homomorphic encryption scheme with a carefully chosen public key matrix. We achieve strong form of security where:
\begin{itemize}
\item The entire protocol (including key generation and verification of deletion)  uses merely classical communication between a \emph{classical lessor (client)} and a quantum lessee (server). 
\item Assuming standard assumptions, our security definition ensures that every computationally bounded quantum adversary could only simultaneously provide a valid classical deletion certificate and yet distinguish ciphertexts with at most some negligible probability.
\end{itemize}
Our security relies on subexponential time hardness of learning with errors assumption. Our scheme is the first scheme to be based on a standard assumption and satisfying the two properties mentioned above.

The main technical novelty in our work is the design of an FHE scheme that enables us to apply elegant analyses done in the context of classical verification of quantumness from LWE (Brakerski et. al.(FOCS'18, JACM'21) and its parallel amplified version in Radian et. al.(AFT'21)) to the setting of secure leasing. This connection to classical verification of quantumness leads to a modular construction and arguably simpler proofs than previously known. An important technical component we prove along the way is an amplified quantum search-to-decision reduction: we design an extractor that uses a quantum distinguisher (who has an internal quantum state) for decisional LWE, to extract secrets with success probability amplified to almost one. This technique might be of independent interest.

\ifllncs
\else
\newpage
\tableofcontents
\newpage
\fi

\jiahui{TO DO: 1. check all column vectors  2.  }

\section{Introduction}


Quantum information, through principles such as the \emph{no-cloning theorem}, offers exciting possibilities for which there are no classical counterparts. An active area of research are primitives such as copy protection \cite{aaronson2009quantum} and quantum money \cite{wiesner1983conjugate}, which simply can't be realized classically. 

Copy protection, a task that captures a wide variety of unclonability related tasks,  is aimed at leveraging the power of unclonable quantum information to prevent piracy of software functionalities. However, constructions of quantum copy protection have turned out to be notoriously hard -- in fact, it is shown to be impossible in general \cite{ananth2020secure} and  only exist relative to structured oracles \cite{aaronsonnew} if we aim at a generic scheme secure for all copy-protectable functions. Only a handful of functions are known to be copy-protectable without using oracles in the construction, but either they are built from very strong, less-standard cryptographic assumptions (e.g. post-quantum iO) \cite{coladangelo2021hidden,liu2022collusion} or the functionalities are very limited (point functions) \cite{coladangelo2020quantum,ananth2022feasibility}. 


Meaningful relaxations to copy protection come in the form of primitives that allow revocation of the quantum software via proofs of "destruction” of the  quantum functionality after use. Thus, if the revocation or proof verifies, a computationally bounded user is deprived of any ability to evaluate the function in any meaningful ways.

In this work, we consider one such recent  notion of secure key leasing \cite{agrawal2023public}, also called key-revocable cryptography in the literature \cite{ananth2023revocable}. This notion is inspired by secure software leasing in \cite{ananth2020secure}, but possesses significantly stronger security guarantees \footnote{Our security notion is \emph{stronger} and has more applicability than the previous, similar primitive called secure software leasing (SSL) \cite{ananth2020secure}. In SSL the adversarial lessee is only semi-malicious, while in \cite{agrawal2023public,ananth2023revocable} and our work it is fully malicious. See \cref{sec:related_works} for detailed discussions.}. 


In secure key leasing, a lessor (also called the client) leases a quantum state to a lessee (also called the server). The quantum state might allow the server to obtain a capability to compute an advanced classical function, such as decryption (in case of public-key encryption with secure key leasing), or PRF evaluation, for polynomially many times.

At a later point in time, the lessor may ask the lessee to return or destroy the key for reasons such as lessee's subscription expiring or the lessor(client) deciding to use another server. The notion of secure key leasing allows the lessee to ``delete" the quantum key and create a deletion certificate which can be verified by the lessor. Once the deletion certificate is sent to the lessor and verified, the security guarantee says that the lessee provably loses the associated capability. This must hold even if the lessee is fully malicious and may make arbitrary attempts to copy the quantum state or learn the underlying classical functionality before the deletion. A number of recent papers have studied the notion of secure key leasing for primitives like PKE, FE, FHE and PRF \cite{ananth2023revocable,agrawal2023public}, in the context of a quantum lessor and quantum lessee with quantum communication. While the notion of secure software leasing is interesting on its own, remarkably it has also found an application in the seemingly unrelated context of \emph{leakage or intrusion detection} \cite{EPRINT:CGLR23} where an adversary breaks into a machine to steal a secret key and one would like to detect if such an attack happened.

Secure key leasing is a uniquely quantum phenomenon which is enabled by the no-cloning theorem, clearly impossible in the classical setting. Therefore, some quantum capabilities are necessary to enable it. However, a natural question is: \emph{to what extent are quantum capabilities necessary? Do both the client and the server need quantum capabilities? Do we require quantum communication in addition to quantum computation?} As discussed in \cite{shmueli2022public,chevalier2023semi}, even in the future when full scale quantum computer come into the household and every user can maintain quantum memory, setting up large-scale end-to-end quantum communication is still expensive and lossy. Moreover, having a client operating only on classical resources and delegating the quantum work to a server is desirable in almost all sorts of circumstances.

A recent beautiful work of  Ananth, Poremba, and Vaikuntanathan \cite{ananth2023revocable} used the dual Regev cryptosystem to build primitives such as public-key encryption (PKE), fully-homomorphic encryption (FHE), and pseudo random functions (PRFs) with key revocation. While their basic constructions require quantum communication as well as a quantum certificate of deletion, they can extend their constructions to achieve classical deletion certificates \cite{ananth2023revocable} as well. However, a key limitation is their work is that the security only holds either based on an unproven conjecture, or only against an adversarial server (leasee) which outputs a valid deletion certificate with probability negligibly close to 1. Proving security against general adversaries based on standard assumptions was left as an open problem. In addition, in both works \cite{agrawal2023public,ananth2023revocable}, the client and the server must have both quantum capabilities and quantum communication is also necessary in the key generation process. This leads to the following natural open problem: 

\begin{center}
\emph{Can we build secure leasing for unlearnable functions e.g. PKE decryption, FHE decryption, and other cryptographic functionalities  with a completely classical lessor/client? and more desirably, from standard assumptions?}
\end{center}
If so, this would lead to constructions which require only bare minimum quantum infrastructure needed in our setting.

\subsection{Our Results}

We answer all the above questions affirmatively thus settling the open problem from \cite{ananth2023revocable}. We obtain the following results:


\begin{theorem}(Informal)
    \label{tyhm:main_informal}
    Assuming the post-quantum sub-exponental hardness of Learning-with-Errors, there exists a secure key leasing scheme for PKE and FHE with a completely classical lessor/client.
    
\end{theorem}

More specifically, from LWE, we can achieve secure leasing for FHE decryption circuit with following properties:
\begin{enumerate}
    \item The key generation is a one-message protocol with one classical message from the lessor to lessee and one back. \footnote{As discussed in \Cref{sec:relation_security_defs}, we need to add a second message to deal with impostor security, which is orthogonal to the major anti-piracy security we prove. For the major anti-piracy security defined in \Cref{def:regular_skl_security_classical},\Cref{def:gamma_skl_game}, the protocol suffices to be a one-message protocol from lessor to lessee.}.
    
    \item Upon being asked to delete the key, the lessee will produce a classical certificate, and send it to the client. The client will perform a classical verification algorithm on the certificate.

    \item As natural in the definition for secure key leasing of classical PKE, the encryption is a public, classical procedure. The decryption can be done quantumly (on the lessee/server's side) or classically (on the lessor/client's side)
    
\end{enumerate}

A key contribution of our work is to connect our problem to classically verifiable proof of quantumness \cite{brakerski2021cryptographic}. The latter is a well-developed research direction with beautiful ideas. This connection allows us to design modular constructions and write proofs of security that are arguably simpler, thanks to a number of technical ideas that we could rely on the literature in classical verification of quantum computation.




  To show the security of our construction, we prove along the way a quantum search-to-decision reduction for LWE that extracts the secret \emph{with the success probability boosted to almost 1}. 

  \begin{theorem}[Informal]
  \label{thm:informal_lwe_searchtodecision}
   Given a quantum distinguisher with an auxiliary quantum state, that distinguishes LWE samples from random with inverse polynomial probability for some arbitrary secret,  there exists an extractor running in time polynomial (in the secret size, the distinguisher's time  and inverse its advantage) that produces the secret with probability negligibly close to 1.    
  \end{theorem}
 To the best of our knowledge, this is the first search-to-decision reduction with an almost-perfect search when using quantum algorithms with  auxiliary quantum states.
 To prove the above theorem, we rely on a measurement implementation on the distinguisher's quantum state \cite{marriott2005quantum,z20,aaronsonnew}. In particular, we make new observations about this measurement procedure, by aligning its properties with the features required by the classical LWE search-to-decision reduction by Regev \cite{Regev05}. 
 
We take a markedly distinct approach compared to the previous work \cite{bitansky2022constructive} on post-quantum reductions, even though both of our approaches set basis upon the \cite{marriott2005quantum,z20} measurement techniques. We believe this technique and the new observations made maybe of independent interest, such as  lifting other \cite{Regev05}-like reductions from classical to post-quantum.  



\subsection{Comparison with Related Works}
\label{sec:related_works}


\paragraph{Secure Key Leasing/Revocable PKE: Comparison with Existing Works }

Previously, secure circuit leasing for PKE was achieved in \cite{agrawal2023public}. \cite{agrawal2023public} achieved it by a clean compiler from any PKE scheme. But they need quantum communication and quantum computation on the lessor's side throughout the protocol, and there is no easy, generic way to turn their scheme into one with classical leaser with classical communication.

A similar result on revocable PKE, FHE and PRF to \cite{ananth2023revocable} was shown by the concurrent work from LWE with interesting ideas from quantum money combined with dual Regev PKE.  However, they need new conjectures for enabling the security proof  where adversaries only have inverse polynomial probability of revocation, otherwise they can only handle  a weaker security where the adversary has to pass the revocation with probability almost 1 \cite{ananth2023revocable}.
Whereas we achieve the stronger definition, i.e. adversaries only need to pass the revocation test with noticeable probability, from LWE alone. Moreover, they also need the lessor to prepare the quantum key and send the key via a quantum channel.

Our techniques also differ greatly from the above works. Our main insight is to manipulate the primitives used in the classical verification of quantum computation protocol and turning it into an FHE scheme from LWE that supports leasing.

\paragraph{Secure Software Leasing:}
SSL, short for Secure Software Leasing in \cite{ananth2020secure} is another relaxation of quantum copy protection. It is a weaker security notion than the notion studied in this work, even though bearing a similar name. 
In SSL, after revocation/deletion, one tests whether the leftover state of the adversary function correctly,  using the leaser’s circuit. In other words, in SSL, the pirate can be arbitrarily malicious in copying the function, but the free-loaders of the pirated programs are semi-honest by following the authentic software’s instructions. In our work, we allow the adversarial state to be run  by any adversarial circuit; as long as it gives the correct output, we count it as a successful attack., i.e. both the pirate and free-loaders can have arbitrary malicious behaviors. 

We defer more related works and discussions on SSL to \Cref{sec:other_related_works}.

\paragraph{Post-quantum LWE Search-to-Decision Reduction with Quantum Auxiliary Input}
\cite{bitansky2022constructive} showed a nice generic method to lift classical reductions to post-quantum reductions for a class of primitives satisfying certain constraints. Even though the classical search-to-decision reduction for LWE fit into the framework required in the above lifting theorem, our proof roadmap requires a stronger statement 
than the one shown in \cite{bitansky2022constructive}.
To the best of our knowledge, there is no generic or black-box method to modify the statement in \cite{bitansky2022constructive} for our use. Besides, our techniques differ significantly from \cite{bitansky2022constructive} 
 (and \cite{chiesa2022post}, by which \cite{bitansky2022constructive} was inspired) in that we do not perform the "nested alternating projections" procedure used in \cite{bitansky2022constructive,chiesa2022post}, 
even though both of our approaches are based on applying the measurement procedures proposed in \cite{marriott2005quantum,z20}. The way we use \cite{marriott2005quantum,z20} aligns intuitively closer with the classical reduction in \cite{Regev05} and has a simpler proof overall. 
Therefore,  we believe our techniques may be of independent interest and can shed light on new observations in (post-)quantum reductions.


Another work \cite{sudo2023quantum} presents a quantum LWE search to decision reduction (with the possibility of amplification) but their quantum distinguisher is simply a unitary that does not have an auxiliary state. Therefore, it is incomparable to our techniques. Our reduction is more general and more challenging to achieve due to the potentially destructible auxiliary quantum state.

\section{Technical Overview}
\newcommand{\pkeskl}{\mathsf{PKESKL}}
\newcommand{\mat}[1]{\mathbf{#1}}

We start by recalling the definition of Public-Key Encryption with Secure Key Leasing (also referred to as $\pkeskl$ from here on). (Levelled) homomorphic encryption with secure key leasing is defined analogously.
\jiahui{add setup and classical commu}
In a $\pkeskl$ scheme with classical leaser, we have algorithms  ($\Setup$, $\KeyGen$, $\Enc$, $\Dec$, $\Del$, $\VerDel$). The $\Setup$
 algorithm (to be run by the lessor) takes in  a security parameter and outputs a classical master public-key $\mpk$ and a classical trapdoor $\td$. The $\KeyGen$ algorithm (to be run by the lessee) takes input $\mpk$ and outputs a classical public-key $\pk$ and a quantum secret key $\rho_{\sk}$. $\Enc$ is defined in a usual way, it is a classical procedure that takes as input the key $\pk$ and a bit $\mu$ and produces a ciphertext $\ct$. The quantum procedure $\Dec$ takes as input a ciphertext $\ct$ and the state $\rho_{\sk}$ and should produce the message that was encrypted. Moreover, the decryption should keep the state $\rho_{\sk}$ statistically close to the initial state as long as the decrypted ciphertext were not malformed.
Then, the deletion algorithm $\delete$ should take in the state $\rho_{\sk}$ and produce a classical certificate $\cert$. Finally, the certificate can be verified by the classical algorithm $\VerDel$ using $\pk$, and the trapdoor $\td$ as other inputs. Conclusively, our scheme will allow the generation of key $\rho_{\sk}$ by a protocol involving merely classical effort on the part of a client (lessor) and quantum effort from the server (lessee).

Our security requirement is also fairly intuitive. In the security game, the challenger generates a classical $\mpk$ along with a trapdoor $\td$ and gives $\mpk$ to  a quantum polynomial time attacker $\A$. $\A$ generates a quantum decryption key $\rho_{\sk}$ and its associated public key $\pk$. $\A$ publishes its $\pk$.
The challenger then asks $\A$ to provide a certificate $\cert$. The game aborts if $\cert$ does not verify. However, if the certificate verifies, we require that $\A$ should not be able to distinguish an encryption of $0$ from an encryption of $1$ with a non-negligible probability. We formalize this by requiring that after the step of successful verification of the certificate, an adversary outputs a quantum state $\rho_{\mathsf{Delete}}$ as a adversarial quantum decyrptor. 

Checking a quantum state 's success probability can be difficult due to its destructible and unclonable nature.
To check if this state is a good decryptor, we define a special binary-outcome projective measurement, called Threshold Implementation, $\ti_{\frac{1}{2}+\gamma}$ (from \cite{z20,aaronsonnew}). This measurement projects the state $\rho_{\mathsf{Delete}}$ onto the subspace of all possible states that are good at distinguishing ciphertexts with probability at least $\frac{1}{2}+\gamma$. If $\ti_{\frac{1}{2}+\gamma}$ outputs $1$ on $\rho_\delete$ with some noticeable probability,
it implies that a noticeable "fraction" of the quantum states can distinguish ciphertexts with probability  at least $1/2+\gamma$.

Passing the $\ti_{1/2+\gamma}$ test combined with handing in a valid certificate, we consider such an adversary to be successful.

The definitions are inspired from the framework introduced by Zhandry \cite{z20}(see Section \ref{sec:strong_security},\Cref{sec:unclonable dec ati} for details) and will imply the existing PKE-SKL definition used in \cite{agrawal2023public,ananth2023revocable}. As we will discuss later,  defining the security via Threshold Implementation also assists us in our security proof. 


We then start by describing the ideas to build a PKE scheme first. These ideas can be lifted to build a (levelled) FHE scheme. Below, we show a series of ideas leading up to our final scheme.

\paragraph{Starting Point: Regev PKE.} Our starting point is a public-key encryption due to Regev \cite{Regev05}, which is also the basis of many FHE schemes such as the BV scheme \cite{FOCS:BraVai11}, and the GSW scheme \cite{C:GenSahWat13}. In Regev PKE, the public key consists of a matrix $\mat{A}\in \Z^{n\times m}_{q}$ where $m=\Omega(n\cdot \log q)$, along with $m$ LWE samples $\mat b=\mat s \mat{A}+\mat{e} \mod q$ where $\mat s \leftarrow \Z^{1\times n}_{q}$ is a randomly chosen and $\mat e$ is a discrete Gaussian distributed error with some parameter $\sigma \ll q$. The classical secret key is simply the secret vector $\mat{s}$ for the LWE sample.

To encrypt a message $\mu \in \{0,1\}$, one samples a binary random vector $\mat r \in \Z^{m\times 1}_{q}$. Encryptor then computes $\ct=(\ct_1,\ct_2)$ where $\ct_1=(\mat A \cdot \mat r$, $\ct_2=\mat{b}\cdot \mat{r}+\mu \cdot \lceil \frac{q}{2}\rceil)$. Observe that the correctness follows from the fact that $\mat{b}\cdot \mat{r}\approx\mat{s}\cdot \mat{A}\cdot \mat{r}$, where $\approx$ means that they are close to each other in the $\ell_2$ norm, since $\mat{r}$ and $\mat{e}$ are both low-norm. The security follows from the fact that due to LWE, $\mat{b}$ is pseudorandom. Therefore, one can apply leftover hash lemma argument to that $\mu$ is hidden. We want to turn this into a scheme that supports leasing. Namely, instead of having the secret key $\mat{s}$, we would like to have it encoded as a quantum state that somehow encodes $\mat{s}$ and (say) another hidden value (for the deletion certificate) so that it's computationally hard for an attacker to produce both? Unfortunately, it's not so clear how to make things work if the scheme has exactly one key $\mat{s}$.

\jiahui{give explanation on why having the binary decomposition function....Seem a bit arbitrary here to have the fucntion J}
\paragraph{Two Key Regev PKE.} Our next idea is to modify the Regev style PKE scheme to work with two secret keys $\mat{s}_0, \mat{s}_1$ instead. The public-key would now consist of a matrix $\mat{A}$ along with two LWE samples $\mat b_0=\mat s_0\mat{A}+\mat{e}_0$ and $\mat b_1=\mat s_1\mat{A}+\mat{e}_1$. To generate an encryption of $\mu \in \{0,1\}$ one outputs $\ct=(\ct',\ct_0,\ct_1)$ where $\ct'=\mat{A}\mat{r}$, $\ct_i= \mat{b}_i\mat{r}+\mu \cdot \lceil \frac{q}{2}\rceil$ for $i\in \{0,1\}$. Observe that now the ciphertext $\ct$ could be decrypted using either of the two secret vectors $\mat s_0$ and $\mat s_1$. This gives us the following idea: perhaps one could encode the secret key as a superposition of the two keys, $\rho_{\sk}=\ket{0,\mat{s}_0}+\ket{1,\mat{\mat{s}_1}}$ (ignoring normalization). Observe that using such a state one could decrypt any honestly generated ciphertext without disturbing the state. 

What is a good candidate for a deletion certificate? Naturally, in order to delete the secret information in the computational basis, a deletion certificate could be a measurement in the Hadamard basis (considering them encoded in binary), giving a string $\bd$ and a value $\alpha = \langle \bd, \bs_0\oplus \bs_1 \rangle$. 
 We could perhaps hope that if an algorithm supplies the measurement outcome in the Hadamard basis correctly with overwhelming probability, then it would destroy information on the two secret keys. 

How would one analyze such a scheme? The hope would be to arrive at a contradiction from such an adversary. Assuming that such an adversary outputs $(\alpha, \mat d)$ with overwhelming probability and conditioned on that distinguishes encryptions of $0$ from $1$ with inverse polynomial probability, we would hope to extract one of the secrets $\mat s_0$ or $\mat s_1$ with inverse polynomial probability, from the adversary's internal quantum state. Simultaneously, we also have obtained the string $(\alpha, \mat d)$ from the certificate handed in. We would like to argue that no efficient adversary is able to produce both a secret vector and a string $(\alpha, \mat d)$. While this is often a reasonable property, subtleties can come in when combined with other structures of our encryption scheme.
We instead turn to a very similar setting that has been extensively analyzed in prior works in context of classical proof of quantumness and certified randomness \cite{brakerski2021cryptographic}. This will save us the effort of reinventing the wheel.
\jiahui{too early to mention search to decision, just say extract}

\paragraph{Inspiration from Classical Verification of Quantum Computation.} Let us now recall the very similar setting in \cite{brakerski2021cryptographic} which constructed proofs of quantumness from LWE. They build \emph{Noisy Claw-free Trapdoor families (NTCF)} from LWE and show that a related task above is hard. Without going into the unnecessary semantics, we describe the setting relevant to us. 

Consider the following game: The challenger samples $\mat{A}$ and an LWE sample $\mat{k}=\mat{s}\mat{A}+\mat{e} \mod q$ where $\mat{s}$ is (say) binary random secret $\mat{s}$ and $\mat{e}$ is chosen from a discrete Gaussian distribution with parameter $\sigma$. For any such sample, there exists a BQP algorithm that could sample a vector $\mat y=\mat{x}_0\mat{A}+\mat{e}'$ where  $\mat{e}'$ chosen from discrete Gaussian with parameter $\sigma' \gg \sigma$ (by a superpolynomial factor), along with a special state $\qsk$. The state is of the form $\qsk =\ket{0,\mat{x}_0}+\ket{1,\mat{x}_1}$ where $\mat{x}_1=\mat{s}-\mat{x}_0$. Measuring this state yields either $\mat{x}_0$ so that $\mat{x}_0\mat{A}\approx \mat{y}$ or $\mat{x}_1$ so that $\mat{x}_1\mat{A}$ is close to $\mat{y}-\mat{k}$. On the other hand, if we measured in the Hadamard basis we will again obtain $(\mat{d},\alpha)$ for a random string $\mat{d}$ so that $\alpha=\langle \mat{d}, \mat{x}_0\oplus \mat{x}_1\rangle$ \footnote{In the actual protocol, $\alpha$ is not equal to $\langle \bd, \bs \rangle$ where $\bs = \bx_0 \oplus \bx_1$. In fact, we first apply a binary decomposition to $\bx_0, \bx_1$ and then measure in Hadamard basis. i.e. we have $\alpha=\langle \mat{d}, \mathcal{J}(\mat{x}_0)\oplus \mathcal{J}(\mat{x}_1)\rangle \mod 2$, where $\mathcal{J}$ is the invertible binary decomposition function. Such a binary decomposition is needed for the security proof of the NTCF to go through in \cite{brakerski2021cryptographic}}. 

 The test of quantumness asks an efficient quantum algorithm to produce a valid LWE sample $\mat{y}=\mat{x}_0\mat{A}+\mat{e'}$ of the kind above and based on a random challenge produce either one of $(\mat{x}_0,\mat{x}_1)$ for $b\in \{0,1\}$ or a ``valid"  tuple $(\mat{d},\alpha)$ as above.

\jiahui{I don't think much discussion on proof of quantumness. Just say these properties they show along the way also hold for quantum provers. Revised already}

This above protocol has a quantum correctness and a classical soundness, which we omit discussions due to irrelevance. 
However, to enable 
capabilities such as certified randomness/verification of quantum computation\cite{mahadev2018classical},
they show a more advanced soundness property against quantum provers.
That is, even a BQP adversary cannot \emph{simultaneously} produce these responses  (the "computational response" for $b = 0$,  and the "Hadamard response" for $b = 1$), assuming quantum security of LWE. 

In short, we conclude the properties we need from the  NTCF-based protocol: any QPT algorithm, given $\mat{A}$ and $\mat{k}=\mat{s}\mat{A}+\mat{e}$, can produce a valid LWE sample $\mat{x}_0\mat{A}+\mat{e}' \approx \mat{x}_0\mat{A}+\mat{e}' + \mat{k} =\mat{y}$
. If it is asked to output \emph{either}, (1) one of $\mat{x}_0,\mat{x}_1$, \emph{or}, (2) a non-trivial pair $(\mat{d},\alpha)$ such that $\alpha=\langle \mat{d}, \mathcal{J}(\mat{x}_0)\oplus \mathcal{J}(\mat{x}_1)\rangle$, it can do so with probability 1. However, if it is asked to output \emph{both} (1) and (2) at once for the same $\by$, it cannot achieve so with probability  noticeably higher than trivial guessing for one of them.

\paragraph{From Noisy Trapdoor Claw-free Families to Our Scheme.} Now we gradually approach a scheme similar to the two-key Regev PKE scheme which can benefit from the NTCF security properties described above. The idea is that we have a public key $\mat{A}$, an LWE sample $\mat{k}=\mat{s}\mat{A}+\mat{e}$ where $\mat{s}$ is a random binary secret and $\mat{e}$ is Gaussian distributed exactly like the distribution used for proof of quantumness. Along with $\mat{k}$, the public key also consists of an LWE sample $\mat{y}=\mat{x}_0\mat{A}+\mat{e}'$ chosen as per the distribution in the proof of quantumness again. The client maintains a trapdoor $\td$: $\td$ consists of a trapdoor matrix $\mat{T}$ of $\mat{A}$ that can be generated at the same time as sampling $\mat{A}$ (recall that $\td$ will be used to verify the deletion certificate later). The leased decryption state $\rho_{\sk}$ is the superposition $\ket{0,x_0}+\ket{1, x_1}$ where $\mat{x}_1=\mat{x}_0-\mat{s}$.

To encrypt a bit $\mu$, one samples a random binary string $\mat{r}$ and computes a ciphertext $\ct=(\ct_1,\ct_2,\ct_3)$ where $\ct_1=\mat{A}\mat{r}$, $\ct_2=\mat{k}\mat{r}$ and $\ct_3=\mat{y}\mat{r}+\mu\lceil\frac{q}{2} \rceil$. Observe that the ciphertext can be decrypted by both $\mat{x}_0$ or $\mat{x}_1$. This is because $\ct_3-\mat{x}_0\mat{A}$ is close to $\mu \lceil \frac{q}{2}\rceil$ and similarly, $-\ct_2+\ct_3+\mat{x}_1\mat{A}$ is also close to $\mu \lceil \frac{q}{2}\rceil$. Thus, there is a way to decrypt a ciphertext coherently without disturbing the state $\rho_{\sk}$, by the gentle measurement lemma \cite{aaronson2004limitations}. 
Accordingly, we set the deletion certificate to be the string $(\alpha,\mat{d})$ to use the property of NTCF.

\paragraph{Security: First Attempt}
 We first discuss a very weak definition of security which the above simple scheme (almost) satisfies already. In this definition, a successful BQP adversary provides a deletion certificate that's valid with probability $1-\negl$ for some negligible $\negl$ and then conditioned on that distinguishes ciphertexts with a noticeable probability. One observation is that given a decryptor that's capable of distinguishing encryptions of zero from encryptions of one, must also enable distinguishing ciphertexts of zero of the form
$(\ct_1,\ct_2,\ct_3)$ where $\ct_1=\mat{A}\cdot \mat{r}$, $\ct_2=\mat{k}\cdot \mat{r}$ and $\ct_3=\mat{y}\cdot \mat{r}$ from truly random strings with a noticeable probability. This means that this distinguisher distinguishes samples of the form $(\mat{A}\mat{r}, \mat{k}\mat{r}, \mat{y}\mat{r})$ from random. Since $\mat{k}$ is pseudorandom due to LWE it must also distinguish $(\mat{A}\mat{r}, \mat{k}\mat{r}, \mat{y}\mat{r})$ from random where $\mat{k}$ is chosen at random. Observe that $\mat{y}=\mat{x}_0\mat{A}+\mat{e}'$, $\mat{y}\mat{r}=\mat{x}_0\mat{A}\mat{r}+\mat{e}'\mat{r}$. Given that $\mat{e}'\mat{r}$ losses information on $\mat{r}$, and $\mat{A},\mat{k}$ are now chosen at random, one can appeal to LHL to show that such a distinguisher should distinguish between $(\ct_1,\ct_2,\ct_3)$ generated as $(\mat{a},u,\langle \mat{x}_0, \mat{a} \rangle+\mat{e}'\mat{r})$ from random, where $\mat{a}\in \Z^{n\times 1}_{q}$ and $u\in \Z_q$ are random. 

Together, $(\ct_1,\ct_3)$ are now almost distributed as an LWE sample in $\mat{x}_0$. If we have a quantum search to decision reduction for LWE using the quantum distinguisher (with an internal state), we should be able to extract $\mat{x}_0$ efficiently, giving rise to a contradiction. That is, this reduction would have produced a valid $(\alpha,\mat{d})$ with overwhelming probability and conditioned on that $\mat{x}_0$ with inverse polynomial probability (something ruled out by the NTCF in \cite{brakerski2021cryptographic}).

In this work, we construct a post-quantum (quantum) reduction that runs in time polynomial in $(B,n,\frac{1}{\epsilon}, \log q, T_{\A})$ where $\epsilon$ is the distinguishing advantage of an adversary that distinguishes samples LWE in a fixed secret from random samples and recovers that secret with probability polynomial in $\epsilon$\footnote{While we can also obtain such a reduction directly from plugging in Theorem 7.1 in \cite{bitansky2022constructive}, as we will discuss later, their reduction's success probability is not sufficient in our proof of the standard security. }. Above $B$ is the bound on the coordinates of $\mat{x}_0$ which we will set to be slightly superpolynomial and $T_{\A}$ is the running time of the adversary. 

There is a minor flaw in the above argument that can be fixed. The distribution $(\ct_1,\ct_2,\ct_3)$ formed as  $(\mat{A}\mat{r}, \mat{k}\mat{r}, \mat{y}\mat{r})$ does not behave statistically close to $(\mat{a},u,\langle \mat{a},\mat{x}_0\rangle+e)$ for truly random $\mat{a},u$ and an error $e$ sampled according to LWE distribution. The reason for that is that if we consider $\mat{y}\mat{r}=\mat{x}_0\mat{A}\mat{r}+\mat{e}'\mat{r}$,  the error $e=\mat{e}'\mat{r}$ might not behave as an statistically independent discrete Gaussian. To fix this issue, we modify our encryption algorithm to have a smudging noise $\mat{e}''$ with superpolynomially larger parameter $\sigma''\gg \sigma$ and construct $(\ct_1,\ct_2,\ct_3)$ as $(\mat{A}\mat{r},\mat{k}\mat{r},\mat{y}\mat{r}+\mat{e}''+\mu\lceil \frac{q}{2} \rceil)$. With this smudging $\mat{e}'\mat{r}$ can now be drowned by $\mat{e}''$ effectively now behaving as a fresh LWE sample in $\mat{x}_0$ with slightly larger noise. 
Accordingly, to make the final security proof go through, we design our search-to-decision reduction to also work for samples statistically-close to LWE versus samples close to random, as opposed to the simple case of real LWE versus truly random. 


\paragraph{Stronger Security from Parallel Repetition}
While the scheme above gives rise to a PKE scheme satisfying a weak but a non-trivial security guarantee, such a definition is not really enough. It might be possible that an adversary simply  guesses the correct certificate $(\alpha,\mat d)$ by picking $\mat{d}$ randomly and choosing $\alpha$ at random as well. With $\frac{1}{2}$ probability, this adversary might be right in producing a valid certificate. Since its state is still preserved, it could continue to successfully decrypt all ciphertexts (moreover, it could simply measure the state in the standard basis and keep a classical key). We'd like to achieve a stronger definition where even if the  certificate passes with a noticeable probability, upon successful passing, the advantage should not be non-negligible.

Our main idea there is that one could create $\lambda$ independent instances of the scheme where one encrypts secret sharing of the bit $\mu$ where $\lambda$ is the security parameter. The idea is that now the public key would consist of $\lambda$ independent matrices $\mat{A}_i$, $\lambda$ independent binary secret LWE samples $\mat{k}_i=\mat{s}_i\mat{A}_i+\mat{e}_i$ along with $\mat{y}_i = \mat{x}_{i,0}\mat{A}_i+\mat{e}'_i$ for $i\in [\lambda]$.  The deletion certificate would consist of $\lambda$ vectors $(\alpha_i, \mat d_i)$ and each of them are independently verified. The hope again is that if one is able to distinguish encryptions of $0$ from encryptions of $1$, then such an adversary should essentially be able to distinguish for each component $i\in [\lambda]$. We could use this to extract all $\{\mat{x}_{i,b_i}\}_{i\in [\lambda]}$($b_i = 0$ or 1). Then, we can appeal to the results from the parallel repetition variant of the game used for proofs of quantumness. This variant has been studied in \cite{radian2019semi}. The soundness property ensures that it is computationally hard to come up with noticable probability valid responses $\{(\alpha_i, \mat d_i)\}_{i\in [\lambda]}$ and $\{\mat{x}_{i,b_i}\}_{i\in [\lambda]}$($b_i = 0$ or 1) for all indices.


\paragraph{Lifting the Scheme to Support Homomorphism} The above idea could work however, we would have to examine a lot of care to extract the secrets $\{\mat{x}_{b_i,i}\}_{i\in [\lambda]}$ as a quantum state can evolve over time. At this point, we move to directly construct a (levelled) FHE scheme. The structural properties of our levelled FHE scheme would solve two problems in one shot. It will not only yield as an FHE scheme, but will also simplify our analysis. 

To lift to FHE, we take inspiration for the GSW scheme \cite{C:GenSahWat13}. In the GSW scheme, the public key consists of a pseudorandom LWE matrix $\mat{B}\in \Z^{N\times M}_{q}$ where $M=\Omega(N\log q)$ such that there exists one secret vector $\mat{s}\in\Z^{1\times N}_{q}$ so that $\mat{s}\mat{B}$ is small norm. Such matrices can be generated by sampling first $N-1$ rows at random and the last row generated as an LWE sample with the coefficient being the first $N-1$ rows. The encryption of a bit $\mu$ is of the form $\mat{B}\mat{R}+\mu \mat{G}$ where $\mat{R}$ is a small norm random binary matrix and $\mat{G}$ is a special gadget matrix \cite{EC:MicPei12}. The consequence of this is that $\mat{B}\mat{R}$ behaves essentially like a random LWE sample with the same secret vector as $\mat{B}$ and one could argue security by appealing to LHL. We omit a description of why the ciphertext could be homomorphically computed on (one could refer to either \cite{C:GenSahWat13} or our technical sections).

To lift to such an FHE scheme, we work with a specially chosen $\mat{B}$. We choose it as:

 \begin{align*}
        &\mat{B}= \begin{bmatrix} \mat{k}_1=  \bs_1\bA_1 +\be_1  \\ \bA_1 \\ \cdots
        \\ \cdots \\
       \mat{k}_{\lambda}=\bs_\lambda\bA_\lambda +\be_\lambda \\  \bA_\lambda 
        \\\sum_{i \in [\lambda]} \by_i = \sum_{i\in [\lambda]} \mat{x}_{i,0}\mat{A}_i+\mat{e}'_i\end{bmatrix} 
    \end{align*}

Observe that there are many vectors $\mat{v}$ so that $\mat{v}\mat{B}$ is small norm. This is crucial because we would like the ciphertexts to be decryptable using a secret key $\rho_{\sk}=\otimes_i^\lambda \rho_{\sk,i}$ where $\rho_{\sk,i}=\ket{(0,\mat{x}_{i,0})}+\ket{(1,\mat{x}_{1,i})}$. The idea is that $\mat{x}_{0,i}\mat{A}_i$ is close to $\mat{y}_i$ and similarly $\mat{k}_i-\mat{x}_{1,i}\mat{A}_i$ is also close to $\mat{y}_i$. Thus, $\mat{B}$ can be decrypted by viewing $\qsk$ as a vector we can perform a GSW-like decryption in superposition, with a gentle measurement on the final output.

For security, we consider the structure of the ciphertext $\ct=\mat{B}\mat{R}+\mu\mat{G}$, the term $\mat{B}\mat{R}$ is of the form:

 \begin{align*}
        &\mat{B}\mat{R}= \begin{bmatrix}  \bs_1 \bA_1\mat{R} +\be_1\mat{R} \\ \bA_1\mat{R} \\ \cdots
        \\ \cdots \\
       \bs_\lambda\bA_\lambda \mat{R} +\be_\lambda\mat{R} \\  \bA_\lambda \mat{R}
        \\\sum_{i \in [\lambda]} \by_i\mat{R} = \sum_{i\in [\lambda]} \mat{x}_{i, 0}\mat{A}_i \mat{R}+\mat{e}'_i\mat{R}\end{bmatrix} 
    \end{align*}

Our intuition to extract $\mat{x}_{i,0}$(w.l.o.g. or $\mat{x}_{i,1}$) for all $i\in [\lambda]$ is as follows. First we observe that it will suffice to devise an extractor that will succeed with  high probability in the case when $\mat{k}_i$ is chosen to be random as opposed to  be pseudorandom due to LWE. If we are able to do that,  we should be able to extract $\mat{x}_{i,b_i}$ with similar probability in the world where the $\mat{k}_i$'s are pseudorandom due to LWE security. 

To realize our relaxed goal,  we observe that the last row $\sum_{i\in [\lambda]} \mat{y}_i\mat{R}$ is close to a linear equation of the form: 
\begin{align*}
\sum_{i\in [\lambda]} \mat{y}_i\mat{R}  \approx \sum_{i\in [\lambda]} \mat{x}_{i, 0} \mat{V}_i 
\end{align*}
where $\mat{V}_i=\mat{A}_i\mat{R}$. 
Thus, we modify our encryption algorithm to have smudging noise. 

After we modify the encryption algorithm to compute $\mat{B}\mat{R}+\mat{E}+\mu \mat{G}$ where $\mat{E}$ is zero everywhere else except the last row containing discrete Gaussian with parameter $\sigma''$ superpolynomially more than the norm of $\mat {e}'_i\mat{R}$, we make a nice observation: the last row of $\mat{B}\mat{R}+\mat{E}$ would now be distributed as an LWE sample:  in terms of the secret $(\mat{x}_{1, b_1},\mat{x}_{2, b_2},\ldots,\mat{x}_{\lambda, b_\lambda})$($b_i = 0 \text{ or } 1$) with the coefficient vector $[\mat{V}^{\top}_1,\ldots,\mat{V}^{\top}_{\lambda}]^{\top}$. 
Now we could appeal the LHL to replace $\{(\mat{A}_i\mat{R},\mat{V}_i\mat{R},\mat{k}_i\mat{R})\}_{i\in [\lambda]}$ to completely random and the last row by a fresh sample with error according to discrete Gaussian in parameter $\sigma''$, with the long secret 
$(\mat{x}_{1, b_1},\mat{x}_{2, b_2},\ldots,\mat{x}_{\lambda, b_\lambda})$ and random coefficient vector $[\mat{V}^{\top}_1,\ldots,\mat{V}^{\top}_{\lambda}]^{\top}$. If an adversary now distinguish such ciphertexts from random, we should be able to extract the entire long secret vector in one shot $(\mat{x}_{1, b_1},\mat{x}_{2, b_2},\ldots,\mat{x}_{\lambda, b_\lambda})$ using our proposed quantum search-to-decision reduction. 

\paragraph{Completely Classical Communication and Classical Lessor}
Our protocol comes with classical lessor and classical communication for free.
We observe that from the property of the underlying NTCF-based protocol, the lessor only has to run $\Setup$ and sends the classical $\mpk = \{\bA, \bs\bA+\be\}$(naturally extended to parallel-repeated case) to the
lessee. The lessee can prepare its own quantum key given $\mpk$ 
by preparing a superposition of $\sum_{b,\bx, \be'}\ket{b, \bx_b } \ket{\by = \bx_b \bA + \be' + b \cdot \bs\bA }$ by sampling $\be'$ on its own. It then measures the $\by$-register to obtain a public key $\by$ and a quantum decryption key of the form $\ket{0,\bx_0} 
+\ket{1,\bx_1}$. 
Working with the properties of NTCF (shown in \cite{brakerski2021cryptographic}), the security of our scheme will be guaranteed even for maliciously generated $\by$.


\subsection{Detailed Overview on Security Proof: Reduction to NTCF and the Use of Search-to-Decision Reduction for LWE}
\label{sec:searchdecisiontecover}

\paragraph{Threshold Implementation}
Before going into more technical details, we briefly describe the properties of the 
measurement we perform on the adversarial decryptor state to test if it succeeds on distinguishing ciphertexts with high enough probability. We will leverage the properties of this measurement in our security proofs.

Threshold Implementation (and its related efficient measurement implementations in \Cref{sec:unclonable dec ati}) is a powerful technique by Zhandry (which is further inspired by Mariott-Watrous's work on QMA amplification \cite{marriott2005quantum}) used in a number of recent works \cite{coladangelo2021hidden,aaronsonnew,chiesa2022post,liu2022collusion,ananth2023revocable}. 

The Threshold Implementation $\ti_{\gamma+1/2}$ has the following properties and relations to our security:
\begin{enumerate}
\item   We will call $\rho$ a good decryptor if we $\ti_{\gamma+1/2}$ applied on $\rho$ outputs 1. 
  
    \item For a successful adversary in our game, it must produce a  decryptor $\rho$ that is good  with probability $p$ for some noticeable $p$ (apart from giving a valid certificate with noticeable probability).

    \item For the remaining state $\rho'$ after performing the above $\ti_{\gamma+1/2}$, given the outcome being 1, applying the same $\ti_{\gamma+1/2}$ on $\rho'$ will yield outcome 1 with probability 1.
\end{enumerate}

The above statement basically says that the measurement $\ti_{\gamma+1/2}$ is projective 
 and it "preserves" the advantage of the quantum distinguisher's state. 




\paragraph{Proof Outline}
We now go into more details about our high level proof strategy. Recall that in the security game, a successful attacker would first need to output a valid deletion certificate (which we denote by the event $\mathsf{CertPass}$) along with that it must produce a state $\rho_{\mathsf{delete}}$ for which our test of good decryptor $\mathsf{TI}_{\frac{1}{2}+\gamma}$ for some noticeable $\gamma$ passes with inverse-polynomial probability. Namely, $\mathsf{TI}_{\frac{1}{2}+\gamma}(\rho_{\mathsf{delete}})=1$ with inverse polynomial probability. We call this event $\mathsf{GoodDecryptor}$. It must hold that $\Pr[\mathsf{GoodDecryptor} \wedge \mathsf{CertPass}]$ is noticeable for such a successful attacker. Since we are guaranteed that there is a noticeable chance of overlap of the two events, our hope is that our extraction procedure would extract $\{\mathbf{x}_{0,i}\}$ just with enough probability to induce a noticeable probability overlap between $\mathsf{CertPass}$ and a successful extraction causing our reduction to win the parallel repetition game of Radian et al. We call the event when the extraction holds as $\mathsf{ExtractionOccurs}$.

At this point it is tempting to think of a search-to-decision classical reduction for LWE and compile it to a quantum-reduction via known methods such as the one by Bitansky et al. \cite{bitansky2022constructive}. Unfortunately, we don't have a citable theorem from this work to use  directly\footnote{As discussed in \Cref{sec:related_works}, in fact all existing works on post-quantum
LWE search-to-decision reductions are not directly applicable to our setting.}.
We need precise bounds on the probability of success, and the running time for the extraction. 

Let us briefly explain the story: for the extraction to occur, we would need to move to a world where $\mathbf{k}_i$'s are switched with random. While this won't change the probability of $\mathsf{GoodDecryptor}$ by a non-negligible amount due to LWE security, in this hybrid deletion certificates will no longer exist. We could infer from there using reductions compiled using  \cite{bitansky2022constructive}) that $\mathsf{ExtractionOccurs}$ with probability $\frac{1}{\poly(\lambda)}$ for some polynomial $\poly$. This probability of extraction will still be the same (up to a negligible loss) if we went back to the world where $\{\mat{k}_i\}$ are again from the LWE distribution. However, we can't infer from this that the event of $\mathsf{ExtractionOccurs}$ overlaps with $\mathsf{CertPass}$. A similar issue was encountered by Ananth et. al. \cite{ananth2023revocable}. To address this issue the \cite{ananth2023revocable} had to rely on new conjectures.

To address this issue, we devise a high probability search-to-decision reduction that extracts (in the world where $\mathbf{k}_i$'s are random) with probability $1-\negl(\lambda)$ whenever $\mathsf{GoodDecryptor}$ holds. Thus, when we switch $\mathbf{k}_i$'s with LWE, due to LWE security, $\mathsf{ExtractionOccurs}$ also succeeds with all but negligible security whenever $\mathsf{GoodDecryptor}$ holds. Since we are guaranteed that there is a noticeable overlap between $\mathsf{GoodDecryptor}$ and 
$\mathsf{CertPass}$, this implies a noticeable overlap between $\mathsf{CertPass}$ and $\mathsf{ExtractionOccurs}$.

\paragraph{Ideas from Classical Search-to-Decision Reduction}
Our quantum search-to-decision reduction is inspired by Regev's search to decision reduction \cite{STOC:Regev05}. We now describe the setup of our reduction. We consider a simpler setup than in our technical section (which is a bit specific to our setting).

The adversary gets as input an LWE sample of the form $(\mat{A},\mat{x}\mat{A}+\mat{e} \mod p)$ where $\mat{x}$ is arbitrary secret with bounded norm $B$, $\mat{A}\leftarrow \Z^{n\times m}_{p}$ for large enough $m=\Omega(n\log p)$, $\mat{e}$ is Gaussian distributed with parameter $\sigma$. The reduction has a quantum state $\rho$ that 
has an inverse polynomial weight on vectors that  enables distinguishing samples $(\mat{A}',\mat{x}\mat{A}'+\mat{e}' \mod p)$ for randomly chosen $\mat{A}'\in \Z^{n\times m}_{p}$ and error $\mat{e}'$ sampled according to discrete Gaussian with parameter $\sigma'$ superpolynomially larger than $\sigma$ from truly random with probability $\frac{1}{2}+\gamma$.

We first describe the classical intuition and then describe how to lift that intuition to the quantum setting. Classically, we can consider recovering $\bx$ coordinate by coordinate. Say the first coordinate $x_1 \in [-B,B]$, we could choose a total of $2B$ guesses. For each guess $g$, we could consider the process of generating samples as:
\begin{itemize}
\item Sample a matrix $\mat{C}\in \Z^{n\times m}_{p}$ so that it is random subject to the bottom $n-1$ rows are $0$.
\item Sample $\mat{R}$ to be a random binary matrix $\{0,1\}^{m\times m}$. 
\item Set $\mat{A}'=\mat{A}\mat{R}+\mat{C}$ and $\mat{b}'=(\mat{x}\mat{A}+\mat{e})\mat{R}+ \mat v+\mat{e}'$ where $\mat{v}$ is the guess $g$ times the first (and the only non-zero) row of $\mat{C}$. Here $\mat{e}'$ is generated from the discrete Gaussian vector with parameter $\sigma'$.
\end{itemize}

If our guess was correct, the sample that we end up producing $(\mat{A}',\mat{b}')$ is distributed statistically close to the distribution for the distinguishing problem (LWE with secret $\mat{x}$). This is because $\mat{e}'+\mat{e}\mat{R}$ due to noise flooding is within $\poly(m)\cdot\frac{\sigma}{\sigma'}$ statistical distance. Similarly, $\mat{A}'$ due to LHL is within $p^{-n}$ statistical distance from uniform if $m$ is sufficiently large. This means that the distribution is within $\eta=\poly(m)\frac{\sigma}{\sigma'}$ statistical distance from the relevant distinguishing problem for a corect guess, i.e. the "LWE" case.

On the other hand, if the guess is incorrect, then one can observe that the distribution produces samples that are within even smaller statistical distance $\eta' =\poly(m)\cdot\frac{1}{p^n}$ from truly random distribution (no noise flooding required).

Thus, if the guess is correct, a classical adversary can distinguish the above distribution from random with probability at least $\frac{1}{2}+\gamma-O(\eta)$, and likewise if the guess is incorrect the maximum distinguishing probability is $\frac{1}{2}+O(\eta')$. We could therefore test the adversary by making roughly $\poly(\frac{1}{\gamma},\log \delta)$ calls to the distinguisher to guarantee the guess is correct with probability $1-\delta$. Setting $\delta = p^{-n}$, the reduction will extract $\mat{x}$ in time that's polynomial in $B,n,m,\log p,\frac{1}{\gamma}$ (bound on $\mat{x}$'s norm).

\paragraph{Our Quantum Search to Decision Reduction}

Moving to the quantum setting, we have to address a number of challenges. If we use our state to distinguish and LWE sample from random in the way above, the state could get destructed, preventing us from doing repetitions. In particular, the classical reduction needs to run the distinguisher many times over randomized inputs and "measure" its outcome to obtain useful information, which seems implausible when using a quantum distinguisher.


To overcome these issue, we will leverage the properties of Threshold Implementation.
We make a key observation that the procedure, where we repeatedly create samples according to
our guess $g$ and check on the distinguisher's output distribution to get an estimate on whether $g$ is correct, can happen "inside" the measurement procedure $\ti$.

Suppose we have efficient projective measurements  $\ti_{g, 1/2+ \gamma}=(\ti_{g,1/2+ \gamma},\mathbf{I}-\ti_{g, 1/2+\gamma})$ for various guesses $g$.
$\ti_{g,1/2+\gamma}$ projects the state $\rho$ onto vectors that distinguish the above distribution made from the guess $g$ above from truly random with probability at least $1/2 + \gamma $. 
For simplicity, we call it $\ti_{g}$. from now on.

We consider two other projections $\ti_{\lwe}$ and $\ti_{\unif}$: 
\begin{itemize}
   
    \item $\ti_\lwe$ projects onto distinguishers good at distinguishing the ciphertexts in our scheme from uniform random samples with probability at least $1/2 +\gamma$.  (Intuitively, these are distinguishers for "noisy" LWE instances versus uniform random)

    \item $\ti_\unif$ projects onto distinguishers good at distinguishing close-to-uniform random samples from uniform random samples with probability at least $1/2 +\gamma$. 
\end{itemize}




Our goal is to show that $\Pr[\mathsf{ExtractionOccurs}] \geq \Pr[\ti_{\lwe}(\rho) = 1] -\negl(\lambda)$. Thus it suffices to consider a world where $\ti_{\lwe} (\rho) = 1$ already happens and show that $\Pr[\mathsf{ExtractionOccurs}] \geq 1-\negl(\lambda)$ in this world. 
$\rho'$ is the  post measurement state. 

We now make the following observations, by recalling the properties of $\ti$:
\begin{itemize}
\item  Let $\rho'$ be the state we get post measurement of $\ti_{g_i, 1/2+\gamma} (\rho) \to 1$. 

\item 
We start working with the state $\rho'$ at the beginning of our extraction algorithm. 

\item Recall that by the projective property of $\ti$, given that outcome is $1$, we have $\mathsf{Pr}[\ti_{\lwe},\ (\rho')]=1$.

\item We start with the first entry of $\bx$ and pick the smallest possible value as our guess $g$ for this entry.

\item As we have discussed in the classical setting, when the guess $g$ is correct, we get to create samples statistically close to LWE samples; when the guess $g$ is incorrect, we get to create samples close to uniform random. Combining these with properties of $\ti_g$ for statistically close measurements, we can show that:

\begin{itemize}
    \item \textbf{When the guess is correct:} If we apply projection $\ti_g$ on it for a correct guess $g$, we will have $\Pr[\ti_g(\rho') = 1]$ overwhelmingly close to $1$.
In this setting, we are statistically close to measuring the original $\ti_\lwe$. That is, if a distinguisher can distinguish between (noisy) LWE versus real random. We will therefore get output $1$ with overwhelming probability as a consequence of $\ti$ being a projection. 

We can then assign $g$ to the entry of $\bx$ we are guessing for, and move on to the next entry.

\item \textbf{When the guess is incorrect:}  If we apply projection $\ti_g$ on it for an incorrect guess $g$, we will have $\Pr[(\mathbf{I}-\ti_g)(\rho')=1] = \Pr[\ti_g(\rho')=0] $ equal to $1-\negl(\lambda)$.
In this case, we are asking the distinguisher to distinguish between uniform vs. negligibly-close-uniform.
Clearly, no distinguisher can  distinguish with noticeable advantage $\gamma$ if the statistical distance is negligible. Therefore, $\ti_\unif$ will output 0 for all possible states. 

We then move on to perform $\ti_{g}$ with the next possible value of $g$.

\end{itemize}
\end{itemize}
A key observation is that, in both cases above, because of the measurement's projective property 
, the leftover state $\rho'$ will remain unchanged in terms of its usability: in the case where $\ti_{g}$ outputs 1, we know that if we apply $\ti_{g}$ for a correct $g$ again, we get the same outcome with probability $1$ due to the projective property; in the case of incorrect $g$'s, similarly, after we have the first $\ti_{g}$ that outputs 0, the remaining $\ti_g$ will always output 0 on any $\rho$ with probability 1. Therefore, the state will remain undisturbed.
Thus, we could find out $\mathbf{x}$ entry by entry while keeping our quantum state's performance almost unchanged.

\paragraph{Actual Algorithm: New Observations on the Approximate Threshold Implementation}

However, the $\ti$ used above is not an efficient measurement
and we need to make the above procedure efficient.
Fortunately, the work of Zhandry \cite{z20}, showed how to replace $\ti$ by an efficient, approximate projection, namely approximate threshold implementation($\ati$). 
 These approximate projections induce an error upon each measurement, and since they are not "exact projections", we cannot make the same argument as in the above $\ti$ analysis. We account for these issues by using the properties of these approximate projections.

 In the case when the approximate threshold implementation $\ati_g$ corresponds to a correct guess $g$, we use properties similar to those shown in \cite{z20} to demonstrate that 
the next $\ati_g$ with a correct guess will give the same outcome with probability close to 1,  similar to the clean $\ti$ setting.

The remaining question is: what happens when we apply $\ati_g$ with an incorrect $g$? Intriguingly, we make the following new observation (see details in \Cref{sec:invariant_ati_uniform}): the operation $\ati_g$ with an incorrect $g$ acts almost like an identity operator on \emph{any state} and hardly changes the state itself, in terms of \emph{trace distance}. When we use the \cite{z20,marriott2005quantum} $\ati$ measurement to test if a quantum adversary can distinguish between two identical distributions, naturally the measurement will output $0$ with probability close to 1 (because  no adversary can possibly distinguish). More importantly, the measurement acts exactly like an identity operator.
When we test a quantum adversary on distinguishing two close statistically-close distributions, the $\ati$  measurement acts close to an identity operator.
In a nutshell, each such $\ati_g$ with incorrect $g$  incurs only exponentially small deviation additively on the trace distance of the state. Therefore, after we have applied many $\ati_g$ with incorrect $g$'s,
our next $\ati_g$ with correct $g$ will yield outcome 1 with overwhelming probability as the previous $\ati_g$ with correct $g$ does because the state is almost unchanged since that time.



 In the end, we make sure that the measurement errors will not accumulate on distinguisher's state to a degree that affects the algorithm's performance.
 Conclusively, the dominating error is incurred by every measurement with $\ati_g$ with a correct $g$, which is some inverse superpolynomial. But this error will only accumulate additively for polynomially many times (the number of measurements with correct $g$ is the same as the dimension of the LWE secret). On the other hand, the errors incurred by $\ati_g$ with an incorrect $g$ are exponentially small. Therefore, the overall loss is negligible.
 The exact reduction used in our security proof works for any subexponential number of guesses, when working with subexponential-hardness of LWE.
See details in \Cref{sec:ful_extraction_analysis}.

 More generically, if we consider a clean setting where our input distributions are real LWE versus truly uniform samples, we even have less to worry about because all losses in the above measurements will be exponentially small.
Therefore, our search-to-decision reduction works for any LWE parameters (\Cref{sec:lwe_search_to_decision}) and runs in time roughly $\poly(\lambda, n, B)$ where $n$ is the dimension of the LWE secret and $B$ is the domain size of each entry of the secret.

\ifllncs


\section{Preliminaries: Testing Quantum Adversaries} \label{sec:unclonable dec ati}


\ifllncs
\paragraph{Additonal Preliminaries}
We refer more prelims on quantum information and computation to \Cref{appendix:quantum_info}, prelims on lattices in \Cref{sec:latticeprelims} and on Noisy claw-free trapdoor families
in \Cref{sec:NTCF_prelim}.
\else\fi

In this section, we include several definitions about measurements, which are relevant to testing whether quantum adversaries are successful in the security games of \Cref{sec:defs}. Part of this section is taken verbatim from \cite{aaronsonnew}. As this section only pertains directly to our security definitions for secure key leasing schemes, the reader can skip ahead, and return to this section when reading \Cref{sec:defs}. 


\vspace{1mm}

In classical cryptographic security games, the challenger typically gets some information from the adversary and checks if this information satisfies certain properties.
However, in a setting where the adversary is required to return \emph{quantum} information to the challenger, classical definitions of ``testing'' whether a quantum state returned by the adversary satisfies certain properties may result in various failures as discussed in \cite{z20}, as this state may be in a superposition of ``successful'' and ``unsuccessful'' adversaries, most likely unclonable and destructible by the classical way of "testing" its success. 

In short, we need a way to test the success probability of an adversarial quantum state analogous to what happens classically, where the test does not completely destruct the adversary's state. Instead, the state after the test has its success probability preserved in some sense.
Such a procedure allows us to do multiple measurements on the state without rendering it useless, and thus facilitates quantum reductions.

\revise{

\subsection{Projective Implementation and Threshold Implementation}
\label{sec:project_imp}
Motivated by the discussion above, \cite{z20} (inspired by \cite{marriott2005quantum})
formalizes a new measurement procedure for testing a state received by an adversary. We will be adopting this procedure when defining security of secure key leasing schemes in  \Cref{sec:defs}.
Consider the following procedure as a binary POVM $\cP$ acting on an alleged quantum-decryptor program $\rho$: sample a ciphertext $x$, evaluates the quantum decryptor on $x$, and checks if the output is correct.   
In a nutshell, the new procedure consists of applying an appropriate projective measurement which \emph{measures} the success probability of the tested state $\rho$ under $\cP$, and to output ``accept'' if the success probability is high enough. Of course, such measurement will not be able extract the exact success probability of $\rho$, as this is impossible from we have argued in the discussion above. Rather, the measurement will output a success probability from a finite set, such that the expected value of the output matches the true success probability of $\rho$. We will now describe this procedure in more detail.

The starting point is that a POVM specifies exactly the probability distribution over outcomes $\{0,1\}$ (``success'' or ``failure'') on any decryptor program, but it does not uniquely determine the post-measurement state. Zhandry shows that, for any binary POVM $\cP = (P, I-P)$, there exists a particularly nice implementation of $\cP$ which is projective, and such that the post-measurement state is an eigenvector of $P$. In particular, Zhandry observes that there exists a projective measurement $\cE$ which \emph{measures} the success probability of a state with respect to $\cP$. More precisely,
\begin{itemize}
    \item $\cE$ outputs a \emph{distribution} $D$ of the form $(p, 1-p)$ from a finite set of distribution over outcomes $\{0,1\}$. (we stress that $\cE$ actually outputs a distribution).
    \item The post-measurement state upon obtaining outcome $(p,1-p)$ is an \emph{eigenvector} (or a mixture of eigenvectors) of $P$ with eigenvalue $p$.
\end{itemize}

A measurement $\cE$ which satisfies these properties is the measurement in the common eigenbasis of $P$ and $I-P$ (such common eigenbasis exists since $P$ and $I-P$ commute). 

Note that since $\cE$ is projective, we are guaranteed that applying the same measurement twice will yield the same outcome. Thus, what we obtain from applying $\cE$ is a state with a ``well-defined'' success probability with respect to $\cP$: we know exactly how good the leftover program is with respect to the initial testing procedure $\cP$.

Formally, to complete the implementation of $\cP$, after having applied $\cE$, one outputs the bit $1$ with probability $p$, and the bit $0$ with probability $1-p$. This is summarized in the following definition.

\jiahui{add more prelims about eigenvalues and eigenvectors}


\begin{definition}[Projective Implementation of a POVM]
\label{def:project_implement}
    Let $\cP = (P, \bI -P)$ be a binary outcome POVM. Let $\cD$ be a finite set of distributions $(p, 1-p)$ over outcomes $\{0, 1\}$. Let $\cE = \{E_p\}_{(p, 1-p) \in \cD}$ be a projective measurement with index set $\cD$. Consider the following measurement procedure: 
    \begin{itemize}
        \item[(i)] Apply the projective measurement $\cE$ and obtain as outcome a distribution $(p, 1-p)$ over $\{0, 1\}$;
        \item[(ii)] Output a bit according to this distribution, i.e. output $1$ w.p $p$ and output $0$ w.p $1-p$. 
    \end{itemize}
    We say the above measurement procedure is a projective implementation of $\cP$, which we denote by $\projimp(\cP)$, if it is equivalent to $\cP$ (i.e. it produces the same probability distribution over outcomes).
\end{definition}

We emphasize that a projective implementation includes two steps: (1) measure a state under $\cE$ to obtain a probability distribution $(p, 1-p)$, (2) output $1$ with probability $p$ and otherwise output $0$.
Zhandry shows that any binary POVM has a projective implementation, as in the previous definition.

\begin{lemma}[Adapted from Lemma 1 in \cite{z20}]
\label{lem:proj_implement}
    Any binary outcome POVM $\mathcal{P} = (P, Q)$ has a projective implementation $\projimp(\cP)$.
    
    Moreover, if the outcome is a distribution $(p, 1-p)$ when measuring under $\cE$, the collapsed state $\rho'$ is a mixture of eigenvectors of $P$ with eigenvalue $p$, and it is also a mixture of eigenvectors of $Q$ with eigenvalue $1 - p$. Thus, we have $\Tr[P \rho'] =p$.
\end{lemma}

\paragraph{POVMs and Mixture of Projective Measurements}
In this work, we are interested in projective implementations of POVMs with a particular structure. These POVMs represent a challenger's test of a quantum state received from an adversary in a security game (like the POVM described earlier for testing whether a program evaluates correctly on a uniformly random input). These POVMs have the following structure:
\begin{itemize}
    \item Sample a projective measurement from a set of projective measurements $\mathcal{I}$, according to some distribution $D$ over $\mathcal{I}$.
    \item Apply this projective measurement.
\end{itemize}

We refer to POVMs of this form as \emph{mixtures of projective measurements}. The following is a formal definition.

\begin{definition}[Collection and Mixture of Projective Measurements] \label{def:mixture_of_projective}
Let $\mathcal{R}$, $\mathcal{I}$ be sets. Let $\{(P_i, Q_i)\}_{i \in I}$ be a collection of binary projective measurements $(P_i, Q_i))$ 
over the same Hilbert space $\cH$ where $\cP_i$ corresponds to
output 0, and $\cQ_i$ corresponds to output 1. We will assume we can efficiently measure the $\cP_i$ for
superpositions of $i$, meaning we can efficiently perform the following projective measurement over
$\cI \otimes \cH$:
\begin{align}
\label{eqn:collection of proj measure}
(\sum_i \ket{i}\bra{i}\otimes P_i, \sum_i \ket{i}\bra{i}\otimes Q_i)
\end{align}

 Let  $\cD: \mathcal{R} \rightarrow \mathcal{I}$ be some distribution. 
The \emph{mixture of projective measurements} associated to $\mathcal{R}$, $\mathcal{I}, \cD$ and $\{(P_i, Q_i)\}_{i \in I}$ is the binary POVM $\cP_\cD = (P_\cD, Q_\cD)$ defined as follows: 
\begin{align}
\label{equ: mixture of proj measure}
 P_D = \sum_{i \in \cal I} \Pr[i \gets \cD(R)] \, P_i ,\,\,\,\,\,\text{  }\,\,\,\,\,  Q_\cD = \sum_{i \in \cal I} \Pr[i \gets \cD(R)] \, Q_i, 
\end{align}
\end{definition}

\jiahui{Given an example of ProjImp like what D means in a decryption}


In other words, $\cP_D$ is implemented in the following way: sample randomness $r \gets \cR$, compute the index $i = D(r)$, and apply the projective measurement $(P_i, Q_i)$. Thus, for any quantum state $\rho$, $\Tr[P_D \rho]$ is the probability that a projective measurement $(P_i, Q_i)$, sampled according to the distribution induced by $D$, applied to $\rho$ outputs $1$.

\paragraph{Example of Our POVM} To further explain the above definition, we consider a concrete example: when the input quantum state $\rho$ is supposedly a decryptor for an encryptions with respect to public key $\pk$, and our goal is to measure if $\rho$ can distinguish between encryption of message 0 and message 1.

The distribution $\cD$ is 
the distribution over all randomness used to encrypt a message and the coin flip $b$ to decide which ciphertext to feed to the adversary. For a ciphertext $\ct_i \gets \cD$, $\cP_{\ct_i}$ is the measurement that runs the adversary $\rho$ on $\ct_i$ and tests if the outcome $b' = b$.

\vspace{1em}

\paragraph{Threshold Implementation}
The concept of threshold implementation of a POVM was a simple extension of Projective Implementation \Cref{def:project_implement}, proposed by Zhandry, and formalized by Aaronson, Liu, Liu, Zhandry and Zhang~\cite{aaronsonnew}.

As anticipated, the procedure that we will eventually use to test a state received from the adversary will be to:
\begin{itemize}
    \item[(i)] \emph{Measure} the success probability of the state,
    \item[(ii)] Accept if the outcome is large enough. 
\end{itemize}
As you may guess at this point, we will employ the projective measurement $\cE$ defined previously for step $(i)$. We call this variant of the projective implementation a \emph{threshold implementation}.
From the above lemma, it is straightforward to see that if the outcome of $\projimp(\cP)$ is $(d_1, d_0)$, let $\rho'$ be the collapsed state; then $\Tr[P \rho'] = d_1$. 

In the example of decryption, let $\cP_D = (P_D, Q_D)$ be the binary outcome POVM for checking whether a quantum decryptor works correctly on a uniformly sampled ciphertext. Applying $\cP_D$ is equivalent to doing the projective implementation $\projimp(\cP_D)$. 
If in the first stage of $\projimp(\cP_D)$, the measurement $\cE$ outputs $D = (p, 1-p)$ on a quantum state $\ket \psi$, we say the collapsed state is a `$p$-good decryptor' because it is a eigenvector of $P_D$ with eigenvalue equal to $p$.  if we apply the real test procedure $\cP_D$ (taken the randomness of how a ciphertext is sampled) on the collapsed state, it will decrypt correctly with probability $p$.

The following is a formal definition for Threshold Implementation.
\begin{definition}[Threshold Implementation]
\label{def:thres_implement}
Let $\cP = (P, Q)$ be a binary POVM. Let $\projimp(\cP)$ be a projective implementation of $\cP$, and let $\cE$ be the projective measurement in the first step of $\projimp(\cP)$ (using the same notation as in Definition \ref{def:project_implement}). Let $\gamma >0$. We refer to the following measurement procedure as a \emph{threshold implementation} of $\cP$ with parameter $\gamma$, and we denote is as $\ti_\gamma(\cP)$.
\begin{itemize}
        \item Apply the projective measurement $\cE$, and obtain as outcome a vector $(p, 1-p)$;
        \item Output a bit according to the distribution $(p, 1-p)$: output $1$ if $p \geq \gamma$, and $0$ otherwise. 

        
\end{itemize}
\end{definition}

For simplicity, for any quantum state $\rho$, we denote by $\Tr[\ti_{\gamma}(\cP) \, \rho]$ the probability that the threshold implementation applied to $\rho$ \textbf{outputs} $\mathbf{1}$. Thus, whenever $\ti_{\gamma}(\cP)$ appears inside a trace $\Tr$, we treat $\ti_{\gamma}(\cP)$ as a projection onto the $1$ outcome (i.e. the space spanned by eigenvectors of $P$ with eigenvalue at least $\gamma$).

Similarly to \Cref{lem:proj_implement}, \cite{aaronsonnew} has the following lemma.

\begin{lemma}
\label{lem:threshold_implementation}
    Any binary outcome POVM $\mathcal{P} = (P, Q)$ has a threshold   implementation $\ti_{\gamma}(\cP)$ for any $\gamma$. 
\end{lemma}

\vspace{1em}

\paragraph{Purifying Mixtures of Projection}
Before we go into the next subsection \Cref{sec:ati} on how to compute these projections, we first define the purified version of the above operations, called Controlled Projections:

\begin{definition}[Controlled Projective Measurements] \label{def:controlled projection}
Let $\cP = \{\cP_i = (P_i
, Q_i)\}, i \in \cI$ be a collection of projective measurements
over $\cH$. Let $\cD$ a distribution with random coin set $\cR$. We will abuse notation and let $\cR$
also denote the $|\cR|$-dimensional Hilbert space. The controlled projection is the measurement:
$\cproj_{\cP,\cD} = (\cproj^0_{\cP,\cD}
, \cproj^1_{\cP,\cD}) $ where:
\begin{align*}
\cproj^0_{\cP,\cD} = \sum_{r \in \cR}  \ket{r}\bra{r} \otimes P_{\cD(r)}, \cproj^1_{\cP,\cD} = \sum_{r \in \cR}  \ket{r}\bra{r} \otimes Q_{\cD(r)}
\end{align*}
\end{definition}
$\cproj_{\cP,\cD}$ can be implemented using the measurement $(\sum_i \ket{i}\bra{i}\otimes P_i, \sum_i \ket{i}\bra{i}\otimes Q_i)$. First, initialize control
registers $\cI$ to $0$. Then perform the map $\ket{r}\ket{i} \to \ket{r}\ket{i \oplus \cD(r)}$ to the $\cR \otimes \cI$ registers. Next,
apply the mixture of projective measurements assumed in \Cref{eqn:collection of proj measure}. Finally, perform the map
$\ket{r}\ket{i} \to \ket{r}\ket{i \oplus \cD(r)}$ again to uncompute the control registers, and discard the control registers.

\vspace{1em}


\subsection{Approximating Threshold Implementation}
\label{sec:ati}
\emph{Projective} and \emph{threshold} implementations of POVMs are unfortunately not efficiently computable in general.

However, they can be approximated if the POVM is a mixture of projective measurements, as shown by Zhandry \cite{z20}, using a technique first introduced by Marriott and Watrous \cite{marriott2005quantum} in the context of error reduction for quantum Arthur-Merlin games.

\paragraph{The Uniform Test}
Before we describe the $\ati$ algorithm, we define the projection on register $\cR$ that tests if it is a superposition with each $r \in \cR$ weighted according to the distribution $\cD$ \footnote{The distribution $\cD$ used in $\cproj$ can in fact be an arbitrary distribution, instead of uniform. Without loss of generality, we can use a uniform test $\isuniform$ and map each $r$ to an arbitrary distribution $\cD(r)$ as specified in \Cref{def:controlled projection}.}: $\isuniform = (\ket{\mathbf{1}_\cR}\bra{\mathbf{1}_\cR}, \mathcal{I} - \ket{\mathbf{1}_\cR}\bra{\mathbf{1}_\cR} )$ where:
\begin{align*}
    \ket{\mathbf{1}_\cR} = \frac{1}{\sqrt{\vert \cR \vert}} \sum_{r \in \cR} \ket{r} 
\end{align*}

\paragraph{The Algorithm $\ati$}
We present the algorithm $\ati_{\cP,\cD, \gamma}^{\epsilon,\delta}$ using the syntax from \cite{z20}:

Our algorithm is parameterized by a distribution $\cD$, collection of projective
measurements $\cP$, and real values $0 < \epsilon, \delta, \gamma \leq 1$, and is denoted as $\ati_{\cP,\cD, \gamma}^{\epsilon, \delta}$. 
The running time of the following algorithm in \Cref{fig:ati_algortihm} is $\frac{2\ln(4/\delta)}{\epsilon^2}$.

\begin{figure}
    \centering
    \caption{ATI Algorithm}
    
\begin{mdframed}[frametitle = {Approximate Threshold Implementation}]

On input a quantum state
$\ket{\psi}$ over Hilbert space $\cH$:

\begin{enumerate}
 \item Initialize a state $\ket{\mathbf{1}_\cR} \ket{\psi}$. 

 \item Initialize a classical list $L := (1)$.
 
\item Repeat the following loop a total of $T = \frac{\ln(4/\delta)}{\epsilon^2}$ times:

\begin{enumerate}
    \item Apply $\cproj_{\cP,\cD}$ to register $\cR \otimes \cH$. Let $b_{2i-1}$ be the measurement outcome and set $L := (L, b_{2i-1})$.

    \item Apply $\isuniform$ to register $\cR$. Let $b_{2i}$ be the measurement outcome and set $L := (L, b_{2i})$.
\end{enumerate}

\item Let $t$ be the number of index $i$ such that $b_{i-1} = b_i
$ in the list $L = (0, b_1
, . . . , b_{2T})$, and $\tilde{p} := t/2T$.

\item If $b_{2T} = 0$, repeat the loop again until $b_{2i} = 1$.

\item Discard the $\cR$ register.

\item Output 1 if $\tilde{p} \geq \gamma$ and 0 otherwise. 
\end{enumerate}
\end{mdframed}
\label{fig:ati_algortihm}
\end{figure}

\begin{remark}
    Note that in the ATI algorithm, if we output value $\tilde{p}$ after Step 6 directly, the algorithm is called API(approximate projective implementation, \cite{z20}), as the approximate algorithm for projective implementation \Cref{def:project_implement}.
\end{remark}

\paragraph{ATI Properties}
We will make use of the following lemma from a subsequent work of Aaronson et al.~\cite{aaronsonnew}.



\begin{lemma}[Corollary 1 in \cite{aaronsonnew}]
\label{cor:ati_thresimp}
    For any $\epsilon, \delta, \gamma \in (0,1)$, any collection of projective measurements $\cP = \{(P_i, Q_i)\}_{i \in \mathcal{I}}$, where $\mathcal{I}$ is some index set, and any distribution $D$ over $\mathcal{I}$, there exists a measurement procedure $\ati^{\epsilon, \delta}_{\cP, D, \gamma}$ that satisfies the following:
    \begin{enumerate}
        
        \item $\ati^{\epsilon, \delta}_{\cP, D, \gamma}$ implements a binary outcome measurement. For simplicity, we denote the probability of the measurement \textbf{outputting} $\mathbf{1}$ on $\rho$ by $\Tr[\ati^{\epsilon, \delta}_{\cP, D, \gamma}\, \rho]$. 
    
        \item For all quantum states $\rho$, $\Tr[\ati^{\epsilon, \delta}_{\cP, D, \gamma-\epsilon}\, \rho]\geq \Tr[\ti_\gamma(\cP_D)\, \rho]-\delta$. 
        
        \item 
        For all quantum states $\rho$, 
        let $\rho'$ be the post-measurement state after applying  $\ati^{\epsilon, \delta}_{\cP, D, \gamma}$ on $\rho$, and obtaining outcome $1$. Then, $\Tr[\ti_{\gamma-2 \epsilon}(\cP_D)\, \rho'] \geq 1 - 2\delta$. 
        \item
        The expected running time is $O(T_{\cP, D} \cdot 1/\epsilon^2 \cdot 1/(\log \delta))$, where $T_{\cP, D}$ is the combined running time of sampling according to $D$, of mapping $i$ to $(P_{i}, Q_{i})$, and of implementing the projective measurement $(P_{i}, Q_{i})$ \footnote{In our setting, $T$ is simply the running time of the quantum algorithm $\rho$.}. 
    \end{enumerate}
\end{lemma}

Intuitively the corollary says that if a quantum state $\rho$ has weight $p$ on eigenvectors with eigenvalues at least $\gamma$, then the measurement $\ati^{\epsilon, \delta}_{\cP, D, \gamma}$ will produce with probability at least $p - \delta$ a post-measurement state which has weight $1 - 2 \delta$ on eigenvectors with eigenvalues at least $\gamma - 2 \epsilon$. Moreover, the running time for implementing $\ati^{\epsilon, \delta}_{\cP, D, \gamma}$ is proportional to $\poly(1/\epsilon, 1/(\log \delta))$, which is a polynomial in $\lambda$ as long as $\epsilon$ is any inverse polynomial and $\delta$ is any inverse sub-exponential function.

\paragraph{$\ti$ and $\ati$ For Computationally/Statistically Indistinguishable Distributions}
The following theorems will be used in the proof of security for our SKL encryption scheme in Section \ref{sec:main_security_proof}.

Informally, the lemma states the following. Let $\cP_{D_0}$ and $\cP_{D_1}$ be two mixtures of projective measurements, where $D_0$ and $D_1$ are two computationally indistinguishable distributions. Let $\gamma, \gamma'>0$ be inverse-polynomially close. Then for any (efficiently constructible) state $\rho$, the probabilities of obtaining outcome $1$ upon measuring $\ti_{\gamma}(\cP_{D_0})$ and $\ti_{\gamma'}(\cP_{D_1})$ respectively are negligibly close.


\begin{theorem}[Theorem 6.5 in \cite{z20}] \label{thm:ti_different_distribution}
Let $\gamma >0$. Let $\cP$ be a collection of projective measurements indexed by some set $\cal I$. Let $\rho$ be an efficiently constructible mixed state, and let $D_0, D_1$ be two efficiently sampleable and computationally indistinguishable distributions over $\cal I$. For any inverse polynomial $\epsilon$, there exists a negligible function $\delta$ such that
\begin{align*}
    \Tr[\ti_{\gamma - \epsilon}(\cP_{D_1}) \rho] \geq \Tr[\ti_{\gamma}(\cP_{D_0}) \rho] - \delta \,,
\end{align*}
where $\cP_{D_i}$ is the mixture of projective measurements associated to $\cP$ and $D_i$.
\end{theorem}

\begin{corollary}
[Corollary 6.9 in \cite{z20}]
\label{cor:ati_computational_indistinguishable}
Let $\rho$ be an efficiently constructible, potentially mixed state, and let $\cD_0, \cD_1$ be two
computationally indistinguishable distributions. Then for any inverse polynomial $\epsilon$ and any function
$\delta$, there exists a negligible $\negl(\cdot)$ such that:
$$ \Tr[\ati_{\cD_1, \cP, \gamma - 3\epsilon}^{\epsilon, \delta}(\rho)] \geq \Tr[\ati_{\cD_0, \cP, \gamma}^{\epsilon, \delta}(\rho)] - 2\delta - \negl(\lambda) $$

\end{corollary}

\subsubsection{Identity-like Approximate Threshold Implementation}
\label{sec:invariant_ati_uniform}

In this section, we introduce a property of ATI when using under certain "dummy" mixture of projections(POVM), will act  like an identity operator on the input quantum state.

We first recall the following POVM (also called mixture of projections \Cref{def:mixture_of_projective}) $ \cP_{\cD_0,\cD_1}$:
\begin{definition}[Distribution Distinguishing POVM]
\label{def:distribution_distinguish_povm}
The POVM(mixture of projections) measurement $\cP=\cP_{\cD_0,\cD_1} $ is defined as follows on any input quantum state:
\begin{itemize}
    \item The randomness $\cR$ is over two distributions $\cD_0, \cD_1$ and a uniform $b \in \{0,1\}$.

    \item Sample a bit $b \gets \{0,1\}$ and then sample $v \gets \cD_b$.

    \item Feed $v$ to the input quantum state which will output a guess $b'$.

    \item Output 1 if $b' = b$; 0 otherwise.
\end{itemize}
\end{definition}

\paragraph{Dummy Distribution Distinguisher ATI Acts like Identity Operator}
Now we give a theorem on the property of $\ati^{\epsilon,\delta}_{\cP, 1/2+\gamma}$, where $\cP = \cP_{\cD_0,\cD_1}, $ is defined in \Cref{def:distribution_distinguish_povm} and show that it almost acts as an identity operator on any input state when $\cD_0, \cD_1$ are statistically close.

\begin{theorem}
\label{thm:invariant_ati_uniform}
  Suppose the statistical distance of $\cD_0,\cD_1$ is $\eta$.
$\ati^{\epsilon,\delta}_{\cP, 1/2+\gamma}$ is the ATI for mixture of projections $\cP = \cP_{\cD_0,\cD_1}$ defined in \Cref{def:distribution_distinguish_povm}, with parameters $\gamma,\epsilon,\delta$ such that $\gamma$ is inverse polynomial and $\epsilon < \gamma$. The distance $\eta$ also satisfies $\eta < \epsilon/2$. We have the following properties, with probability $(1-\delta)$:

\begin{itemize}
    \item For any input state $\rho$,
let $\rho'$ be the state after applying $\ati^{\epsilon,\delta}_{\cP, 1/2+\gamma}$ on $\rho$, we have $\lVert \rho-\rho'\rVert_{\Tr} \leq O( \eta \cdot \frac{\ln(4/\delta)}{\epsilon})$ \footnote{In our applications,  $\eta \cdot \frac{\ln(4/\delta)}{\epsilon}$ is inverse exponential in security parameter $\lambda$.}.

\item $\ati^{\epsilon,\delta}_{\cP, 1/2+\gamma}$ will output outcome 0.
\end{itemize}

\end{theorem}


\ifllncs
Due to limitation of space, we refer the proof to \Cref{sec:invariant_ati_proof_append}.

We can also conclude corollary \Cref{cor:invariant_ati_same_distr} for the setting where $\cD_0, \cD_1$ are the same distribution: directly by plugging in $\eta = 0$.
In this case $\ati$ acts exactly as an identity operator.

\else

\begin{proof}
    We obtain this theorem by making observations on the state after running the ATI algorithm, see \Cref{fig:ati_algortihm}.

    We will use Claim 6.3 in \cite{z20}, and some explanations under theorem 6.2 in \cite{z20}.(For the clarity of presentation, we will integrate \cite{z20}'s claims into our context instead of citing it directly.)

We can write the input state at the beginning as $\rho =  \sum_i \lambda_i \ket{\phi_i}\bra{\phi_i}$, where $p_i$ is an eigenvalue with eigenvector $\ket{\phi_i}$ for measurement $\cP_{\cD_0,\cD_1}$ in \Cref{def:distribution_distinguish_povm}(see discussions in\Cref{sec:project_imp} and 
\Cref{lem:proj_implement}). 
In other words, $\cP_{\cD_0,\cD_1}\ket{\phi_i} = p_i\ket{\phi_i}$. Meanwhile, 
$\{\ket{\phi_i}\}_i$ form a basis, and the density operator $\rho$ can always be written in terms of coefficients $\{\lambda_i\}_i$ with their corresponding state in $\{\ket{\phi_i}\}_i$  because such a spectral decomposition always exists for any $\rho$.

Since $\cD_0, \cD_1$ are $\eta$ close in statistical distance, we know that $p_i$ must take values in $[1/2, 1/2+\eta]$ for all $i$,  because no distinguisher can perform better than distinguishing them with probability larger than $1/2+\eta$ (and can go as bad as a uniform random guess). 

We consider the state after doing $T$ number of iterations of the main loop in $\ati^{\epsilon,\delta}_{\cP, 1/2+\gamma}$ algorithm (see \Cref{fig:ati_algortihm}).
Let $L = (L_1, L_2, \cdots, L_{2T})$ be the sequence of measurement outcomes we obtain in 
the ATI algorithm from the alternating tests $\cproj_{\cP}$ and $\isuniform$; let $t$ be the number of $i$ such that $L_{i-1} = L_i$ we have in $L$.

Since at the end of the algorithm $\ati$, we will stop at a $\isuniform$ test with measurement outcome 1,
we obtain the following state state, as claimed by  \cite{z20} Claim 6.3.
\begin{align}
    \label{eqn:state_after_ati}
    \rho' = \sum_i \lambda_i \cdot \frac{p_i^t \cdot (1-p_i)^{2T-t}}{\sum_i \lambda_i \cdot p_i^t \cdot (1-p_i)^{2T-t}} \ket{\phi_i}\bra{\phi_i}
\end{align}
\cite{z20} gave the form of an unnormalized pure state w.l.o.g.; we can easily derive that the form of a (normalized) mixed state evolved from its initilized form $\rho =  \sum_i \lambda_i \ket{\phi_i}\bra{\phi_i}$ will be the above. 

Also by \cite{z20} claim 6.3, according to Chernoff-Hoeffding bound, we have for $ T \geq \frac{\ln(4/\delta)}{\epsilon^2}$ that:
\begin{align}
\label{eqn:chernoff}
    \Pr[\lvert \tilde{p} - p \lvert \leq \epsilon/2] \geq  1- e^{-2(2T)/(\epsilon^2/2)}  \geq 1-\delta, \text{ where } \tilde{p} = \frac{t}{2T}
\end{align}
Here $\tilde{p} = \frac{t}{2T}$ is the value we obtain at the next-to-last step in ATI algorithm in \Cref{fig:ati_algortihm} (before we check $\tilde{p} \geq 1/2+\gamma$); $p$ is an outcome we obtain if we apply an ideal projective implementation of $\cP$-- in short we obtain some $p = p_{i^*}$ with probability $\sum_j \lambda_j$  where all $\ket{\phi_j}$ have eigenvalues $p_j = p_{i^*}$ under $\cP$ (see discussions in \Cref{sec:project_imp}).
But we don't need to worry about the detailed distribution of $p$ here: since we must have $p\in [1/2, 1/2+\eta]$, we can deduce that $\frac{(1-\epsilon)}{2} \cdot 2T \leq t \leq (\frac{(1+\epsilon)}{2} + \eta)\cdot 2T$. Since we have $\eta < \epsilon/2$  by our assumption and let $T = \ln(4/\delta)/\epsilon^2$ we have $T-\frac{\ln(4/\delta)}{\epsilon} \leq t \leq T+\frac{2\ln(4/\delta)}{\epsilon}$, with probability $1-\delta$. 
Now we analyze the trace distance between $\rho, \rho'$. 
Without loss of generality, first assume $t\geq T$ and $ \lambda_i - \lambda_i \cdot \frac{p_i^t \cdot (1-p_i)^{2T-t}}{\sum_i \lambda_i \cdot p_i^t \cdot (1-p_i)^{2T-t}} \geq 0, \forall i$.
\begin{align*}
    \lVert \rho - \rho' \rVert_{\Tr} &= \frac{1}{2}  \sum_i \lvert (\lambda_i - \lambda_i \cdot \frac{p_i^t \cdot (1-p_i)^{2T-t}}{\sum_i \lambda_i \cdot  p_i^t \cdot (1-p_i)^{2T-t}}) \rvert  \\
    & = \frac{1}{2} \sum_i \lambda_i (1-  \frac{p_i^t \cdot (1-p_i)^{2T-t}}{\sum_i \lambda_i \cdot p_i^t \cdot (1-p_i)^{2T-t}})  \\
    & = \frac{1}{2} \sum_i \lambda_i (1-  \frac{p_i^T \cdot (1-p_i)^{T} \cdot p_i^{t-T} (1-p_i)^{T-t}}{\sum_i \lambda_i \cdot  p_i^T \cdot (1-p_i)^{T} p_i^{t-T} (1-p_i)^{T-t}})  \\
    & \leq \frac{1}{2} \sum_i \lambda_i (1-  \frac{p_i^T \cdot (1-p_i)^{T} \cdot p_i^{t-T} (1-p_i)^{T-t}}{\sum_i \lambda_i \cdot (\frac{1}{2})^{2T} (\frac{1/2+\eta}{1/2-\eta})^{t-T}} )  \\
    & \leq  \frac{1}{2} \sum_i \lambda_i (1-  \frac{p_i^T \cdot (1-p_i)^{T} \cdot p_i^{t-T} (1-p_i)^{T-t}}{(\frac{1}{4})^T \cdot (\frac{1/2+\eta}{1/2-\eta})^{t-T}}) \allowdisplaybreaks
    \\
    & \leq \frac{1}{2} \sum_i \lambda_i (1-  \frac{(1/4 - \eta^2)^T \cdot (\frac{1/2-\eta}{1/2+\eta})^{t-T}}{(\frac{1}{4})^T \cdot (\frac{1/2+\eta}{1/2-\eta})^{t-T}}) \allowdisplaybreaks  \\
    & = \frac{1}{2} \sum_i \lambda_i (1-  (1-4\eta^2)^T \cdot (\frac{1/2-\eta}{1/2+\eta})^{2(t-T)}) \allowdisplaybreaks \\
    & \leq  \frac{1}{2} \sum_i \lambda_i (1-  (1-4T\eta^2) \cdot (1 -4\eta)^{2t-2T}) \allowdisplaybreaks \\
    & \leq \frac{1}{2} \sum_i \lambda_i (1-  (1-4T\eta^2)  \cdot  (1 -4\eta(2t-2T))) \allowdisplaybreaks \\ 
    & \leq \frac{1}{2} \sum_i \lambda_i ( 1-(1-4\frac{\ln(4/\delta)}{\epsilon^2} \cdot\eta^2 - 8\eta \cdot \frac{2\ln(4/\delta)}{\epsilon}  +  4\eta^2 \cdot  8\eta \cdot \frac{2(\ln(4/\delta))^2}{(\epsilon^3)})) \\
    & \leq  O(\eta \cdot \frac{\ln(4/\delta)}{\epsilon}) \text{, since  $\eta < \epsilon$}.
\end{align*}
The second to last inequality is because $T = \frac{\ln(4/\delta)}{\epsilon^2}$, and $ t-T \leq\frac{\ln(4/\delta)}{\epsilon}$.

The analysis for $t < T$ is symmetric, by swapping the assignment of $p_i, 1-p_i$ for values $(1/2+\eta, 1/2-\eta)$ to satisfy the relaxations and will obtain the same outcome.
Similarly, considering the other extreme case when $ \lambda_i - \lambda_i' <0, \forall i$ where $\lambda_i' = \lambda_i \cdot \frac{p_i^t \cdot (1-p_i)^{2T-t}}{\sum_i \lambda_i \cdot \cdot p_i^t \cdot (1-p_i)^{2T-t}}$, we will have: 
\begin{align*}
     \lVert \rho - \rho' \rVert_{\Tr} & \leq
     \frac{1}{2} \sum_i \lambda_i ( (1-4\eta^2)^{-T} \cdot (\frac{1/2+\eta}{1/2-\eta})^{2(t-T)}-1)    \\   
     & \leq \frac{1}{2} \sum_i \lambda_i (  (1-4T\eta^2)^{-1}  \cdot  (1 - 4\eta(2t-2T))^{-1} - 1) \\
     & \leq \frac{1}{2} \frac{O(\eta \cdot \frac{\ln(4/\delta)}{\epsilon})}{1 - O(\eta \cdot \frac{\ln(4/\delta)}{\epsilon})} \\
     & \leq O(\eta \cdot \frac{\ln(4/\delta)}{\epsilon})
\end{align*}
where the last inequality is because $\eta < \epsilon$ and $\ln(4/\delta)$ is inverse polynomial; therefore there exists some constant $0<c<1$ such that $ 1- c \leq 1- O(\eta \cdot \frac{\ln(4/\delta)}{\epsilon})$. 

Therefore overall we have $ \lVert \rho - \rho' \rVert_{\Tr} \leq \frac{1}{2}\sum_{i: \lambda_i - \lambda_i' \geq 0} \vert \lambda_i - \lambda_i'  \vert + \frac{1}{2}\sum_{i: \lambda_i - \lambda_i' < 0} \vert \lambda_i - \lambda_i'  \vert  \leq O(\eta \cdot \frac{\ln(4/\delta)}{\epsilon})$.

We can also observe from \Cref{eqn:chernoff} that with probability $1-\delta$, $\frac{1}{2}-\epsilon/2 \leq \tilde{p} \leq 1/2+\epsilon << 1/2+\gamma$ and $\ati$ will output 0.
\end{proof}

\begin{corollary}
\label{cor:invariant_ati_same_distr}
     Suppose the statistical distance of $\cD_0,\cD_1$ is 0.
$\ati^{\epsilon,\delta}_{\cP, 1/2+\gamma}$ is the ATI for mixture of projections $\cP = (\cP_{\cD_0,\cD_1},  \cQ_{\cD_0,\cD_1})$ defined in \Cref{def:distribution_distinguish_povm}, with parameters $\gamma,\epsilon,\delta$ where $\gamma$ is inverse polynomial and $\epsilon < \gamma$. For any input state $\rho$,
let $\rho'$ be the state after applying $\ati^{\epsilon,\delta}_{\cP, 1/2+\gamma}$ on $\rho$, with probability $(1-\delta)$, we have that $\lVert \rho-\rho'\rVert_{\Tr} = 0$ and $\ati^{\epsilon,\delta}_{\cP, 1/2+\gamma}$ outputs 0.  
\end{corollary}

\jiahui{is this probability $1-\delta$ or 1?}

We can obtain this corollary directly by plugging in $\eta = 0$.
In this case $\ati$ acts exactly as an identity operator.

\fi

\section{Secure Key Leasing with Classical Communication: Definition}
\label{sec:defs}
\newcommand{\Eval}{\mathsf{Eval}}

\jiahui{removed the quantum key gen defs}

\subsection{Secure Key Leasing PKE with Classical Lessor}

 \jiahui{stress that our actual construction is not really "interactive" because we only need one lessor message and ask the user to output its public key}
\begin{definition}[Public Key Encryption with Classcal Leaser]
\label{def:interactive_skl}
A PKE with secure key leasing with a classical vendor consists of the algorithms ($\Setup$, $\KeyGen$, $\Enc$, $\Dec$, $\delete$, $\VerDel$) defined as follows:
\jiahui{add $\Setup$ algo}
\begin{item}

    \item $\Setup(1^\lambda)$:  take input a security parameter $\lambda$, output a (classical) master public key $\mpk$ and a (classical) trapdoor $\td$.
    
     \item $\KeyGen(1^\lambda)$:   
       take as input a (classical) public key $\mpk$, output
      a quantum decryption key $\rho_\sk$ and a classical public key $\pk$.
      
        \item $\Enc(\pk,\mu)$: given a public key $\pk$ and a plaintext $\mu \in \{0,1\}$, output a ciphertext $\ct$.

        \item $\Dec(\rho_\sk,\ct)$: given a quantum state $\rho_\pk$ and a 
        ciphertext $\ct$, output a message $\mu$ and the state $\rho_\sk'$
        \item $\delete(\rho_\sk)$: given the quantum state $\rho_\sk$, output a classical deletion certificate $\cert$
        \item $\VerDel(\pk,\td,\cert)$: given a public key $\pk$, a classical certificate $\cert$ and the trapdoor $\td$, output $\mathsf{Valid}$ or $\mathsf{Invalid}$.
\end{item}


\end{definition}

We refer the readers to \Cref{sec:relation_security_defs} on a few side comments and remarks about the definition.

\paragraph{Interactive Key Generation}
In our actual classical lessor protocol, the lessor and lessee run an interactive protocol $\InKeyGen$:
 \begin{itemize}
     \item $\InKeyGen(1^\lambda)$: an interactive protocol that takes in a security parameter $1^\lambda$; $\A$ interacts with the challenger (classically) and they either output a public key $\pk$, a trapdoor $\td$ and a quantum secret key $\rho_\sk$ or output a symbol $\bot$.
 \end{itemize}

The necessity of an interactive protocol is to prevent "impostors" who sees the master public key and generate public key-decryption key pairs on their own without consent of the lessor. See \Cref{sec:relation_security_defs} for details. 
This security is orthogonal to our security defined in \Cref{def:regular_skl_security_classical} and \Cref{def:gamma_good_decryptor} because our adversary implicitly commits to one $\pk$. 



\paragraph{Correctness} A PKE Scheme with secure quantum key leasing  ($\Setup, \KeyGen$, $\Enc$, $\Dec$, $\delete$, $\VerDel$) satisfies correctness if the following hold.\\\\
\textbf{Decryption Correctness}: There exists a negligible function $\negl(\cdot)$, for all $\lambda \in \N$, for all $\mu \in \mathcal{M}$:

    $$\Pr \left[\Dec(\rho_\sk,\ct) = \mu : \begin{matrix} 
    (\mpk, \td) \gets \Setup(1^\lambda) \\
    (\pk,\rho_\sk)\leftarrow \KeyGen(\mpk) \\ \ct \leftarrow \Enc(\pk,\mu) \end{matrix}\right] \ge 1 - \negl(\lambda)$$ 

\paragraph{Reusability}
the above decryption correctness should hold for an arbitrary polynomial number of uses.

\paragraph{Verifying Deletion Correctness}: There exists a negligible function $\negl(\cdot)$, for all $\lambda \in \N$:

$$ \Pr \left[\mathsf{Valid} \leftarrow\VerDel(\pk,\td,\cert) : \begin{matrix} (\mpk, \td) \gets \Setup(1^\lambda) \\
    (\pk,\rho_\sk)\leftarrow \KeyGen(\mpk) \\ \cert \leftarrow \delete(\rho_\sk) \end{matrix}\right] \ge 1 - \negl(\lambda) $$

\jiahui{OLD regular security below. }
\paragraph{IND-SKL-PKE Security}
We give a "classical friendly" security definition same as the one used in \cite{ananth2023revocable,agrawal2023public}, except that we have a classical leaser.


 %

Then in the next subsection \Cref{sec:strong_security}, we will then present a "strong" security definition in the following section, \cref{sec:strong_security}. The latter is the one we will actually use in the proof and will imply the following IND-PKE-SKL security.

\begin{definition}[IND-PKE-SKL Security(Classical Client)]
\label{def:regular_skl_security_classical}

where the experiment $\text{IND-PKE-SKL}(\A,1^\lambda, b \in \{0,1\})$ between a challenger and the adversary $\A $ is defined as follows:
\begin{itemize}
   \item The challenger runs $\Setup(1^\lambda) \to (\mpk, \td)$.
    It sends $\mpk$ to the adversary $\A$. $\A$ computes $(\pk, \rho_\sk) \gets \KeyGen(\mpk)$ and publishes $\pk$. 
    
    
    \item The challenger requests that $\A$ runs the deletion algorithm $\delete$. $\A$ returns a deletion certificate $\cert$. 
    

    \item The challenger runs $\VerDel(\pk,\td,\cert)$ and continues if $\VerDel(\pk,\td,\cert)$ outputs $\valid$; else the challenger outputs $\bot$ and aborts.
    

    \item $\A$ submits a plaintext $\mu \in \{0,1\}^\ell$ to the challenger.
    \item The challenger flips a bit $b \samp \{0,1\}$.
    \item If $b=0$, the challenger sends back the ciphertext $\ct \leftarrow \Enc(\pk,\mu)$. If $b=1$ the challenger sends a random ciphertext from the possible space of all ciphertexts for of $\ell$-bit messages, $\ct \samp \mathcal{C}$
    \item Output $\A$'s guess for the bit b, b'. 
    
\end{itemize}
A PKE Scheme with Secure Key Leasing and fully classical communication ($\Setup, \KeyGen$, $\Enc$, $\Dec$, $\delete$, $\VerDel$) is secure if, for every $\QPT$ adversary $\A$ if there exists negligible functions $\negl(\cdot)$
 such that one of the following holds for all $\lambda \in \N$:  
\begin{align*}
    & \lvert \Pr\left[  \text{IND-PKE-SKL}(\A, 1^\lambda, b=0) = 1 \right] - \Pr\left[  \text{IND-PKE-SKL}(\A, 1^\lambda, b=1) = 1 \right] \rvert \le  \negl(\lambda) 
\end{align*}
\end{definition}

\begin{remark}
Regarding the security in \Cref{def:regular_skl_security_classical}:
    In other words, in order to win, the adversary $\A$ needs to do both of the following for some noticeable $\epsilon_1, \epsilon_2$:
    \begin{enumerate}
        \item $\lvert \Pr\left[  \text{IND-PKE-SKL}(\A, 1^\lambda, b=0) = 1 \right] - \Pr\left[  \text{IND-PKE-SKL}(\A, 1^\lambda, b=1) = 1 \right] \rvert \geq  \epsilon_1(\lambda)$ and

    \item $\Pr[\text{IND-PKE-SKL}(\A, 1^\lambda, b \in \{0,1\}) \neq \bot] \geq \epsilon_2(\lambda)$
    \end{enumerate}
We need the second inequality above to hold because in the case where \\ $\Pr[\text{IND-PKE-SKL}(\A, 1^\lambda, b \in \{0,1\}) = \bot] \geq 1- \negl_2(\lambda)$,for some negligible $\negl_2(\cdot)$, the probabilities $\Pr\left[  \text{IND-PKE-SKL}(\A, 1^\lambda, b=0) = 1 \right]$ and \ifllncs
\\ \else \fi $ \Pr\left[  \text{IND-PKE-SKL}(\A, 1^\lambda, b=1) = 1 \right] $ will also have negligible difference.
\end{remark}

\subsection{Strong SKL-PKE PKE Security: Threshold Implementation Version}
\label{sec:strong_security}

\jiahui{This is a more rigorous/ physically implementable security definition of the security definition. }

In this section, we define a security notion we call Strong SKL-PKE, which is described via the measurement $\ti$ in \Cref{sec:unclonable dec ati}. We show that it implies the regular security notion in the previous section. Note that Strong SKL-PKE is the security definition we prove our scheme's security with respect to.

To define the strong security/implementable security, we first define what it means to test the success probability of a quantum decryptor.


\begin{definition}[Testing a quantum decryptor] 
\label{def:gamma_good_decryptor}
   Let $\gamma \in [0,1]$. Let $\pk$ be a public key and $\mu$ be a message. We refer to the following procedure as a {test for a $\gamma$-good quantum decryptor} with respect to $\pk$ and $\mu$ used in the following sampling procedure:
   \begin{itemize}
       \item The procedure takes as input a quantum decryptor $\rho$.
       \item Let $\mathcal{P} = (P, Q)$ be the following mixture of projective measurements (in terms of Definition \ref{def:mixture_of_projective}) acting on some quantum state $\rho$:
       \begin{itemize}
       \item Compute $\ct_0 \leftarrow \Enc(\pk, \mu)$, the encryption of message $\mu \in \{0,1\}$.
       \item Compute $\ct_1 \gets 
       \mathcal{C}$, a random ciphertext from the possible space of all ciphertexts for 1-bit messages.
       \item Sample a uniform $b \leftarrow \{0,1\}$. 
       \item Run the quantum decryptor $\rho$ on input $\ct_b$. Check whether the outcome is $b$. If so, output $1$, otherwise output $0$.
      \end{itemize}
       \item Let $\ti_{1/2 + \gamma}(\cP)$ be the threshold implementation of $\cP$ with threshold value $\frac{1}{2} + \gamma$, as defined in \Cref{def:thres_implement}. Run $\ti_{1/2 + \gamma}(\cP)$ on $\rho$, and output the outcome. If the output is $1$, we say that the test passed, otherwise the test failed.
   \end{itemize}
\end{definition}

By Lemma \ref{lem:threshold_implementation}, we have the following corollary. 

\begin{corollary}[$\gamma$-good Decryptor]
\label{cor: gamma good dec}
    Let $\gamma \in [0,1]$. Let $\rho$ be a quantum decryptor. 
    Let $\ti_{1/2 + \gamma}(\cP)$ be the test for a $\gamma$-good decryptor defined above. Then, the post-measurement state conditioned on output $1$ is a mixture of states which are in the span of all eigenvectors of $P$ with eigenvalues at least $1/2+\gamma$. 
    
\end{corollary}

Now we are ready to define the strong $\gamma$-anti-piracy game. 

\begin{definition}[$\gamma$-Strong Secure Key Leasing Security Game]
\label{def:gamma_skl_game}
 Let $\lambda \in \mathbb{N}^+$, and $\gamma \in [0,1]$.
The strong $\gamma$-PKE-SKL game is the following game between a challenger and an adversary $\mathcal{A}$.
\begin{enumerate}

 \item The challenger runs $\Setup(1^\lambda) \to (\mpk, \td)$.
    It sends $\mpk$ to the adversary $\A$. $\A$ computes $(\pk, \rho_\sk) \gets \KeyGen(\mpk)$ and publishes $\pk$. 
    

    \item  The challenger requests that $\A$ runs the deletion algorithm $\delete(\qsk)$. $\A$ returns a deletion certificate $\cert$ to the challenger. 

    \item The challenger runs $\VerDel(\pk,\td,\cert)$ and continues if $\VerDel(\pk,\td,\cert)$ returns $\valid$; else it outputs $\bot$ and aborts, if $\VerDel(\pk,\td,\cert)$ returns $\invalid$.

    \item 
   $\A$ outputs a message $\mu$ and a (possibly mixed) state $\rho_\delete$ as a quantum decryptor.
    
    \item 
    The challenger runs the test for a $\gamma$-good decryptor on $\rho_\delete$
    with respect to $\pk$ and $\mu$. The challenger outputs $1$ if the test passes, otherwise outputs $0$.
\end{enumerate}
We denote by $\sf{StrongSKL}(1^{\lambda}, \gamma, \mathcal{A})$ a random variable for the output of the game. 
\end{definition}

\begin{definition}[Strong PKE-SKL Security]
\label{def:strong_skl}
 Let $\gamma: \mathbb{N}^+ \rightarrow [0,1]$. A secure key leasing scheme satisfies strong $\gamma$-SKL security, if for any QPT adversary $\A$,  there exists a negligible function $\negl(\cdot)$ such that for all $\lambda \in \N$: 
  \begin{align}
    \Pr\left[b = 1, b \gets \sf{StrongSKL}(1^{\lambda}, \gamma(\lambda), \A) \right]\leq \negl(\lambda)
    \end{align}

\end{definition}

Next, we show that the strong PKE-SKL definition(\Cref{def:strong_skl}) implies the IND-PKE-SKL definition(\Cref{def:regular_skl_security_classical}).

\begin{claim}
\label{thm: strong skl implies regular}
Suppose a secure key leasing scheme satisfies strong $\gamma$-SKL security (Definition \ref{def:strong_skl}) for any inverse polynomial $\gamma$, then it also satisfies IND-PKE-SKL security (Definition \ref{def:regular_skl_security_classical}).
\end{claim}

\begin{proof}
We refer the reader to \cref{sec:relation_security_defs} for the proof.
\end{proof} 


\section{Secure Key Leasing with Classical Lessor/Client: Construction}

\subsection{Parameters}
\label{sec:scheme_parameters}
We present our parameters requirements for the following construction.
All the parameters are the same as in \Cref{sec:ntcf_parameter}, with a few new paramters added.

Let $\lambda$ be the security parameter. All other parameters are
functions of $\lambda$. Let $q \geq 2$ be a prime integer. Let $\ell, k, n, m, w \geq 1$ be polynomially bounded functions of $\lambda$ and
$B_L, B_V, B_P, B_{P'}$ be Gaussian parameters and $B_X, B_S$ are norm parameters such that the following conditions hold:
\begin{enumerate}
    \item $n = \Omega(\ell \log q)$ and $m = \Omega(n \log q)$

 \item $k$ can be set to $\lambda$.

    
 \item $w = \lceil n  \log q \rceil$.

    \item $B_P = \frac{q}{2C_T
\sqrt{mn \log q}}$ where $C_T$ is the constant in \Cref{thm:trapdoor_samp}.

\item $2\sqrt{n} \leq B_L \leq B_V \leq B_P \leq B_X \leq B_{P'}$

\jiahui{the last ratio requirement "$\frac{B_{P'}}{ B_X^4 \cdot B_P}$ is superpolynomial in $\lambda$" is removed}
\item The ratios $\frac{B_V}{B_L}$, $\frac{B_P}{B_V}$, $\frac{B_{P'}}{B_P}$, $\frac{B_X}{B_S}$ are super-polynomial in $\lambda$.

\item $B_{P'} \cdot m^d \leq q$ where $d$ is the depth of the circuit in FHE evaluation, to be further discussed in \Cref{sec:fhe}.

\item We denote $[B_X]$ as all integer taking values $[-B_X, \cdots, B_X]$. 

\item In our scheme, we set $B_S := \{0,1\}$ and $B_X$ to be a quasipolynomial.

\end{enumerate}

\subsection{Scheme Construction}
\label{sec:construction}

 We first present our construction for a secure key leasing protocol. 
 
 We provide the protocol description where lessor is completely classical in \Cref{sec:protocol_description}. 

\begin{itemize}

    \item $\Setup(1^\lambda)$:
        On input the security parameter $1^\lambda$, the $\Setup$
    algorithm works as follows:
    \begin{itemize}
        \item Sample $k = \lambda$  matrices $\bA_i \in \mathbb{Z}_q^{n \times m}$ along with their trapdoors $\td_i$ using the procedure $\mathsf{GenTrap}(1^n,1^m,q)$(\Cref{thm:trapdoor_samp}):
    $(\bA_i, \td_i) \samp  \mathsf{GenTrap}(1^n,1^m,q),\\ \forall i \in [k]$. 
    \item Sample $\bs_{i} \samp [B_s]^{n}, \forall i \in [k]$ and $\be_{i} \samp D_{\Z_q^{ m},B_V}, , \forall i \in [k]$.
    
    \item Output $\mpk = \{f_{i,0}, f_{i,1}\}_{i=1}^k =  \{(\bA_i, \bs_i\bA_i+ \be_i)\}_{i=1}^k $ and the trapdoor $\td = \{\td_i\}_{i=1}^k$. 
    \end{itemize}

    \item $\KeyGen$($\mpk$): 
    \begin{itemize}
    \item Take in $\mpk = \{f_{i,0}, f_{i,1}\}_{i=1}^k = \{(\bA_i, \bs_i\bA_i+\be_i)\}_{i=1}^k$ 
 .

    \item Prepare the quantum key of the form $ \rho_{\sk} = \bigotimes_{i=1}^k \left( \frac{1}{\sqrt{2}}(\ket{0,\bx_{i,0}}+\ket{1,\bx_{i,1}})\right)$ along with $\{\by_i\}_{i \in [k]}$ where $\by_i=f_{i,0}(\bx_{0,i})=f_{i,1}(\bx_{1,i})$ for each $i \in [k]$, according to the procedure in \Cref{sec:efficient_range_prepare_ntcf}.

    Note that $f_{i,b}(\bx_{b,i}) = \bx_{b,i} \bA_i+\be_{i}'+b_i \cdot \bs_i \bA_i,\; \be_{i}' \samp D_{Z_q^{m},B_P};  \bx_{i,b_i} \in [B_X]^n, \forall i \in [k]$.

    \item Output public key $\pk = \{(\bA_i, \bs_i \bA_i + \be_i, \by_i)\}_{i=1}^k$ and quantum decryption key $\qsk$. 

    \end{itemize}
    
    \item $\Enc(\pk,\mu)$: On input a public key $\pk = \{(\bA_i, \bs_i\bA_i + \be_i,\by_i)\}_{i=1}^k$ and a plaintext $\mu \in \{0,1\}$ the algorithm samples $\bR\samp\{0,1\}^{m \times m}$ and computes a ciphertext as follows: 

       \begin{align*}
        &\ct= \begin{bmatrix}  \bs_1\bA_1 +\be_1 \\ \bA_1 \\ \cdots
        \\ \cdots \\
       \bs_k\bA_k +\be_k \\  \bA_k 
        \\\sum_{i \in [k]} \by_i \end{bmatrix} \cdot \bR + \bE + \mu \cdot \bG_{(n+1)k+1}
    \end{align*}
     where $\bG_{(n+1)k+1}$ is the gadget matrix of dimensions $(nk+k+1)\times m$.

      $\bE \in \Z_q^{\left(nk+k+1 \right) \times m}$ is a matrix with all rows having $\mathbf{0}^m$ except the last row being $ \be''$, where $\be'' \gets \cD_{\Z_q^m, B_{P'}}$.

      Output $\ct$.




     \item \textsf{Dec$(\qsk, \ct)$}: On quantum decryption key $\rho_\sk = \bigotimes_{i=1}^k \left( \frac{1}{\sqrt{2}}(\ket{0, \bx_{i,0}}+\ket{1, \bx_{i,1}})\right)$ and a ciphertext $\ct \in \Z_q^{(nk+k+1)\times m}$, the decryption is performed in a coherent way as follows. 
    \begin{itemize}
         \item View the key as a vector of dimension $1 \times (n+1)k$; pad the key with one-bit of classical information, the value -1 at the end, to obtain a vector with dimension $1 \times (nk+k+1)$: \begin{align*}
           &  \begin{bmatrix}
             \frac{1}{\sqrt{2}}((\ket{0, \bx_{1,0})}+\ket{(1, \bx_{1,1})})| \cdots | \frac{1}{\sqrt{2}}(\ket{(0, \bx_{k,0})}+\ket{(1, \bx_{k,1}})) | \ket{-1} 
         \end{bmatrix} \\
         &=  \frac{1}{\sqrt{2^k}}\sum_{\mat{b}_j \in \{0,1\}^k}\lvert\bb_{j,1},\bx_{1,\bb_{j,1}}, \cdots, \bb_{j,k},\bx_{k,\bb_{j,k}} \rangle \ket{-1}
         \end{align*} 
         Denote this above vector by $\mathbf{sk}$.

        \item Compute $(-\mathbf{sk}) \cdot \ct \cdot \bG^{-1}  (\mathbf{0}^{nk+k} \vert \lfloor\frac{q}{2}\rfloor)$ coherently and write the result on an additional empty register $\mathsf{work}$ (a working register other than the one holding $\mathbf{sk}$).

        \item Make the following computation in another additional register $\mathsf{out}$: write $\mu' = 0$ if the outcome in the previous register of the above computation is less than $\frac{q}{4}$; write $\mu' = 1$ otherwise. 
        Uncompute the $\mathsf{work}$ register in the previous step using $\mathbf{sk}$ and $\ct$.
         Measure the final $\mathsf{out}$ register and output the measurement outcome.  

     \end{itemize}

      \item \textsf{Delete}($\rho_\sk$): 
      \begin{itemize}
      \item For convenience, we name the register holding state $\frac{1}{\sqrt{2}}(\ket{0, \bx_{i,0}}+\ket{1, \bx_{i,1}})$ in $\rho_\sk$ as register $\mathsf{reg_i}$.
          \item For each register $\mathsf{reg_i}, i \in [k]$: apply invertible function $\cJ: \mathcal{X} \to \{0,1\}^w$  where $\cJ(x)$ returns the binary decomposition of $x$. Since it is invertible, we can uncompute the original $\bx_{i, b_i}$'s in $\mathsf{reg_i}$ and leave only $\ket{0, \cJ(\bx_{i,0})} + \ket{1, \cJ(\bx_{i,1})}$. 

        \item Apply a quantum Fourier transform over $\mathbb{F}_2$ to all $k(w+1)$ qubits in registers $\{\mathsf{reg}_i\}_{i \in [k]}$ 
          \item  measure in computational basis to obtain a string $(c_1,\bd_1, \cdots, c_k, \bd_k) \in \{0,1\}^{wk+k}$. 
      \end{itemize}

     \item \textsf{VerDel}($\td, \pk,\cert$): On input a deletion certificate $\cert \{c_i, \bd_i\}_{i=1}^n$, public key $\pk = \{f_{i,0}, f_{i,1}, \by_i\}_{i \in k}$ and the trapdoor $\td = \{\td_i\}_{i \in k}$.
     \begin{itemize}
         \item    Compute $\bx_{i,b_i} \leftarrow \mathsf{INV}(\td_i,  b, \by_i)$ for both $b = 0,1$.

        \item Check if $\bx_{i,b_i} \in [B_X]^n$ for all $i \in [k], b_i \in \{0,1\}$. If not, output $\invalid$.
        If yes, continue.

        \item Check if $\lVert \by_i - \bx_{i,b_i} -b_i\cdot \bs_i\bA_i \rVert \leq B_P\sqrt{m}$,
        for all $i \in [k], b_i \in \{0,1\}$. 
        If not, output $\invalid$.
        If yes, continue.
        
         \item  
     Output $\valid$ if $c_i=\bd_i \cdot(\cJ(\bx_{i,0})\oplus \cJ(\bx_{i,1})(\mod 2)), \forall i \in [k]$ and $\invalid$ otherwise.
     \end{itemize}
     
    
\end{itemize}

\ifllncs
We defer the proof on correctness to \Cref{sec:correctness} and the security proof to \Cref{sec:main_security_proof}.
\else 
\fi

\bibliographystyle{alpha}
\bibliography{bib,abbrev3,custom,crypto} 

\appendix

\section{Preliminaries}

\subsection{Quantum Information an Computation}
\label{appendix:quantum_info}

We refer the reader to \cite{nc02} for a reference of basic quantum information and computation
concepts.

A quantum system $Q$ is defined over a finite set $B$ of classical states. In this work we will consider $B = \{0,1\}^n$. A \textbf{pure state} over $Q$ is a unit vector in $\mathbb{C}^{|B|}$, which assigns a complex number to each element in $B$. In other words, let $|\phi\rangle$ be a pure state in  $Q$, we can write $|\phi\rangle$ as:
\begin{equation*}
    |\phi\rangle = \sum_{x \in B} \alpha_x |x\rangle
\end{equation*}
where $\sum_{x \in B} |\alpha_x|^2 = 1$ and $\{|x\rangle\}_{x \in B}$ is called 
the ``\textbf{computational basis}'' of $\mathbb{C}^{|B|}$. The computational basis forms an orthonormal basis of $\mathbb{C}^{|B|}$.

A mixed state is a probability distribution over pure states. The most-often used distance between two possibly mixed states $\rho,\sigma$ is trace distance: $\lVert \rho - \sigma \rVert_{\Tr} = \frac{1}{2} \Tr|\rho-\sigma|$.

\subsection{Quantum Measurements}
\label{sec:prelim_quantum_measurements}
In this work, 
we mainly use two formalisms for quantum measurements. The first, a \emph{positive operator valued measure} (POVM), is a general form of quantum measurement.

A POVM $\cM$ is
specified by a finite index set I and a set $\{M_i\}_{i \in \cI}$ of hermitian positive semidefinite matrices $M_i$ with
the normalization requirement $\sum_{i \in \cI} M_i = \bI$. 
applying a POVM $\cM$ to a quantum state $\ket{\psi}$, the result of the measurement is $i$ with probability
$p_i = \langle \psi \vert M_i
\vert \psi \rangle$. The normalization requirements for $\cM$ and $\ket{\psi}$ imply that $
\sum_i p_i = 1$, and therefore
this is indeed a probability distribution. We denote by $\cM(\ket{\psi})$ the distribution obtained by applying
$\cM$ to $\ket{\psi}$.

We note that any quantum measurement $\cE = \{E_i\}_{i \in \cI}$ is associated with a POVM $\cM = POVM(\cE)$ with
$M_i = E_i^\dagger E_i$
. We will call $\cE$ an implementation of $\cM$. Note that each POVM may be implemented by many possible quantum
measurements, but each quantum measurement implements exactly one POVM.

A projective measurement is a quantum measurement where the $E_i$ are projections: $E_i$ are
hermitian and satisfy $E_i^2
 = E_i$. Note that $\sum_i E_i = \sum_i E_i^\dagger E_i
 = I$ implies that $E_i E_j = 0$ for $i \neq j$.
A projective POVM is a POVM where $M_i$ are projections.

\jiahui{add trace distance def}

    \begin{lemma}[Almost As Good As New(Gentle Measurement) Lemma \cite{aaronson2004limitations}] 
\label{lem:gentle_measure}
Suppose a binary-outcome POVM measurement $(\cP, \cQ)$ on a mixed state $\rho$ yields a
particular outcome with probability 
$1-\epsilon$.  Then after
the measurement, one can recover a state $\tilde{\rho}$ such that $ \left\lVert \tilde{\rho} - \rho \right\rVert_{\mathrm{tr}} \leq \sqrt{\epsilon}$.
Here $\lVert\cdot\rVert_{\mathrm{tr}}$ is the trace distance. 
\end{lemma}

\ifllncs

\subsection{Proof for \Cref{thm:invariant_ati_uniform}: Identity-like ATI}

In this section, we prove the property of ATI when using under certain "dummy" mixture of projections(POVM), will act  like an identity operator on the input quantum state, namely \Cref{thm:invariant_ati_uniform}.

\label{sec:invariant_ati_proof_append}

\else\fi

\section{Lattice Preliminaries}
\label{sec:latticeprelims}
\newcommand{\cL}{\mathcal{L}}
\newcommand{\R}{\mathbb{R}}
\renewcommand{\vec}[1]{\mathbf{#1}}
In this section, we recall some of the notations and concepts about lattice complexity problems and lattice-based cryptography that will be useful to our main result.
\subsection{General definitions}
A \emph{lattice} $\cL$ is a discrete subgroup of $\R^{m}$,
or equivalently the set \[\cL(\vec{b}_{1},\dots,\vec{b}_{n})=\left\{ \sum_{i=1}^{n}x_{i}\vec{b}_{i} ~:~ x_{i}\in\Z\right\} \]
of all integer combinations of $n$ linearly independent vectors
$\vec{b}_{1},\dots,\vec{b}_{n} \in \R^{m}$. Such $\vec{b}_i$'s form a \emph{basis} of $\cL$.

The lattice $\cL$ is said to be \emph{full-rank} if $n=m$.
We denote by $\lambda_{1}(\cL)$ the first minimum of $\cL$, defined as the length of a shortest non-zero vector of $\cL$.

\newcommand{\boundeddgaussian}{\mathcal{D}^{bounded}}
\newcommand{\dgaussian}{\mathcal{D}}

\paragraph{Discrete Gaussian and Related Distributions}
For any $s>0$, define \[\rho_s(\vec{x})=\exp(-\pi\|\vec{x}\|^2/s^2)\]for all $\vec{x}\in\R^n$.
We write $\rho$ for $\rho_1$. For a discrete set $S$, we extend $\rho$ to sets by
$\rho_s(S)=\sum_{\vec{x}\in S}\rho_s(\vec{x})$. Given a lattice $\cL$,
the \emph{discrete Gaussian} $\dgaussian_{\cL,s}$ is the distribution over $\cL$
such that the probability of a vector $\vec{y}\in\cL$ is proportional to $\rho_s(\vec{y})$:
\[
  \Pr_{X\leftarrow \dgaussian_{\cL,s}}[X=\vec{y}]=\frac{\rho_s(\vec{y})}{\rho_s(\cL)}.
\]

Using standard subgaussian tail-bounds, one can show we can show the following claim.

\begin{claim}
\label{claim:bounded}
Let $m\in \N$, $\sigma >0$, then it holds that: 
$$\Pr_{\vec e \leftarrow \dgaussian_{\Z^{m},\sigma}} [\Vert \vec e \Vert > m\sigma] <\exp(-\tilde{\Omega}(m)).$$
\end{claim}
The following is Theorem 4.1 from~\cite{GPV08} that shows an efficient algorithm to sample from the discrete Gaussian. 

Now we recall the Learning with Errors (LWE) distribution.

\begin{definition}
Let $\sigma = \sigma(n) \in (0,1)$. For $\vec{s} \in \Z_p^n$, the LWE distribution $A_{\vec{s}, p ,\sigma}$ over $\Z_p^n \times \Z_p$ is sampled by independently choosing $\vec{a}$ uniformly at random from $\Z_p^n$, and $e \leftarrow \dgaussian_{\Z,\sigma}$, and outputting $\left(\vec{a}, \langle \vec{a}, \vec{s}\rangle + e \mod p\right)$. 
The LWE assumption $\mathsf{LWE}_{n,m,\sigma,p}$ states that $m$ samples from $A_{\vec{s}, p ,\sigma}$ for a randomly chosen $\vec s\leftarrow \Z^{n}_{p}$ are indistinguishable from $m$ random vectors $\Z^{n+1}_{p}$.
\end{definition}

\subsection{Trapdoor Sampling for LWE}
We will need the following definition of a lattice trapdoor~\cite{GPV08,C:MicPei13,Vinodlecturetrapdoor}. For $\mathbf{A} \in \Z_q^{n \times m}$, we define the rank $m$ lattice
\[
\cL^{\bot}(\mathbf{A}) =\{ \vec{z} \in \Z^m \: :\: \vec{A} \vec{z} = \vec{0} \pmod q \} \;. 
\]
A lattice trapdoor for $\vec{A}$ is a set of short linearly independent vectors in $\cL^\bot(\vec{A})$.

\begin{definition}
A matrix $\vec{T} \in \Z^{m \times m}$ is a $\beta$-good lattice trapdoor for a matrix $\vec{A} \in \Z_q^{n\times m}$ if
\begin{enumerate}
    \item $\vec{A} \vec{T} = \vec{0} \pmod q$.
    \item For each column vector $\vec{t}_i$ of $\vec{T}$, $\|\vec{t}_i\|_\infty \le \beta$.
    \item $\vec{T}$ has rank $m$ over $\R$. 
\end{enumerate}
\end{definition}

\begin{theorem}\cite{GPV08,C:MicPei13}
\label{thm:trapdoor_samp}
There is an efficient algorithm $\mathsf{GenTrap}$ that, on input $1^n, q, m= \Omega(n \log q)$, outputs a matrix $\vec{A}$ distributed statistically close to uniformly on $\Z_q^{n \times m}$, and a $O(m)$-good lattice trapdoor $\vec{T}$ for $\vec{A}$.

Moreover, there is an efficient algorithm
$\mathsf{INVERT}$ that, on input $(\bA, \mathbf{T})$ and $\bs^\top\bA + \be^\top$ where $\lVert \be \rVert \leq q/(C_T\sqrt{n \log q})$
 and $C_T$ is a universal constant, returns $\bs$
and $\be$ with overwhelming probability over $(\bA, \mathbf{T}) \gets \mathsf{GenTrap}(1^n,1^m, q)$.
\end{theorem}

\begin{lemma}
\label{lem:noise_flooding}
   [Noise Smudging, \cite{dodis2010public}
   Let $y, \sigma > 0$. Then, the statistical distance between the
distribution $\cD_{\Z, \sigma}$ and $\cD_{\Z,\sigma+y}$ is at most $y/\sigma$.
\end{lemma}

\begin{lemma}(Leftover Hash Lemma(\textsf{SIS} version))
    \label{lem:LHL}

    The leftover hash lemma says that if $m = \Omega(n \log q)$, then if you sample $\bA \gets \Z_q^{n \times m}, \bx \gets \{0,1\}^m$ and $\by \gets \Z^m_q$, we have:
    $$ (\bA, \bA\cdot \bx) \approx_{q^{-n}} (\bA, \by) )$$
\end{lemma}

\jiahui{add noise flooding + statistical uniformity for Gentrap}

\subsection{Structural Properties of GSW Homomorphic Encryption}
Our scheme will build upon the GSW levelled FHE construction \cite{C:GenSahWat13}. We recall its structure here, for more details refer to \cite{C:GenSahWat13}. In GSW scheme, the public keys consists of matrix $\mat{B}\in \Z^{n\times m}_{q}$ where $m=\Omega(n \log q)$. This matrix is pseudorandom due to LWE. The secret key is a vector $\mat{s}\in \Z^{1\times n}_{q}$ such that $\mat{s}\cdot \mat{B}=\mat{e}$ for a small norm error vector $\mat{e}$. Such matrices can be constructed easily using the LWE assumption. 

For any such matrix $\mat{B}$, to encrypt a bit $\mu \in \{0,1\}$, the encryption algorithm produces a ciphertext $\ct \in \Z^{n\times \lceil n\log q \rceil}_{q}$ as follows. Sample a random matrix $\mat{R}$ to be a random small norm (where for instance each entry is chosen independently from $\{+1,-1\})$. Then we compute $\ct=\mat{B}\mat{R}+\mu \mat{G}$ where $\mat{G}$ is the Gadget matrix \cite{C:MicPei13}: $\bG= [\bI \otimes  [ 2^0, 2^1, \cdots, 2^{\lceil \log q \rceil-1}] \Vert \mathbf{0}^{m \times (m-n\lceil \log q \rceil)} ]  \in \Z^{n \times m}$ and $\bI$ is the $n \times n$ identity. $\bG$ converts
a binary representation of a vector back to its original vector representation over the field $\Z_q$; the associated (non-linear) inverse operation $\bG^{-1}$
converts vectors in $Z_q$
to their binary
representation.


Note that $\mat{B}\mat{R}$ is once again an LWE matrix in the same secret $\mat{s}$ and satisfies $\mat{s}\mat{B}\mat{R}$ is low norm therefore $\mat{s}\ct\approx \mat{s}\mu\mat{G}$ which could be used to learn $\mu$.

One can compute NAND operation as follows. Say we have $\ct_1=\mat{B}\mat{R}_1+\mu_1\mat{G}$ and $\ct_2=\mat{B}\mat{R}_2+\mu_2\mat{G}$. One can compute the bit-decomposition of $\ct_2$, $\mat{G}^{-1}(\ct_2)$ that simply expands out $\ct_2$ component wise by replacing every coordinate of $\ct_2$ by a its binary decomposition vector of size $\lceil \log_2 q\rceil$. Observe that $\mat{G}^{-1}(\ct_2)$ is a low norm matrix in $\Z^{\lceil n\log q \rceil\times \lceil n\log q \rceil}_{p}$ Then, one can compute $\ct_{\times}=\ct_1\mat{G}^{-1}(\ct_2)$. This yields a ciphertext $\ct_{\times} = \mat{B}(\mat{R}_1\mat{G}^{-1}(\ct_2)+\mu_1\mat{R}_2)+\mu_{1}\mu_2\mat{G}$. Observe that the randomness $(\mat{R}_1\mat{G}^{-1}(\ct_2)+\mu_1\mat{R}_2)$ is small norm and the invariant is still maintained. Finally to compute NAND operation one simply outputs $\ct_{\mathsf{NAND}}=\mat{G}-\ct_{\times}$ yielding a ciphertext encryption $1-\mu_1\mu_2$.

The security proof follows form the fact that by LHL, $\mat{B}\mat{R}$ is random provided $\mat{B}$ is chosen at random. By the security of LWE, $\mat{B}$ is pseudorandom therefore $\mat{B}\mat{R}$ is also pseudorandom.

\section{Preliminaries: Noisy Claw Free Trapdoor Families}

\subsection{Noisy Trapdoor Claw-Free Families}

\label{sec:NTCF_prelim}

The following definition of NTCF families is taken verbatim from \cite[Definition 6]{brakerski2021cryptographic}.
For a more detailed exposition of the definition, we refer the readers to the prior work.

\begin{definition}[NTCF family]\label{def:trapdoorclawfree}
Let $\lambda$ be a security parameter. Let $\sX$ and $\sY$ be finite sets.
 Let $\mathcal{K}_{\mathcal{F}}$ be a finite set of keys. A family of functions 
$$\mathcal{F} = \big\{f_{k,b} : \sX\rightarrow \mathcal{D}_{\sY} \big\}_{k\in \mathcal{K}_{\mathcal{F}},b\in\{0,1\}}$$
is called a \emph{noisy trapdoor claw free (NTCF) family} if the following conditions hold:

\begin{enumerate}
\item{\textbf{Efficient Function Generation.}} There exists an efficient probabilistic algorithm $\textrm{GEN}_{\mathcal{F}}$ which generates a key $k\in \mathcal{K}_{\mathcal{F}}$ together with a trapdoor $t_k$: 
$$(k,t_k) \leftarrow \textrm{GEN}_{\mathcal{F}}(1^\lambda).$$
\item{\textbf{Trapdoor Injective Pair.}} For all keys $k\in \mathcal{K}_{\mathcal{F}}$ the following conditions hold. 
\begin{enumerate}
\item \textit{Trapdoor}: There exists an efficient deterministic algorithm $\textrm{INV}_{\mathcal{F}}$ such that for all $b\in \{0,1\}$,  $x\in \sX$ and $y\in \supp(f_{k,b}(x))$, $\textrm{INV}_{\mathcal{F}}(t_k,b,y) = x$. Note that this implies that for all $b\in\{0,1\}$ and $x\neq x' \in \sX$, $\supp(f_{k,b}(x))\cap \supp(f_{k,b}(x')) = \emptyset$. 
\item \textit{Injective pair}: There exists a perfect matching $\sR_k \subseteq \sX \times \sX$ such that $f_{k,0}(x_0) = f_{k,1}(x_1)$ if and only if $(x_0,x_1)\in \sR_k$. \end{enumerate}

\item{\textbf{Efficient Range Superposition.}}
For all keys $k\in \mathcal{K}_{\mathcal{F}}$ and $b\in \{0,1\}$ there exists a function $f'_{k,b}:\sX\to \mathcal{D}_{\sY}$ such that the following hold.
\begin{enumerate} 
\item For all $(x_0,x_1)\in \mathcal{R}_k$ and $y\in \supp(f'_{k,b}(x_b))$, INV$_{\mathcal{F}}(t_k,b,y) = x_b$ and INV$_{\mathcal{F}}(t_k,b\oplus 1,y) = x_{b\oplus 1}$. 
\item There exists an efficient deterministic procedure CHK$_{\mathcal{F}}$ that, on input $k$, $b\in \{0,1\}$, $x\in \sX$ and $y\in \sY$, returns $1$ if  $y\in \supp(f'_{k,b}(x))$ and $0$ otherwise. Note that CHK$_{\mathcal{F}}$ is not provided the trapdoor $t_k$.
\item For every $k$ and $b\in\{0,1\}$,
$$ \E_{x\leftarrow_U \sX} \big[H^2(f_{k,b}(x),f'_{k,b}(x))\big] \leq \mu(\lambda),$$
 for some negligible function $\mu(\cdot)$. Here $H^2$ is the Hellinger distance. Moreover, there exists an efficient procedure  SAMP$_{\mathcal{F}}$ that on input $k$ and $b\in\{0,1\}$ prepares the state
\begin{equation*}
    \frac{1}{\sqrt{|\sX|}}\sum_{x\in \sX,y\in \sY}\sqrt{(f'_{k,b}(x))(y)}\ket{x}\ket{y}.
\end{equation*}

\end{enumerate}

\item{\textbf{Adaptive Hardcore Bit.}}
For all keys $k\in \mathcal{K}_{\mathcal{F}}$ the following conditions hold, for some integer $w$ that is a polynomially bounded function of $\lambda$. 
\begin{enumerate}
\item For all $b\in \{0,1\}$ and $x\in \sX$, there exists a set $G_{k,b,x}\subseteq \{0,1\}^{w}$ such that $\Pr_{d\leftarrow_U \{0,1\}^w}[d \notin G_{k,b,x}]$ is negligible, and moreover there exists an efficient algorithm that checks for membership in $G_{k,b,x}$ given $k,b,x$ and the trapdoor $t_k$. 
\item There is an efficiently computable injection $\mathcal{J}:\sX\to \{0,1\}^w$, such that $\mathcal{J}$ can be inverted efficiently on its range, and such that the following holds. If
\begin{align*}\label{eq:defsetsH}
H_k &= \big\{(b,x_b,d,d\cdot(\mathcal{J}(x_0)\oplus \mathcal{J}(x_1)))\,|\; b\in \{0,1\}, (x_0,x_1)\in \mathcal{R}_k, d\in G_{k,0,x_0}\cap G_{k,1,x_1}\big\},
\\
\overline{H}_k &= \{(b,x_b,d,c)\,|\; (b,x,d,c\oplus 1) \in H_k\big\},
\end{align*}
then for any quantum polynomial-time procedure $\mathcal{A}$ there exists a negligible function $\mu(\cdot)$ such that 
\begin{equation*}\label{eq:adaptive-hardcore}
\left|\Pr_{(k,t_k)\leftarrow \textrm{GEN}_{\mathcal{F}}(1^{\lambda})}[\mathcal{A}(k) \in H_k] - \Pr_{(k,t_k)\leftarrow \textrm{GEN}_{\mathcal{F}}(1^{\lambda})}[\mathcal{A}(k) \in\overline{H}_k]\right| \leq \mu(\lambda).
\end{equation*}
\end{enumerate}

\end{enumerate}
\end{definition}

\subsection{(Extended) NTCF from LWE}
\label{sec:ntcf_from_lwe}
\begin{theorem}[{\cite[Theorem 4.1]{brakerski2021cryptographic} \cite[Theorem 9.2]{mahadev2018classical}}]
  \label{thm:lwe-ntcf}
  Assuming the post-quantum hardness of $\lwe_{n,m,q,B_L}$, (extended) NTCF families exist.
\end{theorem}

The following construction description is mostly taken verbatim from \cite{mahadev2018classical}.
\subsubsection{Parameter Choice}
\label{sec:ntcf_parameter}

Let $\lambda$ be the security parameter. All other parameters are
functions of $\lambda$. Let $q \geq 2$ be a prime integer. Let $\ell, n, m, w \geq 1$ be polynomially bounded functions of $\lambda$ and
$B_L, B_V , B_P$ be Gaussian parameters, $B_X, B_S$ be norm bounds, such that the following conditions hold:
\begin{enumerate}
    \item $n = \Omega(\ell \log q)$ and $m = \Omega(n \log q)$

    \item $w = n \lceil \log q \rceil$.

    \item $B_P = \frac{q}{2C_T
\sqrt{mn \log q}}$ where $C_T$ is the constant in \cref{thm:trapdoor_samp}.

\item $2\sqrt{n} \leq B_L \leq B_V \leq B_P \leq B_X$

\item The ratios $\frac{B_V}{B_L}$, $\frac{B_P}{B_V}$, $\frac{B_{P'}}{B_P}$, $\frac{B_X}{B_S}$
 are all super-polynomial in $\lambda$.

\item We denote $[B_X]$ as all integer taking values $[-B_X, -B_X+1 \cdots, B_X-1, B_X]$. Similarly for $B_S$. $B_S$ is in fact $\{0,1\}$ in our setting.
\end{enumerate}

\subsubsection{Noisy Trapdoor claw-Free Families (2-to-1 Mode)}
Let Let $\sX = \Z^n_q$
and $\sY = \Z^m_q$
. The key space is $\Z^{n\times m}_q \times \Z^m_q$.
 For $b \in \{0, 1\}$, $\bx \in \sX$ and key $k = (\bA, \bs\bA +\be )$ where $(\bA, \td_A) \gets \mathsf{GenTrap}(1^n, 1^m, q)$, $\bs\gets \Z_q^n, \be\gets \cD_{\Z_q^m, B_L}$, the density $f_{k,b}(x)$ is defined as:
 \begin{align*}
    \forall \by \in \sY: (f_{k,b}(\bx))(\by) = \cD_{\Z^m_q, B_P}(\by -\bx^\top\bA - b\cdot \bs^\top\bA)
 \end{align*}
It follows that:
\begin{align*}
   & \mathrm{SUPP}(f_{0,k}(\bx)) = \{\by = \bx\bA + \be' \vert \lVert \be' \rVert \leq B_P\sqrt{m} \} \\
   & \mathrm{SUPP}(f_{1,k}(\bx)) = \{\by = \bx\bA + \be' + \bs\bA  \vert \lVert \be' \rVert \leq B_P\sqrt{m} \} \\
\end{align*}
$(\bx_0, \bx_1)$ will be an injective pair such that $f_{0,k}(\bx_0) = f_{1,k}(\bx_1)$ if and only if $\bx_0 = \bx_1 +\bs$.

\subsubsection{Noisy Trapdoor Injective Families (Injective Mode)}
\label{sec:ntcf_injective}
We now describe the trapdoor injective functions, or informally, the "injective mode" of trapdoor claw-free functions.
Let Let $\sX = \Z^n_q$
and $\sY = \Z^m_q$. The key space  is $\Z^{n\times m}_q \times \Z^m_q$.
 
 For $b \in \{0, 1\}$, $\bx \in \sX$ and key $k = (\bA, \bu )$, where $(\bA, \td_A) \gets \mathsf{GenTrap}(1^n, 1^m, q)$, 
  $\bu$ is sampled uniformly random up to the condition that there \emph{does not} exist $\bs,\be$ such that $\bs^\top\bA + \be^\top = \bu$ and $\lVert \be \rVert \leq \frac{q}{C_T\sqrt{n\log q}}$, which happens with all but negligible probability.

 
 The density $g_{k,b}(x)$ is defined as:
 \begin{align*}
    \forall \by \in \sY: (g_{k,b}(\bx))(\by) = \cD_{\Z^m_q, B_P}(\by -\bx\bA - b\cdot \bu)
 \end{align*}

The injective trapdoor functions $g_{b,k}$ looks like follows:
\begin{align*}
   & \mathrm{SUPP}(g_{0,k}(\bx)) = \{\by = \bx^\top\bA + \be' \vert \lVert \be' \rVert \leq B_P\sqrt{m} \} \\
   & \mathrm{SUPP}(g_{1,k}(\bx)) = \{\by = \bx\bA + \be' + \bu \vert \lVert \be' \rVert \leq B_P\sqrt{m} \} \\
\end{align*}
$\mathrm{SUPP}(g_{0,k}(\bx))$ and $\mathrm{SUPP}(g_{1,k}(\bx))$ are disjoint as long as $B_P \leq \frac{q}{2C_T \sqrt{mn\log q}}$.

There is also an inversion function $\mathsf{INV}$ in the injective mode: $\mathsf{INV}(\td, \by) \to \bx$ will give some unique $\bx$ on input $(\td, \by)$.

\begin{lemma}\cite{mahadev2018classical}
\label{lem:twomodes}
The 2-to-1 mode and injective mode  are computationally indistinguishable by $\lwe_{n,m,q,B_L}$.
\end{lemma}

\subsubsection{Efficient Range Preparation for NTCF with Small LWE Secrets}
\label{sec:efficient_range_prepare_ntcf}
We refer the reader to Section 4.3 of \cite{brakerski2021cryptographic} for a detailed description of $\mathrm{SAMP}_\lwe$ procedure to efficiently prepapre the claw state:  $ (\ket{0, \bx_{0}}+ \ket{1, \bx_{1}}),y$.
We describe it here briefly for the sake of coherence in presentation.

In order for our security parameters in the reduction to go through, we deviate slightly from the exact parameters on $\sX$ and $\bs$ in \cite{brakerski2021cryptographic}. We work with small secrets $\sX = [B_X]^n$. The range for $\bs$ is $[B_S]^n = \{0,1\}^n$ is the same as \cite{brakerski2021cryptographic}. 

 At the first step, the procedure
creates the following superposition: 
$$\sum_{\be'\in \Z_q^m} \sqrt{\cD_{\Z_q, B_P}(\be')}\ket{\be'} $$

In step 2 of $\mathrm{SAMP}_\lwe$, we prepare uniform superposition of $\bx \in [B_X]$. 
$$\frac{1}{\sqrt{2|B_X|^n}} \sum_{ b \in \{0,1\}, \bx, \be' }\sqrt{\cD_{\Z_q, B_P}(\be')}\ket{b,\bx}\ket{\be'} 
$$

The rest of the steps are the same.
In step 3, we apply the key $(\bA, \bu)$, controlled by the bit $b$ to indicate whether to use $\bu$, where $\bu =\bs\bA+\be$ in the 2-to-1 setting and $\bu \samp \Z_q^m$ in the injective mode.  
We get a state: 
$$ \frac{1}{\sqrt{2|B_X|^n}} \sum_{ b \in \{0,1\}, \bx, \be' }\sqrt{\cD_{\Z_q, B_P}(\be')}\ket{b,\bx}\ket{\be'}\ket{ \bx \bA+\be'+ b \cdot \bu} $$
Next, we measure the last register to obtain $\by = \bx_b \bA + \be' + \bb \cdot \bs\bA$. We can then uncompute the register containing $\be'$ using the information in register containing $b,\bx$, the key $(\bA, \bu)$ and the last register.

It is easy to observe that the efficient range preparation in \cite{brakerski2021cryptographic} Section 4.3 and acquirance of the claw state also works in our setting with our choice of parameters having $B_X/B_S$ superpolynomially large, more specifically $B_S = \{0,1\}$ and $B_X$ a quasipolynomial.

With probability $(1-\negl(\lambda))$, when one measures the image register to obtain a value $\by$, we will obtain the state $\frac{1}{\sqrt{2}} (\ket{0, \bx_{0}}+ \ket{1, \bx_{1}})$
where $f_{0,k}(\bx_0) = f_{1,k}(\bx_1)= \by$


\subsection{Parallel Repetition of An NTCF-based Protocol}
\label{sec:parallel_ntcf}


We first define the single-instance game from \cite{radian2019semi}. The game is abstracted as a "1-of-2" puzzle with "2-of-2 soundness", where the verifier randomly asks the prover to output a preimage $\bx \in \sX$ or an adaptive hardcore bit for the same image $\by \in \sY$. 

\begin{definition}[1-of-2 Puzzle from NTCF \cite{radian2019semi}]
The protocol proceeds as follows, using the notations from \Cref{sec:NTCF_prelim}. 
    \begin{itemize}
    \item The verifier samples a key $(k, \td)\gets \mathsf{GEN}_{\mathcal{F}}(1^\lambda)$ and send $k$ to the prover. The verifier keeps the trapdoors $\td$ 

    \item The prover sends back a committed image value $\by$.

    \item The verifier samples a random bit $\delta \in \{0,1\}$ and sends $\delta$ to the prover.

    \item If $\delta = 0$, the prover sends back some $\bx \in \sX$; else if $b = 1$, the prover sends back a string $(c,\bd)$.
    \item The verifier does the following checks on each  $(\by, \bx)$ or $(\by, c, \bd)$:
    \begin{itemize}
        \item When $\delta = 0$: Check $\bx \in \mathsf{INV}(\td, b \in \{0,1\}, \by)$ \footnote{This step can also be performed publicly using $CHK_{\mathcal{F}}$.}.

        \item When $\delta = 1$: Find both $\bx_{0}, \bx_{1}$ using $\mathsf{INV}(\td, b \in \{0,1\}, \by)$. Check if $c = \bd \cdot(\mathcal{J}(\bx_{0})\oplus \mathcal{J}(\bx_{1}))$.
    \end{itemize}
    \end{itemize}
\end{definition}

\cite{radian2019semi} showed the following property for the above protocol using the LWE-based NTCF from \cite{brakerski2021cryptographic}. 

\textbf{1-of-2 Completeness:} Any BQP prover will answer one of the challenges for $\delta = 0$ or $\delta = 1$ with probability 1.

\textbf{2-of-2 Soundness:}  The 2-of-2 soundness error in the above protocol is the probability that a prover can provide both the 1-challenge answer $\bx$ and the 0-challenge answer $(c,\bd)$ correctly.

The above protocol has 2-of-2 soundness $1/2$ for any BQP prover \cite{radian2019semi,brakerski2021cryptographic}.

\paragraph{Parallel Repitition}
We now describe a special type of parallel-repeated protocol based on the NTCF. In this protocol, we only consider the "2-of-2" setting: the verifier samples multiple keys independently; for every single key, the verifier simply asks the prover to provide both the answer to the 0-challenge and the answer to the 1-challenge. 

Its parallel repetition soundness was shown in \cite{radian2019semi}.

\begin{definition}[Parallel-Repeated 2-of-2 NTCF-protocol)]
The protocol proceeds as follows, using the notations from \Cref{sec:NTCF_prelim}.
\begin{itemize}
    \item The verifier samples $\ell$ number of keys $(k_i, \td_i)\gets \mathsf{GEN}_{\mathcal{F}}(1^\lambda), i \in [\ell]$ independently and send $\{k_i\}_{i \in [\ell]}$ to the prover. The verifier keeps the trapdoors $\{\td_i\}_{i \in [\ell]}$ 

    \item The prover sends back $\ell$ tuple of values $\{(\by_i, \bx_i, c_i, \bd_i)\}_{i \in [\ell]}$.

    \item The verifier does the following checks on each  $(\by_i, \bx_i, c_i, \bd_i)$ for $i \in [\ell]$:
    \begin{itemize}

        \item Find both $\bx_{i,0}, \bx_{i,1}$ using $\mathsf{INV}(\td_i, b \in \{0,1\}, \by_i)$. \item Check if $c_i = \bd_i \cdot(\mathcal{J}(\bx_{i,0})\oplus \mathcal{J}(\bx_{i,1}))$.
    \end{itemize}

    \item If all the checks pass, the verifier outputs 1; else outputs 0.
\end{itemize}
\end{definition}

\begin{theorem}[Parallel Repetition Soundness of NTCF-based Protocol, \cite{radian2019semi} Theorem 15 rephrased]
\label{thm:parallel_soundness_2of2}
    The above protocol has soundness $(1- \negl(\ell))$ for any BQP prover. 
\end{theorem}

\begin{remark}
    Note that in our construction \Cref{sec:construction} we let the verifier further check if $\by_i, \bx_i$ are well-formed. This will not affect the above soundness because it only puts more restrictions on the prover/adversary.  
\end{remark}


\begin{remark}
    In this work, we only need the soundness to use the above simple protocol where we require the adversary to produce both the "preimage answer" and the "adaptive hardcore bit answer" at the same time.
    Clearly, the "completeness" of the above protocol is not well-defined, but we omit this property in our setting.
    
    We do not need the more complicated version of repetition in \cite{mahadev2018classical} studied in \cite{chia2020classical,alagic2016quantum}.
\end{remark}

\section{Security Proof for SKL-PKE}
\label{sec:main_security_proof}

\begin{theorem}[Security]
\label{thm:security}
    Assuming the post-quantum sub-expoential hardness of $\lwe_{n,m,q,B_{L}}$ with parameter choice in \Cref{sec:ntcf_parameter}, the construction in \Cref{sec:construction} satisfies the $\gamma$-strong SKL-PKE security define in \Cref{def:strong_skl} for any noticeable $\gamma$. 
\end{theorem}



To prove security we consider two hybrids. The first hybrid $\mathbf{Hybrid}_0$ corresponds to the real security game whereas $\mathbf{Hybrid}_1$ a correspond to modified games. We will show that these hybrids are indistinguishable and so the winning probability in the three hybrids are negligibly close and then in the final hybrid $\mathbf{Hybrid}_1$ it must be negiligible.

\begin{figure}

    \caption{Hybrid 0}
    \label{fig:hyb_0}

\begin{mdframed}
\begin{center}
{\bf Hybrid 0}
\end{center}~\\
In this hybrid, the adversary and the challenger play the security game as in Definition \ref{def:gamma_skl_game}.
\begin{enumerate}
\item The challenger runs $\Setup(1^\lambda) \to (\mpk, \td)$. It sends $\mpk$ to the adversary $\A$. $\A$ computes $(\pk, \rho_\sk) \gets \KeyGen(\mpk)$ and publishes $\pk$. 
   
    \item  The challenger requests that $\A$ runs the deletion algorithm $\delete(\qsk)$. $\A$ returns a deletion certificate $\cert$ to the challenger. 
    
    \item The challenger runs $\VerDel(\pk,\td,\cert)$ and continues if $\VerDel(\pk,\td,\cert)$ returns $\valid$; else it outputs $\bot$ and aborts, if $\VerDel(\pk,\td,\cert)$ returns $\invalid$.
    
    \item 
   $\A$ outputs a message $\mu$ and a (possibly mixed) state $\rho_\delete$ as a quantum decryptor.
    
    \item 
    
    The challenger runs the test for a $\gamma$-good decryptor on $\rho_\delete$
    with respect to $\pk$ and $\mu$ (using $\ti_{1/2+\gamma}(\cP_\cD)$). The challenger outputs $1$ if the test passes, otherwise outputs $0$.
\end{enumerate}

\end{mdframed}
\end{figure}

\begin{figure}
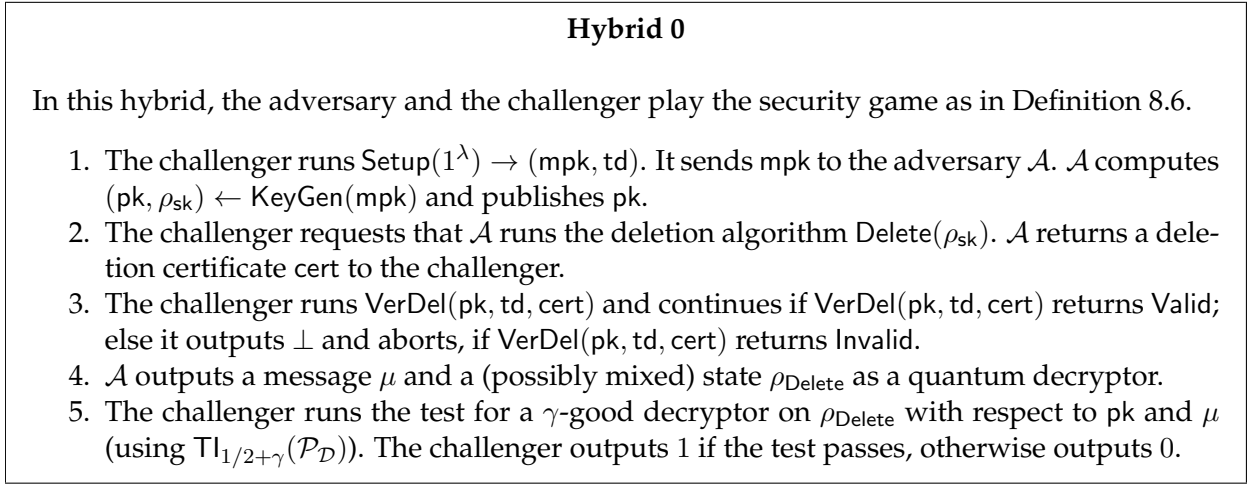

\caption{Hybrid 1}
\label{fig:hyb_2} 

\begin{mdframed}
\begin{center}
{\bf Hybrid 1}
\end{center}~\\
In this hybrid, we replace the check for $\gamma$ good decryptor with an efficient check $\ati^{\epsilon, \delta}_{\cP, D, \gamma}$ where we set $\delta$ and $\epsilon$ to be $\lambda^{-\omega(1)}$ for a tiny super-constant $\omega(1)$, e.g. we can set $\delta$ tp be exponentially small and $\epsilon = \frac{\gamma}{100B_xk n}$. 
\begin{enumerate}
\item The challenger runs $\Setup(1^\lambda) \to (\mpk, \td)$. It sends $\mpk$ to the adversary $\A$. $\A$ computes $(\pk, \rho_\sk) \gets \KeyGen(\mpk)$ and publishes $\pk$. 
    \item  The challenger requests that $\A$ runs the deletion algorithm $\delete(\qsk)$. $\A$ returns a deletion certificate $\cert$ to the challenger. 

      \item The challenger runs $\VerDel(\pk,\td,\cert)$ and if $\VerDel(\pk,\td,\cert)$ returns $\valid$ it outpus $z$; else it outputs $\bot$ if $\VerDel(\pk,\td,\cert)$ returns $\invalid$.
      
    \item 
   $\A$ outputs a message $\mu$ and a (possibly mixed) state $\rho_\delete$ as a quantum decryptor.
   
    \item 
    The challenger runs the test {\color{red}$\ati^{\epsilon, \delta}_{\cP, D, \gamma+1/2-\epsilon}(\rho_{\delete})$} with respect to $\mu$ and $\pk$. The challenger sets $z=1$ if the test passes, otherwise it sets $z=0$. 
    
\end{enumerate}
\end{mdframed}
\end{figure}


We will prove the following statements:
\begin{enumerate}

\item Probability of winning in $\mathbf{Hybrid}_0$ and $\mathbf{Hybrid}_1$ are close by a negligible amount if  $\delta$ are set to be $\lambda^{-\omega(1)}$ for a tiny super constant $\omega(1)$ (we can in fact set it to be exponentially small).
\item We will then prove that if LWE satisfies subexponential security, for the set parameters $\lwe_{n,m,q,B_L}$ probability of winning in $\mathbf{Hybrid}_1$ is negligible. 
\end{enumerate}


Together these claims imply that the probability of winning in $\mathbf{Hybrid}_0$ is negligible.


Claim 1 follows from \Cref{cor:ati_thresimp}: if the inefficient $\gamma$-good decryptor test outputs 1 with probability $p$ on a state $\rho$, then the efficient $\ati_{\cP,\cD, 1/2+\gamma-\epsilon}^{\epsilon,\delta}$ will output 1 on the state $\rho$
with probability $p - \delta$. Since $\delta$ is negligible, $\A$'s overall wining probability will have 
negligible difference.

\subsection{Winning Probability in Hybrid 1}
Next, we show that $\Pr[\mathbf{Hybrid}_1=1]\leq \negl(\lambda)$ for some negligible $\negl(\cdot)$. We will reduce to the security of the parallel repeated NTCF game in \Cref{thm:parallel_soundness_2of2}.

\begin{lemma}
    \label{lem:hyb2_negligible}
    Assuming post-quantum subexponential hardness of $\lwe_{m,m,q,B_V}$ with parameter choice in \Cref{sec:ntcf_parameter}, we have  $\Pr[\mathbf{Hybrid}_1=1]\leq \negl(\lambda)$ for some negligible $\negl(\cdot)$.
\end{lemma}

To show that the winning probability in \textbf{Hybrid 1} is negligible, 
we consider a world where we \textbf{do not check the deletion certificate and let the adversary pass all the time}.
In this world, through a sequence of hybrid games: we will call them \textbf{Games} instead of Hybrids to distinguish from the above Hybrids 0, 1.

Later, we will show how to put back the condition about the deletion certificate check for our analysis via the following argument:

\paragraph{Notations for Events}
For simplicity, we make a few notations for the events that take place in \textbf{Hybrid 1}:
\begin{itemize}
    \item We denote the event that the adversary hands in a valid deletion cerficate, i.e. $\VerDel(\pk,\td, \cert) =  \valid$, as $\mathsf{Cert Pass}$.

    \item We denote the event that
    test $\ati^{\epsilon, \delta}_{\cP, D, \gamma+1/2}(\rho_{\delete})$ outputs 1 with respect to $\mu$ and $\pk$, as $\mathsf{Good Decryptor}$. To simplify notations, we define the new $\gamma$ here to be the $\gamma-\epsilon$ in Hybrid 1.

    \item We denote $\Ext$ as the event where we can obtain the preimages $\{\bx_{i}\}_{i \in [k]} \in \{\mathrm{INV}(\td_i, b\in\{0,1\}, \by_i)\}_{i \in [k]}$(from the remaining state of measurement $\ati^{\epsilon, \delta}_{\cP, D, \gamma+1/2}(\rho_{\delete})$).
\end{itemize}

Suppose the probability that final output $1$ in Hybrid 1 ( \Cref{fig:hyb_2}) is some noticeable $\epsilon$, then we must have $\Pr[\mathsf{Cert Pass} \wedge \mathsf{Good Decryptor}] \geq \epsilon_1$ . To build a reduction that breaks the security of parallel repeated NTCF game \Cref{thm:parallel_soundness_2of2}, we need the following statement to hold: $\Pr[\mathsf{Cert Pass} \wedge \Ext] \geq \epsilon'$, for some noticeable $\epsilon'$, because in this case the reduction can obtain both the deletion certificates $\{c_i, \bd_i\}_{i \in [k]}$ and the preimages $\{\bx_{i, b}\}_{i \in [k]}$, which allow it to win the parallel repeated NTCF game.

Our proof outline is the follows: we would show that when  $\mathsf{Good Decryptor}$ happens, $\Ext$ \emph{always happens} except with negligible probability. Therefore, we have $\Pr[\mathsf{Cert Pass} \wedge \Ext] \geq \epsilon_1 - \negl(\lambda)$ by a simple proabability observaion (\Cref{claim:prob_relation}).

We analyze the probabilities by defining some games in the world where we don't check the deletion certificate and reasoning about them.

\paragraph{Game 0}
This is an experiment same as the one in \Cref{fig:hyb_2}, using the construction \Cref{sec:construction}, except that {\color{red}\emph{the challenger does not perform check on the deletion certificate, i.e. step 5 in \Cref{fig:hyb_2}}}.

\begin{enumerate}

    \item 
    The challenger runs $\Setup(1^\lambda)$: the challenger prepares $\mpk = \{\bA_i,\bs_i\bA_i+\be_i\}_{i \in [k]} $, where $(\bA_i, \td_i) \gets \mathsf{GenTrap}(1^n, 1^m, q), \forall i \in [k]$ and sends it to $\A$. The challenger keeps $\td = \{\td_i\}_{i \in [k]}$
    
    \item  $\A$ receives $\mpk$ and obtains the classical public key $\pk = \{\bA_i, \bs_i\bA_i+\be_i, \by_i\}_{i \in [k]} \gets \KeyGen(1^\lambda)$ and one copy of quantum decryption key $\qsk$. $\A$ publishes $\pk$.

    \item 
   $\A$ outputs a message $\mu$ and a (possibly mixed) state $\rho_\delete$ as a quantum decryptor.
    
    \item 
     The challenger runs the (efficient) test $\ati^{\epsilon, \delta}_{\cP, D, \gamma+1/2}(\rho_{\delete})$ with respect to $\mu$ and $\pk$. The challenger outputs $1$ if the test passes, otherwise it outputs $0$. 
     
\end{enumerate}

\paragraph{Game 1}

This is the same as Game 0 except that all $\bA_i$ are sampled uniformly at random, without a trapdoor.

\begin{enumerate}

    \item 
    The challenger runs $\Setup(1^\lambda)$: the challenger prepares $\mpk = \{\bA_i,\bs_i\bA_i+\be_i\}_{i \in [k]} $, where {\color{red}$\bA_i \gets \Z_q^{n \times m}, \forall i \in [k]$} and sends it to $\A$.
    
    \item  $\A$ receives $\mpk$ and obtains the classical public key $\pk = \{\bA_i, \bs_i\bA_i+\be_i, \by_i\}_{i \in [k]} \gets \KeyGen(1^\lambda)$ and one copy of quantum decryption key $\qsk$. $\A$ publishes $\pk$. 

   $\A$ outputs a message $\mu$ and a (possibly mixed) state $\rho_\delete$ as a quantum decryptor.
    
    \item 
     The challenger runs the (efficient) test $\ati^{\epsilon, \delta}_{\cP, D, \gamma+1/2}(\rho_{\delete})$ with respect to $\mu$ and $\pk$. The challenger outputs $1$ if the test passes, otherwise it outputs $0$. 
\end{enumerate}

\paragraph{Game 2.j: $j = 1, \cdots , k$} This is the same as Game 0 except the following:
\begin{enumerate}
    
    \item During $\Setup$,

    \begin{itemize}
    \item
    {\color{red} For $i \leq j$: the challenger prepares $\mpk_i = (\bA_i, \bu_i)$, where $\bA_i \gets \Z_q^{n \times m}$ and $\bu_i \gets \Z_q^{1 \times m}$ uniformly random.} 
        \item For $i > j$: the challenger prepares $\mpk_i = (\bA_i, \bu_i = \bs_i\bA_i+\be_i)$ the same as in hybrid 0.
    \end{itemize}

    \item $\A$ accordingly obtains public key {\color{red} $\pk = \{\bA_i, \bu_i, \by_i\}_{i \in [k]}$}.
and one copy of quantum decryption key $\qsk$. $\A$ publishes $\pk$. 

    \item 
   $\A$ outputs a message $\mu$ and a (possibly mixed) state $\rho_\delete$ as a quantum decryptor.
    
    \item 
     The challenger runs the (efficient) test $\ati^{\epsilon, \delta}_{\cP, D, \gamma+1/2}(\rho_{\delete})$ with respect to $\mu$ and $\pk$. The challenger outputs $1$ if the test passes, otherwise it outputs $0$.

\end{enumerate}

We then prove the following claims about the above games:

\begin{claim}
\label{claim:game0_1}
    Game 0 and Game 1 are statistically indistinguishable.
\end{claim}

This follows directly from the property of $\mathsf{GenTrap}$ in\Cref{thm:trapdoor_samp}. 

\begin{claim}
\label{claim:game1_2}
    Assuming the hardness of $\lwe_{n,m, q, B_L}$, Game 1 and Game 2.k are indistinguishable.
\end{claim}
\begin{proof}
We claim that each pair in (Game 1, Game 1.1), (Game 1.1, Game 1.2) ... (Game 1.(k-1), Game 1.k) is indistinguishable.
If anyone of them are distinguishable, then there exists some $j$ such that there is an adversary that distinguishes $(\bA_j, \bs_j \bA_j + \be_j)$ and $(\bA_j, \bu_j \gets \Z_q^m)$. 
\end{proof}


\begin{remark}
The above property can also directly follow from the indistinguishability between the 2-to-1 mode and injective mode of NTCF \cite{mahadev2018classical} (see \Cref{lem:twomodes}). 
Note that after switching to $(\bA_i,\bu_i)$, some of the $\by_i$'s in the public key $(\{\bA_i, \bu_i, \by_i\}$ may have the format $\by_i = \bx_{i,1}\bA_i + \be_i' + \bu_i$ or $\by_i = \bx_{i,0}\bA_i + \be_i'$.

Thus,
an honestly encrypted ciphertext $\ct$ in \textbf{Game $2.k$} for message $\mu$ should have the following format:
   \begin{align}
      \label{eqn:ct_game2}
        &\ct= \begin{bmatrix} \bu_1  \\\bA_1 \\ \\ \cdots 
        \\  \bu_k\\
        \bA_k 
        \\ \sum_{i \in [k]}(\bx_{i}\bA_i + \be_{i}' +  b_{i} \cdot \bu_i) \end{bmatrix} \cdot \bR + \bE + \mu \cdot \bG_{(n+1)k+1}
    \end{align}
    where $\bA_i \gets \Z_q^{n \times m}, \bu_i \gets \Z_q^m,\bR \gets \{0,1\}^{m \times m}$ and the $(nk+k+1)$-th row in $\bE$ is $ \be'' \gets \cD_{\Z_q^m, B_{P'}}$. Note that $b_i = 0 \text{ or }, 1, i \in [k]$ are some adversarially chosen bits that come in the $\{\by_i\}_{i \in [k]}$ part of $\pk$.    

\end{remark}

\paragraph{Switching to Plaintext 0}
From now on, without loss of generality, we always consider encrypting message $\mu= 0$. The analysis for the case when $\mu = 1$ should follow symmetrically.

\subsection{Extraction of Preimages via LWE Search-Decision Reduction}

Now we are ready to argue that in \textbf{Game $2.k$}, if the game outputs 1, then there exists an extractor that extracts all preimages $\{\bx_{i, b_i}\}_{i \in [k]}, b_i = 0 \text{ or } 1$ for $\{\by_i\}_{i \in [k]}$.

\begin{theorem}
\label{thm:last_hybrid_extract}
    In the last Game $2.k$, if we have $\ati_{\cP,\cD,1/2+\gamma}^{\epsilon, \delta}(\rho_\delete)$ outputs 1, for some noticeable $\gamma$, then there exists an extractor $\Ext$ such that there is a negligible function $\negl(\cdot)$: $$\Pr[\Ext(\rho_\delete, \pk) \to (\bx_{1}, \cdots, \bx_{k}): \bx_i \text{ is the secret in } \bx_i \bA_i + \be_{i,0}] \geq 1 - \negl(\lambda)$$   
    $\Ext$ runs in time $T'= O( T\cdot kn B_\bx \cdot \mathsf{poly}( \log(1/\delta), 1/\epsilon))$, where $T$ is the running time of the decryptor $\rho_\delete$ and $B_\bx$ is the maximum value 
    for each entry $\bx_{i,j}$ in the secret $\bx$.

In other words, we have in Game 2.k:
    \begin{align*}
    \Pr\left[ \begin{array}{cc}
         &\Ext(\rho_\delete, \pk) \to (\bx_{1}, \cdots \bx_{k}):  \\
         &  \bx_i \text{ is the secret in } \bx_i \bA_i + \be_{i}'
    \end{array}         
    \vert \ati_{\cP,\cD_{2.k}, 1/2+\gamma}^{\epsilon,\delta}(\rho_\delete) = 1 \right] \geq  1-\negl(\lambda)
    \end{align*}
\end{theorem}

We then obtain the following corollary from \Cref{thm:last_hybrid_extract}:
\begin{corollary}
\label{cor:hyb0_extract}
Assuming the subexponential post-quantum hardness security of $\lwe_{n,m,q, B_L}$,
     in Game 0, for any QPT $\A$ with auxiliary quantum input, there exists some negligible $\negl(\cdot)$, it holds that:
    \begin{align*}
    \Pr\left[ \begin{array}{cc}
         &\Ext(\rho_\delete, \pk) \to (\bx_{1}, \cdots \bx_{k}):  \\
         &  \bx_i \in \mathrm{INV}(\td_i, b\in\{0,1\}, \by_i)
    \end{array}         
    \vert \ati_{\cP,\cD_0, 1/2+\gamma}^{\epsilon,\delta}(\rho_\delete) = 1 \right] \geq  1-\negl(\lambda)
    \end{align*}
\end{corollary}

We refer to \Cref{sec:missing_proofs} for proof details.

\subsubsection{Proof Outline Using Inefficient Measurement $\ti$}
\label{sec:extract_proof_outline}
We first give a high level description of the proof for \Cref{thm:last_hybrid_extract}.

As discussed above, we need to show the following: given that we have a good decryptor $\rho_\delete$, we will have to extract (from the remaining state after the "good-decryptor" measurement") all preimages $\{\bx_{i, b_i}\}_{i \in [k]}, b_i = 0 \text{ or } 1$ for $\{\by_i\}_{i \in [k]}$ with probability negligibly close to 1. 

The procedure will be a quantum analogue of a LWE search-to-decision reduction with extraction probability close to 1, where the input distributions to the search-game adversary are not exactly LWE versus random, but statistically close to such.

\paragraph{Three $\ti$ Distributions} 
For clarity, we consider a few simplifications:
\begin{itemize}
    \item We use the ideal, projective threshold implementation $\ti$ in our algorithm. 
    \item We consider the number of instances/repetitions in the protocol to be $k = 1$.
\end{itemize}
In the full proof, we will use the efficient $\ati$ and 
polynomial $k$, via similar arguments. 

 To prove our theorem, it suffices to show that  $\Pr[\text{Extraction of} \bx] \geq 1-\negl(\lambda)$ in a world
 where we are given the following condition at the beginning of the algorithm:
$$ \ti_{\gamma+1/2}(\cP_{\cD_\ct}) \cdot \rho_\delete = 1$$
where $\cP_{\cD_\ct}$ is the following mixture of projections, acting on the state $\rho_\delete$:
\begin{itemize}
      \item Compute $\ct_0 \leftarrow \Enc(\pk, 0)$ \emph{in Game $2.k$}. $\ct_0$ will have the format of \Cref{eqn:ct_game2} with $\mu = 0$ and $k = 1$:
         \begin{align*}
        \ct_0 &= \begin{bmatrix} \bu \\\bA 
        \\ \bx\bA + \be' +  b \cdot \bu \end{bmatrix} \cdot \bR + \bE \\
        & = \begin{bmatrix} \bu \bR \\\bA \bR 
        \\ \bx\bA \bR + \be'\bR +  b \cdot \bu \bR + \be'' \end{bmatrix} 
    \end{align*}
      where $(\bA, \bu)$ are already given in the public key $\pk$; the rest is sampling $\bR \gets \{0,1\}^{m \times m}$ and the $(n+2)$-th row in $\bE$ is $\be_i'' \gets \cD_{\Z_q^m, B_{P'}}$. Note that $b = 0 \text{ or } 1,$ is an adversarially chosen bit that come in the $\by$ part of $\pk$.

       \item Compute $\ct_1 \gets 
       \mathcal{C}$, a random ciphertext from the possible space of all ciphertexts for 1-bit messages.
       
       \item Sample a uniform bit $\ell \leftarrow \{0,1\}$ \footnote{To distinguish from the bit $b$ in $\ct$, we use the notation $\ell$ for this random coin}. 
       \item Run the quantum decryptor $\rho$ on input $\ct_\ell$. Check whether the outcome is $\ell$. If so, output $1$, otherwise output $0$.
\end{itemize}

We then consider a second threshold implementation $\ti_{1/2+\gamma}(\cP_{\cD(g_i)})$. $\cP_{\cD(g_i)}$  is the  following mixture of measurements (we denote the following distribution we sample from as $\cD(g_i)$.):
\begin{itemize}

\item Let $g_i$ be a guess for the $i$-th entry in vector $\bx$.

\item Sample a random $\bc \gets \Z_q^{1 \times m}$, and let matrix $\bC \in \Z_q^{n \times m}$ be a matrix where the $i$-th row is $\bc$ and the rest of rows are $\mathbf{0}$'s.


\item Prepare $\ct_0$ as follows:
     \begin{align*}
        &\ct_0 = \begin{bmatrix} \bu \bR \\ {\color{red}\bA \bR + \bC}
        \\ \bx\bA \bR + \be'\bR +  b \cdot \bu \bR + \be''  + {\color{red} g_i \cdot \bc} \end{bmatrix} 
    \end{align*}
        where $(\bA, \bu)$ are already given in the public key $\pk$; $\bR \gets \{0,1\}^{m \times m}$ and $\be_i'' \gets \cD_{\Z_q^m, B_{P'}}$. Note that $b = 0 \text{ or } 1,$ is an adversarially chosen bit that comes in the $\by$ part of $\pk$.

\item Compute $\ct_1 \gets 
       \mathcal{C}$, a random ciphertext from the possible space of all ciphertexts for 1-bit messages.
       In our case, that is: $\ct_1 \gets \Z_q^{(n+2)\times m}$. 

\item Flip a bit $\ell \gets \{0,1\}$.

\item Run the quantum distinguisher $\rho$ on input $\ct_\ell$. Check whether the outcome is $\ell$. If so, output $1$, otherwise output $0$.
\end{itemize}

We finally consider a third threshold implementation, we call $\ti_{1/2+\gamma}(\cP_{\cD_\unif})$:

\begin{itemize}
    \item Compute both $\ct_0, \ct_1 \gets 
       \mathcal{C}$, as random ciphertexts from the possible space of all ciphertexts for 1-bit messages. In our case, that is: $\ct_0, \ct_1 \gets \Z_q^{(n+2)\times m}$. 

\item Flip a bit $\ell \gets \{0,1\}$.

\item Run the quantum distinguisher $\rho$ on input $ \ct_\ell$. Check whether the outcome is $\ell$. If so, output $1$, otherwise output $0$.
\end{itemize}

\paragraph{The Extraction Algorithm}
We describe the extraction algorithm as follows. It takes input $\pk = (\bA, \bu)$ and a quantum state $\rho_\delete$.

\begin{itemize}
    \item Set our guess for $\bx$ as $\bx' = \mathbf{0} \in \Z^{1\times n}_q$  and $\bx_i'$ is the $i$-th entry of $\bx'$. 
    \item For $i = 1, 2,\cdots, n$:
\begin{itemize}
    \item For $g_i \in [B_\bx]$, where $ [B_\bx]$ is the possible value range for $\bx_i \in \Z_q$:
\begin{enumerate}
\item Let $\rho_\delete$ be the current state from the quantum distinguisher. 
\item Run $\ti_{1/2+\gamma}(\cP_{\cD(g_i)})$ on $\rho_\delete$ wih respect to $\pk$.
\item If $\ti_{1/2+\gamma}(\cP_{\cD(g_i)})$ outputs 1, then set $\bx_i' := g_i$ and let $i := i+1$.

\item If $\ti_{1/2+\gamma}(\cP_{\cD(g_i)})$ outputs 0, the let $g_i := g_i+1$ and go to step 1. 
\end{enumerate}
\end{itemize}
    \item Output $\bx'$. 
\end{itemize}

\paragraph{Measurement-preserving Properties in the Extraction Algorithm}
We will show that the output in the above algorithm $\bx'$ is equal to $\bx$ in the $\bx\bA + \be'$ with probability $(1-\negl(\lambda))$.

First recall that at the beginning of our algorithm,  we are given the following condition: 
$$ \Pr[\ti_{\gamma+1/2}(\cP_{\cD_\ct}) \cdot \rho_\delete \to 1] = 1$$

We then prove the following claims:
\begin{enumerate}
    \item In the above algorithm, when our guess $g_i = \bx_i$, we have:
    \begin{itemize}
        \item $\cD(g_i)$ is statistically indistinguishable from $\cD_\ct$. 

        \item $\ti_{1/2+\gamma}(\cP_{\cD(g_i)})$ outputs 1 on the current state $\rho_\delete$ with $1$ probability.
        
    \end{itemize}

     \item In the above algorithm, when our guess $g_i \neq \bx_i$, we have:
    \begin{itemize}
        \item $\cD(g_i)$ is statistically indistinguishable from $\cD_\unif$.

        \item $\ti_{1/2+\gamma}(\cP_{\cD(g_i)})$ outputs 0 on the current state $\rho_\delete$ with $1$ probability.
        
    \end{itemize}
\end{enumerate}

Combining the above arguments, we can conclude that the algorithm outputs $\bx$ with probability $1$:
whenever $\ti_{1/2+\gamma}(\cP_{\cD(g_i)})$ outputs 1, we know our guess is correct with probability $1$ and our state is undisturbed, so we can move on to guessing the next entry; whenever  $\ti_{1/2+\gamma}(\cP_{\cD(g_i)})$ outputs 0, we know our guess is incorrect with probability 1; therefore our state is undisturbed, and we can move on to the next value for our guess. These invariants preserve throughout the above algorithm, because of the projective property of $\ti$ and gentle measurement in the setting of no accumulating errors. 

To make the above algorithm efficient, we use an approximate version of $\ti$, called $\ati$ (see\Cref{sec:ati} for detailed prelims.). When we move on to using an approximate version of $\ti$, we will have to be more carefull: negligible errors will accumulate additively on the measurement outcomes $\rho_\delete$ after every $\ti$-measurement, but will stay negligible throughout the algorithm with our choice of parameters (see \Cref{sec:ful_extraction_analysis} for details) by the properties of ATI in \Cref{sec:ati}. By these guarantees, we obtain $\bx$ with $1$ probability in the end.


\ifllncs
\subsubsection{Full Extraction Algorithm and Analysis}
We refer the readers to \Cref{sec:ful_extraction_analysis} for the full algorithm and analysis.
\else
We refer the readers to the following section, \Cref{sec:ful_extraction_analysis} for the full algorithm and analysis.
\fi



\subsection{Reduction to Parallel NTCF Soundness}
\label{sec:reduction_to_ntcf}
Now we prove \Cref{lem:hyb2_negligible} and \Cref{thm:security} follows accordingly.

Now we can build a reduction to break the security of the parallel-repeated NTCF-based protocol in \Cref{thm:parallel_soundness_2of2} as follows:

We plug in the condition that \textbf{we have to check the validity of the deletion certificate}.   Recall our notations for the events happening in Hybrid 1:  
\begin{itemize}
    \item We denote the event that the adversary hands in a valid deletion cerficate, i.e. $\VerDel(\pk,\td, \cert) =  \valid$, as $\mathsf{Cert Pass}$.

    \item We denote the event that
    test $\ati^{\epsilon, \delta}_{\cP, D, \gamma+1/2}(\rho_{\delete})$ outputs 1 with respect to $\mu$ and $\pk$, as $\mathsf{Good Decryptor}$.

    \item We denote the event that we can obtain the preimages $\{\bx_{i}\}_{i \in [k]} \in \{\mathrm{INV}(\td_i, b\in\{0,1\}, \by_i)\}_{i \in [k]}$ from the state $\rho_\delete'$ as $\Ext$.
\end{itemize}

Suppose the adversary wins the $\gamma$-strong SKL-PKE game in \Cref{def:strong_skl}, then we must have $\Pr[\mathsf{Cert Pass} \wedge \mathsf{Good Decryptor}] \geq 1/p$ for some noticeable $1/p$. By \Cref{cor:hyb0_extract}, we have \ifllncs $\Pr[\Ext | \mathsf{Good Decryptor}] \\ \geq 1 - \negl(\lambda)$ \else $\Pr[\Ext | \mathsf{Good Decryptor}] \geq 1 - \negl(\lambda)$\fi. 
Therefore we have $\Pr[\mathsf{Cert Pass} \wedge \Ext] \geq 1/p- \negl'(\lambda)$ for some negligible $\negl'(\lambda)$. The relation can be easily observed from a Venn diagram; we deduce it formally in \Cref{sec:prob_relations}.


Now we can build a reduction to break the security of the parallel-repeated NTCF-based protocol in \Cref{thm:parallel_soundness_2of2} as follows:
the reduction plays as the challenger in the strong-SKL game. It passes the $\{f_{i,b}\}_{i\in[k], b \in \{0,1\}}$ from the NTCF challenger to the adversary $\A$ as $\mpk$.

It takes $\{\by_i\}_{i \in [k]}$ and the deletion certificate $\{(c_i, \bd_i)\}_{i \in [k]}$ from the adversary and runs the $\gamma$-good test on the its post-deletion state $\rho_\delete$. If the test fails, then abort; if the test passes, run the extractor $\Ext$ (from \Cref{thm:last_hybrid_extract}) on the post-test state $\rho_\delete'$ to obtain $\{\bx_{i}\}_{i \in [k]}$.
Then it gives $\{\by_i\, c_i, \bd_i, \bx_{i}\}_{i \in [k]}$ all to the NTCF-protocol challenger. By the above argument, if $\A$ wins with probability $1/p$, then the reduction wins with probability $1/p- \negl(\lambda)$.

\section{Proof for \Cref{thm:last_hybrid_extract}: Quantum LWE Search-to-Decision Algorithm with Almost-Perfect Extraction (Noisy Version)}
\label{sec:ful_extraction_analysis}
In this section, we provide a full extraction algorithm in Game $2.k$ for extracting the $\{\bx_{i,b_i}\}_{i \in [k]}$  and its analysis, to prove \Cref{thm:last_hybrid_extract}. We call it "noisy" search-to-decision algorithm because the distribution in \Cref{thm:last_hybrid_extract} is not exactly LWE versus random, but statistically close to.

In \Cref{sec:lwe_search_to_decision} we provide a cleaner version of the algorithm and analysis where the input distributions are real LWE instances versus real uniform random instances, so that we have a first quantum LWE search-to-decision with almost perfection extraction, even with auxiliary quantum inputs.
We believe such a statement for (plain) LWE search-to-decision reduction
will be of independent use.

\noindent 
\textbf{Notations}
For clarity, we take away the index $b_i$ in the vector $\bx_{i, b_i}$ since obtaining $\bx_{i, 0}$ or $\bx_{i, 1}$ does not affect our analysis.
From now on, the subscripts $i,j$ in $\bx_{i,j}$ represents the $j$-th entry in the $i$-th vector $\bx_i$.  

\subsection{Three Distributions for $\ati$}
Similar to the proof outline in \Cref{sec:extract_proof_outline}, we first describe three $\ati$'s with respect to different distributions $\cD$ (and accordingly, the mixture of 
projections $\cP$).

\paragraph{$\ati$ for $\cD_{2.k}$:}

 $\ati_{\cP, \cD_{2.k},1/2+\gamma}^{\epsilon,\delta}$ is the approximate threshold implementation algorithm for the following mixture of projections $\cP_{\cD_\ct}$, acting on the state $\rho_\delete$:
\begin{itemize}
      \item Compute $\ct_0 \leftarrow \Enc(\pk, 0)$ \emph{in Game $2.k$}. $\ct_0$ will have the format of \Cref{eqn:ct_game2} with $\mu = 0$ :
         \begin{align*}
        \ct_0 &   = \begin{bmatrix} \bu_1  \\\bA_1 \\ \\ \cdots 
        \\  \bu_k\\
        \bA_k 
        \\ \sum_{i \in [k]}(\bx_{i}\bA_i + \be_{i}' +  b_{i} \cdot \bu_i) \end{bmatrix} \cdot \bR + \bE  \\
        & = \begin{bmatrix} \bu_1 \cdot  \bR \\\bA_1 \cdot \bR \\ \\ \cdots 
        \\  \bu_k \cdot\bR \\
        \bA_k \cdot\bR 
        \\ \sum_i( \bx_{i}\bA_i \bR + \be_i'\bR +  b_i \cdot \bu_i \bR )+ \be'' \end{bmatrix} 
    \end{align*}
      where $\{\bA_i,\bu_i \}_{i\in [k]}$  are given in the public key $\pk$; $\bR \gets \{0,1\}^{m \times m}$ and the $(n+2)$-th row in $\bE$ is $\be'' \gets \cD_{\Z_q^m, B_{P'}}$. Note that $b_i = 0 \text{ or } 1,$ is an adversarially chosen bit that come in the $\by$ part of $\pk$.

       \item Compute $\ct_1 \gets 
       \mathcal{C}$, a random ciphertext from the possible space of all ciphertexts for 1-bit messages.  In our case, that is: $\ct_1 \gets \Z_q^{(nk+k+1)\times m}$. 
       
       \item Sample a uniform bit $b \leftarrow \{0,1\}$
       \item Run the quantum decryptor $\rho$ on input $\ct_b$. Check whether the outcome is $b$. If so, output $1$, otherwise output $0$.
\end{itemize}

\paragraph{$\ati$ for $\cD(g_{\ell,j})$:}
We then consider a second approximate threshold implementation $\ati_{\cP,\cD(g_{\ell,j}), 1/2+\gamma}^{\epsilon,\delta}$. $\cP_{\cD(g_{\ell,j})}$  is the  following mixture of measurements (we denote the following distribution we sample from as $\cD(g_{\ell,j})$.):
\begin{itemize}

\item Let $g_{\ell,j}$ be a guess for the $j$-th entry in vector $\bx_\ell \in (\bx_1, \cdots, \bx_k)$.

\item Sample a random $\bc \gets \Z_q^{1 \times m}$, and let matrix $\bC \in \Z_q^{n \times m}$ be a matrix where the $j$-th row is $\bc$ and the rest of rows are $\mathbf{0}$'s.


\item Prepare $\ct_0$ as follows:
     \begin{align}
     \label{eqn:ct_g_uniform}
        &\ct_0 = \begin{bmatrix} \bu_1 \bR_1 \\ \bA_1\bR
        \\\cdots \\
        \bu_\ell\bR \\
        {\color{red}\bA_\ell \bR + \bC_\ell} \\
        \cdots\\
        \bu_k \bR \\
        \bA_k \bR \\
        \sum_{i \in [k]}(\bx_{i}\bA_i \bR + \be_i'\bR +  b_i \cdot \bu_i \bR )  +  \be''+ {\color{red} g_{\ell,j} \cdot \bc} \end{bmatrix} 
    \end{align}   

 where $\{\bA_i,\bu_i \}_{i\in [k]}$ are given in the public key $\pk$; $\bR \gets \{0,1\}^{m \times m}$ and the $(n+2)$-th row in $\bE$ is $\be'' \gets \cD_{\Z_q^m, B_{P'}}$. Note that $b_i = 0 \text{ or } 1,$ is an adversarially chosen bit that come in the $\by$ part of $\pk$.

\item Compute $\ct_1 \gets 
       \mathcal{C}$, a random ciphertext from the possible space of all ciphertexts for 1-bit messages.
       In our case, that is: $\ct_1 \gets \Z_q^{(nk+k+1)\times m}$. 

\item Flip a bit $b \gets \{0,1\}$.

\item Run the quantum distinguisher $\rho$ on input $\ct_b$. Check whether the outcome is $b$. If so, output $1$, otherwise output $0$.
\end{itemize}

\paragraph{$\ati$ for $\cD_\unif$:}
We finally consider a third threshold implementation, we call $\ati_{1/2+\gamma,\cP,\cD_\unif}^{\epsilon,\delta}$:

\begin{itemize}
    \item Compute both $\ct_0, \ct_1 \gets 
       \mathcal{C}$, as random ciphertexts from the possible space of all ciphertexts for 1-bit messages. In our case, that is: $\ct_0, \ct_1 \gets \Z_q^{(nk+k+1)\times m}$. 

\item Flip a bit $b \gets \{0,1\}$.

\item Run the quantum distinguisher $\rho$ on input $ \ct_b$. Check whether the outcome is $b$. If so, output $1$, otherwise output $0$.
\end{itemize}

\subsection{Extraction Algorithm}
\label{sec:extractor_algo}
We describe the extraction algorithm as follows. It takes input $\pk = (\{\bA_i, \bu_i\}_{i \in [k]}$ and a quantum state $\rho_\delete$, as well as parameters: threshold(inverse polynomial) $\gamma$, range for each entry in the secret $B_X$, timing parameter (inverse exponential) $\delta$. Note that the secret dimension $k \cdot n$ comes with $\pk$.

\begin{itemize}
    \item Let $\bx_{\ell,j}'$ be the register that stores the final guess for the $j$-th entry of $\bx_\ell$, initialized with all zeros. 

    \item For $\ell = 1, \cdots, k$:
\begin{itemize}
    \item For $j = 1, 2,\cdots, n$:
\begin{itemize}
    \item For $g_{\ell,j} \in [B_X]$, where $[B_X]$ is the possible value range for $\bx_{\ell,j} \in \Z_q$:
\begin{enumerate}
\item Let $\rho_\delete$ be the current state from the quantum distinguisher. 
\item Let $\gamma := \gamma - 3\epsilon$, where $\epsilon = \frac{\gamma}{100B_X kn}$.
\item Run $\ati_{1/2+\gamma,\cP,\cD(g_{\ell, j})}^{\epsilon,\delta}$ on $\rho_\delete$ wih respect to $\pk$. 
\item If $\ati_{1/2+\gamma,\cP,\cD(g_{\ell, j})}^{\epsilon,\delta}$ outputs 1, then set $\bx_{\ell,j}' := g_{\ell,j}$ and move on to the next coordinate, let $j := j+1$ if $j < n$, else let $\ell:=\ell+1, j = 1$.

\item If $\ati_{1/2+\gamma,\cP,\cD(g_{\ell, j})}^{\epsilon,\delta}$ outputs 0, the let $g_i := g_i+1$ and go to step 1. 
\end{enumerate}
\end{itemize}
\end{itemize}
    \item Output $\bx'$. 
\end{itemize}

\subsection{Analysis of the Extractor}
\label{sec:proof_extractor_noisy}
We make the following claims:

\begin{claim}
    \label{claim:gi_correct_distr}
    When the guess $g_{\ell,j} = \bx_{\ell,j}$,
the distributions $\cD_{2.k}, \cD(g_{\ell,j})$ are statistically close by distance $\eta_0 = \poly(k,m) \cdot (\frac{1}{q^n} + \frac{B_P}{B_{P'}})$.
\end{claim}

\begin{proof}
   Note that the two distributions are the same except on how they sample $\ct_0$.
   
     The ciphertext distribution for $\ct_0$ in $\cD_{2.k}$  is statistically close to the following distribution:
      \begin{align}
      \label{eqn:ct_after_lhl}
        &\ct_0 = \begin{bmatrix} 
        {\color{red}\bu_1'}  \\{\color{red} \bA_1'} \\ \\ \cdots 
        \\   {\color{red}  \bu_\ell' }\\
       {\color{red} \bA_\ell'} \\ \cdots
         \\ {\color{red} \bu_k'}\\
       {\color{red} \bA_k' } 
        \\{\color{red} \sum_{i} (\bx_{i}\bA_i' + \be_{i,0}' + b_i \cdot \bu_i' )}  
        \end{bmatrix} 
    \end{align}
where $\bA_i' \gets \Z_q^{n\times m}, \bu_i' \gets \Z_q^{1 \times m}$ are uniform random and given in the public key $\pk$;; $\be_{i,0}' \gets \cD_{\Z_q^m, B_P'}$, where $B_P'/B_P$ is superpolynomial, for all $i \in [k]$;  $b_i = 0 \text{ or } 1, i \in [k]$ are some arbitrary, adversarially chosen bits.

We can view the distribution change as two differences:
We can view the first change as: 
 $\sum \by_i\cdot \bR + \bE_{nk+k+1} = \sum (\bx_{i,b_i}\bA_i\bR + \be_{i}'\bR + \be_i'' + b_i \cdot \bu_i\bR)$ in \Cref{eqn:ct_game2} to $\sum (\bx_i\bA_i\bR + \be_{i,0}')$, where both noise $\be_{i,0}'$ and $\be_i''$ are sampled from  $\cD_{\Z_q^m, B_P'}$.
Since $B_{P'}/B_P$($B_P$ for the Gaussian parameter of the error $\be_i'$ in $\by_j$) is superpolynomial, we can apply noise-flooding/smudging (\Cref{lem:noise_flooding})  and the two are statistically close by $k \cdot B_P/B_{P'}$.

The second difference is
from using $(\bA_i\bR, \bu_i \bR, \sum (\bx_{i,b_i}\bA_i\bR + \be_{i,0}' + b_i \cdot \bu_i\bR) )$
to using $(\bA_i' \gets \Z_q^{n \times m}, \bu_i' \gets \Z_q^m, \sum (\bx_{i,b_i}\bA_i' + \be_{i,0}' + b_i \cdot \bu_i'))$. These two are $k \cdot q^{-n}$-close by the leftover hash lemma. 

We then show that $\ct_0$ in $\cD(g_{\ell,j})$ is also close to this distribution.
First we also apply noise flooding to replace every $\be_i'\bR + \be_i''$ with $\be_{i,0}'$.
We next replace $\bA_i\bR, \bu_i\bR$'s by the LHL with random $\bA_i', \bu_i'$.
We can ignore $\sum_i b_i \cdot \bu_i'$ in the last row from our distribution, because $b_i$
is known to the adversary and they are the same in both distributions.

Then we observe that when $g_{\ell,j} = \bx_{\ell,j}$, we let $\bA_\ell'' = \bA_\ell' + \bC$ where $\bC$ is everywhere 0 except the $j$-th row being uniformly random $\bc$. We also have $\bx_{\ell,j} \bA_\ell''  + \be_{i,0}' + g_{\ell,j} \cdot \bc = [(\bx_{\ell,j}(\bA_{\ell,1,j}''+\bc_{1}) + \sum_{ i\neq j}\bx_{\ell,j}(\bA_{\ell,1,i}''+0), \cdots, (\bx_{\ell,j}(\bA_{\ell,m,j}''+\bc_{m}) + \sum_{i\neq i=j}\bx_{\ell,j}(\bA_{\ell,m,i}''+0)] + \be_{i,0}' = \bx_{\ell,j} \bA_\ell'' + \be_{i,0}'$, where $\bA_{\ell,x,y}''$ denotes the entry in $x$-th column and $y$-th row of $\bA_{\ell}''$. 

\end{proof}

\begin{claim}
    \label{claim:gi_incorrect_distr}
    When the guess $g_{\ell,j} \neq \bx_{\ell,j}$,
the distributions $\cD_{2.k}, \cD_\unif$ are statistically close by distance $\eta_1 = \poly(k,m)\cdot \frac{1}{q^n}$.
\end{claim}

\begin{proof}
the two distributions are the same except on how they sample $\ct_0$.
    $\ct_0$ in $\cD_\unif$ is uniformly sampled from $\Z_q^{(nk+k+1)\times m}$. It remains to show that $\ct_0$ in $\cD(g_{\ell,j})$ is close to this distribution.
    
    Different from \Cref{claim:gi_correct_distr}, we do not need to apply noise flooding but can apply LHL to $\ct_0$ shown in \Cref{eqn:ct_g_uniform} directly because the last row is masked by uniformly random vector and the error $\be_i'\bR$ is masked by uniform as well. Let $\bA_\ell'' = \bA_\ell' +\bC, \bA_\ell' \gets \Z^{n \times m}_q$. $\bA_\ell''$ is uniformly random because the only change is adding the uniform random vector $\bc$ in its $j$-th row.
    Now we observe that  when $g_{\ell,j} \neq \bx_{\ell,j}$. 
    We can ignore the term $\sum_i b_i \cdot \bu_i'$ in the last row from our distribution, because $b_i$
is known to the adversary and they are the same in both distributions. 
    We consider the vector $\bw = \bx_{\ell,j}^\top \bA_\ell''  + \be_{i,0}' + g_{\ell,j} \cdot \bc = [(g_{\ell,j} \cdot \bc_{1}+ \sum_{i}\bx_{\ell,i}\bA_{\ell,1,i}'', \cdots, g_{\ell,j}\cdot \bc_{m}+ \sum_{i}\bx_{\ell,i}\bA_{\ell,m,i}''] + \be_{i}'\bR$. Since $\bc$ is uniformly random, the entire $\bw$ now becomes uniformly random.
    Since the last row of $\ct_0$ in $\cD(g_{\ell,j})$ is  $\sum_{i \neq \ell}  (\bx_{i,b_i}\bA_i' + \be_{i}'\bR + b_i \cdot \bu_i') + \bw$, it is masked by $\bw$ and becomes uniformly random. 
    
\end{proof}

\paragraph{Invariant Through Measurements}
We now demonstrate two important properties of the approximate threshold implementation measurements, when applying with correct guesses $g_{\ell,j}$ and incorrect guesses $g_{\ell,j}$:
\begin{itemize}
    \item When applying ATI with incorrect guesses $g_{\ell,j}$,  the leftover state almost does not change in terms of trace distance, compared to the state before the measurement. 
    
    \item When applying ATI with correct guesses $g_{\ell,j}$, our measurement outcomes will almost always stay the same from the last measurement correct $g_{\ell,j}$, even if the leftover state may change at the end of the measurement.  
     
\end{itemize}

Let us denote the probability of the measurement \textbf{outputting} $\mathbf{1}$ on $\rho$ by $\Tr[\ati^{\epsilon, \delta}_{\cP, \cD, 1/2+\gamma}\, \rho]$. Accordingly  $1- \Tr[\ati^{\epsilon, \delta}_{\cP, \cD, 1/2+\gamma}\, \rho] := \Pr[\ati^{\epsilon, \delta}_{\cP, \cD, 1/2+\gamma}\, \rho \to 0] $.

\begin{corollary}
\label{cor:ati_close_2distri}
    We make the following claim. For any inverse polynomial $\gamma, \epsilon < \gamma$ and exponentially small $\delta$, it holds that: 
    \begin{itemize}
        \item  If $\ati^{\epsilon, \delta}_{\cP, \cD_{2.k}, 1/2+\gamma}\, \rho$ outputs 1 and a leftover state $\rho'$, then $\Tr[\ati^{\epsilon, \delta}_{\cP, \cD(g_{\ell,j})), 1/2+\gamma-3\epsilon}\, \rho'] \geq 1 - 2\delta - \negl(\lambda)  $.

    \end{itemize}

    
\end{corollary}

\begin{proof}
    This follows directly from \Cref{claim:gi_correct_distr}, \Cref{cor:ati_thresimp} and \Cref{cor:ati_computational_indistinguishable}.

By \Cref{cor:ati_thresimp}, we have that if $\ati^{\epsilon, \delta}_{\cP, \cD_{2.k}, 1/2+\gamma}\, \rho$ outputs 1 and a leftover state $\rho'$, then $\\ \Tr[\ati^{\epsilon, \delta}_{\cP, \cD_{2.k}, 1/2+\gamma-3\epsilon}\, \rho']  \geq 1 -  \delta$.
Next, by \Cref{cor:ati_computational_indistinguishable}, we have $\Tr[\ati^{\epsilon, \delta}_{\cP, \cD(g_{\ell,j})), 1/2+\gamma}\, \rho'] \geq \Tr[\ati^{\epsilon, \delta}_{\cP, \cD_{2.k}, 1/2+\gamma}\, \rho']  -  \eta_0'$.  Thus, overall we have $\Tr[\ati^{\epsilon, \delta}_{\cP, \cD(g_{\ell,j})), 1/2+\gamma-3\epsilon}\, \rho'] \geq 1 - \delta - \negl(\lambda) = 1- \negl(\lambda)$ since $\delta$ is exponentially small.
\end{proof}

Next we show that when the guess $g_i$ is incorrect, the corresponding operation acts almost like an identity operator and will almost always output 0.

\begin{claim}
\label{claim:uniform_close_ati}
 For all $\ell \in [k], j\in [n]$,  When the guess $g_{\ell,j} \neq \bx_{\ell,j}$ in the above mixture of projections $\cP = (P_{\cD_{g_i}}, Q_{\cD_{g_i}})$, for any inverse polynomial $\gamma$, any exponentially small $\delta$ and any quantum distinguisher $\rho$, with probability $(1-\delta)$, the following two claims hold by our choice of parameter $\epsilon = \frac{\gamma}{100B_X \cdot n}$:

 \begin{itemize}
     \item For any input state $\rho$,
let $\rho'$ be the state after applying $\ati^{\epsilon,\delta}_{\cP, 1/2+\gamma}$ on $\rho$, we have $\lVert \rho-\rho'\rVert_{\Tr} \leq O( \eta_1 \cdot \frac{\ln(4/\delta)}{\epsilon})$, where $\eta_1$ is the statistical distance between $(\cD(g_{\ell,j}), \cD_\unif)$. 

\item $\ati^{\epsilon,\delta}_{\cP, 1/2+\gamma}$ will output outcome 0.
 \end{itemize}

\end{claim}

\begin{proof}
By \Cref{claim:gi_incorrect_distr}, we know that the statistical distance between 
$\cD(g_{\ell,j})$ and $\cD_\unif$ is exponentially small. Therefore by applying \Cref{thm:invariant_ati_uniform}, we obtain the above two claims.

\end{proof}

\paragraph{Correctness} Now we are ready to prove the correctness of our algorithm.
\begin{lemma}
\label{lem:extractor_correctness}
If at the beginning of the algorithm in \Cref{sec:extractor_algo}, we have that $\ati^{\epsilon, \delta}_{\cP, \cD_{2,k}), 1/2+\gamma}\, \rho$ \textbf{outputs 1} for some inverse polynomial $\gamma$, exponentially small $\delta$ and $\epsilon = \frac{\gamma}{100B_{X} k n}$, the extractor algorithm will output correct secret with probability $(1-\negl(\lambda))$, by our choice of parameters $B_X,B_P, B_{P'}$ in \Cref{sec:scheme_parameters}.
\end{lemma}

We are given that the distinguisher $\rho_\delete$ (let us call it $\rho$ now for brevity) satisfies $\ati^{\epsilon, \delta}_{\cP, \cD_{2.k}, 1/2+\gamma}\rho \to 1$, for some exponentially small $\delta$ and some $\epsilon$ of our own choice (here, $\epsilon = \frac{\gamma}{100B_X kn}$), according to our assumption in \Cref{thm:last_hybrid_extract}. 

 Note that since we have obtained this outcome, we must have already applied once $\ati^{\epsilon, \delta}_{\cP, \cD_{2.k}, 1/2+\gamma}$ on $\rho$ and obtained some remaining state $\rho'$. Now the remaining state $\rho'$ will satisfy the condition that  $\Tr[\ati^{\epsilon, \delta}_{\cP, \cD_{2.k}, 1/2+\gamma-3\epsilon}\rho] \geq 1 - 3\delta$ by item 2 and 3 in \Cref{cor:ati_thresimp}. 
 
 Consider the first time we apply $\ati^{\epsilon, \delta}_{\cP, D(g_{\ell,j}), 1/2+\gamma-4\epsilon}$ (when $\ell,j =  1$ and $g_{\ell,j}$ takes the smallest value in $[B_X]$) in the above algorithm, suppose we have guessed $g_{\ell,j}$ correctly at this point:

 By \Cref{cor:ati_close_2distri}, 
 have that $\Tr[\ati^{\epsilon, \delta}_{\cP, D(g_i), 1/2+\gamma-4\epsilon}\, \rho] \geq 1-3\delta - \negl(\lambda) \geq  1 - \negl(\lambda)$, where $\delta$ is exponentially small. 
 \jiahui{move the following to the end of parameter analysis}


 Suppose the first time we apply $\ati^{\epsilon, \delta}_{\cP, D(g_{\ell,j}), 1/2+\gamma-3\epsilon}$, the guess $g_{\ell,j}$ is incorrect, as shown in \Cref{claim:uniform_close_ati}, we have $1 - \Tr[\ati^{\epsilon, \delta}_{\cP, D(g_{\ell,j}), 1/2+\gamma-3\epsilon}\, \rho'] \geq 1- \delta$.
 Moreover, with probability $(1-\delta)$, the leftover state $\rho''$ from the measurement $\ati^{\epsilon, \delta}_{\cP, D(g_{\ell,j}), 1/2+\gamma-3\epsilon}$ has $O(\eta_1 \cdot \ln(4/\delta)/\epsilon)$ difference 
 in trace distance from the state $\rho'$. 

 We can then continue our algorithm: every time we make a measurement $\ati^{\epsilon, \delta}_{\cP, D(g_{\ell,j}), 1/2+\gamma}$ (note that the threshold parameter $\gamma$ will be lowered by $3\epsilon$ from the threshold in the last measurement).
 We can then perform induction: assume that we get the wanted measurement outcomes (outcome 1 if $g_{\ell,j}$ is correct and 0 if $g_{\ell,j}$ is incorrect)
for all the first $L$-th measurements. Let us call the state after the $L$-th steps in the loop, $\rho_L$. When the $L+1$-th measurement uses an incorrect $g_{\ell,j}$: suppose $L'$ is the number of times we get outcome 0(including $L$-th measurement) since the last time we get outcome 1, then with probability $(1-L'\cdot \delta)$, the state is $1-(1-O(\eta_1 \cdot \ln(4/\delta)/\epsilon))^{L'} \leq L' \cdot O(\eta_1 \cdot \ln(4/\delta)/\epsilon)$ close in trace distance from the last time we measure with a correct $g_{\ell,j}$.

When the $L+1$-th measurement $\ati^{\epsilon, \delta}_{\cP, D(g_{\ell,j}), 1/2+\gamma}$ uses  a correct $g_{\ell,j}$, then with probability $(1- L' \cdot O(\eta_1 \cdot \ln(4/\delta)/\epsilon))(1- \negl(\lambda))$, we obtain outcome $1$, where $L'$ is the number of times we get outcome 0 since the last time we get outcome 1, by the fact that $|\Tr(\cP\rho_{L-L'}) - \Tr(\cP \rho_L)| \leq \lVert \rho_{L-L'} - \rho_L \rVert_{\Tr}$ for all POVM measurements $\cP$.

The total probability  we extract all $k\cdot n$-entries of the LWE secret is: $(1-\negl(\lambda))^{kn} \cdot (1-\delta)^{O(B_X \cdot kn)} \cdot (1-O(\eta_1 \cdot \ln(4/\delta)/\epsilon))^{O(B_X \cdot kn)}$. Here, $(1- \negl(\lambda))^{kn}$ is the total loss incurred by measurements $\ati^{\epsilon, \delta}_{\cP, D(g_{\ell,j}), 1/2+\gamma}$ with correct $g_{\ell,j}$'s, because we will make $kn$ (the dimension of the secret)  number of such measurements; the error $(1-\delta)^{O(B_X \cdot kn)} \cdot (1-O(\eta_1 \cdot \ln(4/\delta)/\epsilon))^{O(B_X \cdot kn)}$ is incurred by measurements $\ati^{\epsilon, \delta}_{\cP, D(g_{\ell,j}), 1/2+\gamma}$ with  incorrect $g_{\ell,j}$'s, because we will make $O(B_X kn)$ (the dimension of the secret multiplied by the range of the secret)  number of such measurements.

Also note that we lower our threshold $\gamma$ by $3\epsilon$ upon every use due to \Cref{cor:ati_thresimp}, but $\gamma$ remains inverse polynomial throughout the algorithm because the lowest it can get is $\gamma -3 \epsilon \cdot 2B_X kn =  \frac{94\gamma}{100}$.

 Let us recall our parameter settings (the followings are all in terms of security parameter $\lambda$): $q^{-n}$ is exponentially small and thus $\eta_1$ is exponentially small;
 $\delta$ is exponentially small. $k,n$ are polynomial. $B_X$ is quasipolynomial.

Therefore, the final success probability is:
\begin{align*}
    & (1- \negl(\lambda))^{kn} \cdot (1-\delta)^{O(B_X \cdot kn)} \cdot (1-O(\eta_1 \cdot \ln(4/\delta)/\epsilon))^{O(B_X \cdot kn)} \\
     \geq & (1- kn\cdot \negl(\lambda)) \cdot (1-O(B_X kn \cdot \delta)) \cdot (1-O(B_X kn \cdot \eta_1 \cdot \ln(4/\delta)/\epsilon)) \\
     \geq & 1 - kn \cdot 
     \negl(\lambda)
\end{align*}
By our parameter settings above, this is $1- \negl(\lambda)$.

\paragraph{Running Time}
By \Cref{cor:ati_thresimp}, the running time of each $\ati^{\epsilon, \delta}_{\cP, D(g_{\ell,j})), 1/2+\gamma}$ is $T \cdot O(1/\epsilon^2 \cdot 1/\log\delta) = T \cdot \poly(1/\epsilon, \cdot 1/\log\delta)$ where $T$ is the running of quantum algorithm $\rho$ and thus polynomial, so the entire algorithm is $O(T nk 
\cdot B_X \cdot \poly(1/\epsilon, 1/\log\delta))$, where $B_X$ is quasipolynomial, $\epsilon = O(\gamma/B_X kn)$ is an inverse quasipolynomial and $\delta$ can be an inverse exponential function. 
Since we rely on a subexponential hardness assumption, such quasipolynomial running time is acceptable.

\paragraph{Conclusion} 
By the above analysis, we can therefore conclude \Cref{thm:last_hybrid_extract}.

\jiahui{OLD proof. Moving to the end}

\section{Remarks on Definitions and Strong SKL Security Implies IND-SKL Security}
\label{sec:relation_security_defs}

\subsection{Prevent Impostors: Authorizing the Public Key}
  As we have discussed, for applications where everyone can encrypt and send it to the lessee, we need to have interaction in the key exchange.
  The key generation algorithm and encryption algorithm are defined as:

  \begin{itemize}
      \item   $\InKeyGen(1^\lambda)$: an interactive protocol that takes in a security parameter $1^\lambda$; $\A$ interacts with the challenger (classically) and they either output a public key $\pk$, a trapdoor $\td$ and a quantum secret key $\rho_\sk$ or output a symbol $\bot$.

      \item $\Enc'(\pk,m)$: an algorithm that takes in an alleged public key $\pk$ and a message $m$, outputs either $\ct \gets \Enc(\pk, m)$ or $\bot$.
 \end{itemize}

  We formalize the security requirement as follows:

  \begin{definition}[Anti-Impostor security]
   A secure key leasing scheme satisfies $\gamma$-anti-impostor security, if for any (non-uniform) QPT adversary $\A = (\A_1, \A_2)$,  there exists a negligible function $\negl(\cdot)$ such that for all $\lambda \in \N$, all messages $m \in \mathcal{M}$:

      $$\Pr \left[\begin{matrix} \A_1(1^\lambda,\pk) \rightarrow  (\pk', m_0, m_1) \\ 
      \A_2 \text{ is $\gamma$-good decryptor w.r.t } \pk',(m_0,m_1) \end{matrix} : \begin{matrix} 
    (\td, \pk,\rho_\sk)\leftarrow \InKeyGen(1^\lambda) \end{matrix}\right] \leq \negl(\lambda)$$ 
      
  \end{definition}

The simple solution to the above security is to
  add a lessor's signature on the final public key $\pk$ as presented in \Cref{sec:protocol_description}.
 The algorithm $\Enc'$ will first verify the signature on the public key and then proceed to encryption; else it would output $\bot$.

  \paragraph{Security Proof for Protocol in \Cref{sec:protocol_description}}
 The security relies on the existential unforgeability of the underlying signature scheme and the anti-piracy security. 

 If the part except for the signature on the fake $\pk'$ is different from the honestly generated $\pk$, then the adversary $\A$ has to forge a signature  on the new public key it generated, then we can use it to break existential unforgeability of the underlying signature scheme.

 If $\pk' = \pk$ on the parts except for the signatures, then we can use it to break the strong $\gamma$-anti-piracy security: the pirate $\A'$ honest generates a public key $\pk$ and give it to $\A$; $\A'$ can delete its own key honestly and pass the deletion check; it runs $\A_1$ to get messages $(m_0,m_1)$ and give to the challenger; then it uses $\A_2$, which is a good decryptor with respect to $\pk,m_0,m_1$, to break the anti-piracy security.

\subsection{Additional Remarks for Our Definitions}



\begin{remark}
    We can additionally add the following verification algorithm on the public key generated, but not necessary for our construction:
\begin{itemize}
    \item $\mathsf{VerPK(\mpk,\td, \pk)}:$ on input a public key $\pk$, a master public key $\mpk$ and trapdoor $\td$, output $\valid$ or $\invalid$.
\end{itemize}
This algorithm allows the lessor to check if a public key generated by the lessee is malformed. This will render the key generation protocol as two-message.

However, in our construction, the verification of deletion inherently prevents attacks by generating malicious public keys. For correctness it only has to hold for honest key generations. We thus omit it.

But having such an algorithm in the definition can open the possibility for other SKL constructions with classical communication for future work.   
\end{remark}

\begin{remark}
    We can also have an additional decryption algorithm in which one can use the classical trapdoor $\td$ to decrypt. We omit it here.
\end{remark}

\begin{remark}
    The extension to encrypting multi-bit message follows naturally from a parallel repetition of scheme above. The security is also naturally extended, see \cite{ananth2023revocable} for discussions.
\end{remark}


\subsection{Proof for \Cref{thm: strong skl implies regular}}

In this section, we prove \Cref{thm: strong skl implies regular}. 

\begin{claim}
    If a construction satisfies \Cref{def:strong_skl} holds for any noticeable $\gamma$,  then it satisfies \Cref{def:regular_skl_security_classical}.
\end{claim}

\begin{proof}
Suppose there's an adversary $\A$
 for \Cref{def:regular_skl_security_classical}:  note that in \Cref{def:regular_skl_security_classical}, in order for $\A$ to win the game, it must give a valid certificate with some inverse polynomial probability $\epsilon_1$ \emph{and} 
 the following holds:
 \begin{align*}
      \lvert \Pr\left[  \text{IND-PKE-SKL}(\A, 1^\lambda, b=0) = 0 \right] - \Pr\left[  \text{IND-PKE-SKL}(\A, 1^\lambda, b=1) = 0 \right] \rvert \geq  \epsilon_2(\lambda)
 \end{align*}
 for some inverse polynomial $\epsilon_1(\lambda)$. Note that the above inequality is implicitly conditioned on the fact that $\VerDel(\pk,\td, \cert) = \valid$. Without loss of generality, we consider the case that $\A$ guesses the correct bit with higher probability. 
 The other case holds symmetrically (by defining the $\gamma$-good decryptor to win with probability less than $1/2 -\gamma$).
 
 We can then deduce the following:
 \begin{align*}
 & \Pr[\text{IND-PKE-SKL}(\A, 1^\lambda, b) \to b' : b' = b] \\
 & =  \frac{1}{2} \cdot  \Pr\left[  \text{IND-PKE-SKL}(\A, 1^\lambda, b=0) \to 0 \right] + \frac{1}{2} \cdot  \Pr\left[  \text{IND-PKE-SKL}(\A, 1^\lambda, b=1) \to 1 \right] \\
 & =  \frac{1}{2} \cdot  \Pr\left[  \text{IND-PKE-SKL}(\A, 1^\lambda, b=0) \to 0 \right] + \frac{1}{2} \cdot (1- \Pr\left[  \text{IND-PKE-SKL}(\A, 1^\lambda, b=1) \to 0 \right]) \\
 &  \geq \frac{1}{2} +  \frac{1}{2}\epsilon_2(\lambda)
 \end{align*}

Therefore, $\A$ (after passting the $\VerDel$-test, should be able to distinguish encryption of $\mu$ from a random value in the ciphertext space with probability $1/2 + 1/2 \cdot \epsilon_2$ for some noticeable $\epsilon_2$.
Clearly, the adversary $\A$ after handing in the deletion certificate is a  $1/2 \cdot\epsilon_2$-good decryptor with noticeable probability. Combining these two conditions above, $\A$ would win in the $1/2 \cdot \epsilon_2$-$\mathsf{StrongSKL}$ game.
 \end{proof}

\ifllncs
\section{Secure Key Leasing with Classical Lessor: Protocol}
\else
\subsection{Secure Key Leasing with Classical Lessor: Protocol}
\fi
\label{sec:protocol_description}

Our one-message key generation protocol and key revocation/deletion protocol follows from our construction in \Cref{sec:construction} and uses a post-quantum digital signature scheme.

\begin{itemize}
    \item \textbf{Interactive Key Generation}:

\begin{itemize}
    \item     the lessor runs $\Setup$ and sends the classical $\mpk = \{\bA_i, \bs_i\bA_i+\be_i\}_{i \in [k]}$ to the
lessee. It keeps the trapdoor $\td$ private.

\item The lessee runs $\KeyGen(\mpk)$ to obtain $\qsk = \frac{1}{\sqrt{2^k}}\bigotimes_i^k \ket{0,\bx_{i,0}} 
+\ket{1,\bx_{i,1}}$ and $\{\by_i\}_{i \in [k]}$. It sends back the public key as $\pk = (\mpk, \{\by_i\}_{i \in [k]})$.

\item The lessor signs a signature $\sigma$ on $\pk$ and publishes $\pk' = (\mpk, \{\by_i\}_{i \in [k]}, \sigma)$ as the final public key.

\end{itemize}

\item \textbf{Encryption:}
The encryption algorithm $\Enc'(\pk',m)$
will accordingly have a first step of verifying the signature $\sigma$ in $\pk'$, if not valid, abort encryption and output $\bot$; else run the encryption in  \Cref{sec:construction}.

\item\textbf{Revocation/Deletion:}
\begin{itemize}
    \item The lessor sends a message to the lessee asking it to delete.

    \item The lessee runs $\delete(\qsk)$ to obtain $\cert = (c_1,\bd_1, \cdots, c_k, \bd_k)$. It sends to the lessor.

    \item Lessor runs $\VerDel(\td,\pk, \cert)$ and outputs $\valid$ or $\invalid$. 
\end{itemize}

\end{itemize}

As discussed, the above interaction in the Interactive Key Genration protocol deals with the impostor security to be discussed in \Cref{sec:relation_security_defs}. We also prove its security in \Cref{sec:relation_security_defs}.


\section{Secure Key Leasing for FHE Security}
\label{sec:fhe}

\subsection{FHE with Secure Key Leasing}
\label{sec:fhe_skl_defs}

To achieve FHE with key leasing, we further define a classical $\Eval$ algorithm:
\begin{description}
    \item $\Eval(\ct_0,\ct_1) \to \ct'$ on input ciphertexts $(\ct_0,\ct_1)$, output a ciphertext $\ct'$ 
\end{description}

\begin{remark} (Levelled FHE with Secure-Key Leasing)
    We could similar define (levelled) homomorphic encryption supporting secret key leasing analogusly, denoted as FHE-SKL. In such a scheme, the $\Setup$ algorithm would also take as input a depth parameter $d$ (that's arbitrary polynomial of the security parameter) and the scheme would support homomorphically evaluating ciphertexts of polynomial sized circuits of depth $d$. The scheme will additionally have an algorithm $\Eval$ that takes as input a circuit $C:\{0,1\}^{\ell}\rightarrow \{0,1\}$ of depth $d$ and polynomial size along with $\ell$ ciphertexts $\ct_1,\ldots,\ct_{\ell}$ encrypting bits $\mu_1,\ldots,\mu_{\ell}$. It should produce an evaluated ciphertext $\hat{\ct}$ that encrypts $C(\mu_1,\ldots,\mu_{\ell})$. We additionally require the scheme to be compact, namely the size of the ciphertexts and evaluated ciphertexts should be polynomial in $\lambda$ and $d$ indepdendent of the circuit size. The security requirements for such a scheme remain the same like an PKE-SKL scheme.
\end{remark}


\begin{remark} (Unlevelled FHE with Secure-Key Leasing)
   We note that the bootstrapping techinque due to Gentry \cite{C:Gentry10} to lift a levelled FHE to unlevelled FHE assuming circular security does not work for FHE with SKL. Observe that if one released encryption of the secret key, then, in the security experiment the adversary could decrypt the secret-key from such an encryption using  the state $\rho_{\sk}$ without disturbing it and then use it later win in the security game.
     Therefore, it seems intuitively impossible to achieve.
     
   We leave the question of giving some meaningful relaxation in defining unleveled FHE-SKL security for future work.
\end{remark}


\subsection{FHE $\Eval$ Algorithm}
To achieve (leveled) fully homomorphic encryption, we add the following algorithm to construction \Cref{sec:construction}, which ollows from \cite{C:GenSahWat13}:

\begin{itemize}
    \item $\mathsf{Eval(\ct_0, \ct_1)}$: the algrithm evaluates an $\mathsf{NAND}$ gate on two ciphertexts $(\ct_0, \ct_1)$.
    It outputs the following:
    $$\bG - \ct_0 \cdot \bG^{-1}(\ct_1)$$
\end{itemize}

Observe that if $\ct_i=\mat{B}\mat{R}_i+\mat{E}_i+\mu_i\mat{G}$ for $i\in \{0,1\}$, the evaluated ciphertext is of the form $\mat{B}\mat{R}+\mat{E}+(1-\mu_1\mu_2)\mat{G}$.

\begin{align*}
\mat{R}=- \mat{R}_0\mat{G}^{-1}(\ct_1) -\mu_0\mat{R}_1\\
\mat{E}=-\mat{E}_0\mat{G}^{-1}(\ct)_1 - \mu_0\mat{E}_1
\end{align*}

Observe that at each NAND operation the norm of the randomness $\mat{R}$ and $\mat{E}$ terms could increase by a factor of $((n+1)k+1) \log q$ (the bigger dimension of $\mat{B}$) as $\mat{G}^{-1}(\ct_1)$ is a binary operation. Assuming subexponential noise to modulus ratio, this means that we can support apriori bounded polynomial depth. In other words, given a depth $d$, assuming subexponential modulus to ratio LWE holds we can set dimensions and the modulus size to be polynomial in $d$ to ensure a supported depth of $d$.

\section{Additional Missing Proofs}
\label{sec:missing_proofs}

\ifllncs
\subsection{Correctness}
\label{sec:correctness}
\paragraph{Decryption Correctness}

The above scheme satisfies decryption correctness.

\begin{proof}
Note that we can write out the entire quantum key as $\qsk =\\ 
\frac{1}{\sqrt{2^k}}\sum_{\mat{b}_j \in \{0,1\}^k}\lvert\bb_{j,1},\bx_{k,\bb_{j,1}}, \cdots, \bb_{j,k},\bx_{1,\bb_{j,k}} \rangle$ 

By applying the decryption procedure, we will have:
    \begin{align*}
              & = \frac{1}{\sqrt{2^k}}\sum_{\mat{b}_j \in \{0,1\}^k}\lvert\bb_{j,1},\bx_{1,\bb_{j,1}}, \cdots, \bb_{j,k},\bx_{k,\bb_{j,k}} \rangle \lvert (-\sum_{i \in [k]}(\bx_{i,\mat{b}_{j,i}}\bA_i + \mat{b}_{j,i}(\bs_i\bA_i+\be_i) )  + \sum_{i \in [k]} \by_{i} ) \\
         & \cdot \bR\bG^{-1}(\mathbf{0} \vert \lfloor\frac{q}{2}\rfloor)  +  \be''\cdot \bG^{-1}(\mathbf{0} \vert \lfloor\frac{q}{2}\rfloor) + \bsk \cdot \mu \bG \bG^{-1} (\mathbf{0} \vert \lfloor\frac{q}{2}\rfloor) \rangle_{\sf work} \ket{0}_{\sf out} \\
         & = \frac{1}{\sqrt{2^k}} \sum_{\mathbf{b_j} \in (\{0,1\}^k} \lvert\bb_{j,1},\bx_{1,\bb_{j,1}}, \cdots, \bb_{j,k},\bx_{k,\bb_{j,k}} \rangle \lvert ( \sum_{i \in [k]} (-\mathbf{b}_{j,i} \cdot \be_{i} +\be_{i}' )\cdot \bR + \be'' )\bG^{-1}(\mathbf{0} \vert \lfloor\frac{q}{2}\rfloor)   + \lfloor\frac{q}{2}\rfloor \cdot \mu \rangle_{\sf work} \ket{0}_{\sf out} \\
         & = \frac{1}{\sqrt{2^k}} \sum_{\mathbf{b_j} \in (\{0,1\}^k} \lvert\bb_{j,1},\bx_{1,\bb_{j,1}}, \cdots, \bb_{j,k},\bx_{k,\bb_{j,k}} \rangle \ket{\mu}_{\sf out} \text{ (round, write on ${\sf out}$ and uncompute ${\sf work}$)}
     \end{align*} 
     
\end{proof}
Since we have $\lVert \sum_i (\be_{i}' + \be_{i}) \rVert \leq 2k\cdot \sqrt{m} \cdot B_{P}, \lVert \be''  \rVert \leq \sqrt{m} \cdot B_{P'}$ and $\lVert (\sum_{i \in [k]} (-\mathbf{b}_{j,i} \cdot \be_{i} +\be_{i}' )\cdot \bR + \be'' )\bG^{-1}(\mathbf{0} \vert \lfloor\frac{q}{2}\rfloor) \rVert  \leq \frac{q}{4}$ for all support $\bb_j \in \{0,1\}^k $ in the final outcome, we will obtain $\mu$ with probability $(1 - O(e^{-q/B_{P'}}))$ for all support in the above state. We can write the final output for $\mu$ in the third (rightmost) register and use $\ct$ to uncompute the second register and recover $\qsk$.

\jiahui{check how many times can we apply gentle measurement!}

\paragraph{Reusability}
Reusability of the decryption key follows from correctness with probability $(1 - O(e^{-q/B_{P'}}))$ where the exponent itself $-\frac{q}{B_{P'}}$ is subexponentially (and of course superpolynomially) small. Therefore, we can apply the gentle measurement lemma \cite{aaronson2018shadow} after each honest decryption, for arbitrarily polynomial times and the error still remains negligible. 

\paragraph{Deletion Verification Correctness} 
The first two steps in $\VerDel$ will pass for any honestly prepared $\{\by_i\}_{i \in [k]}$ with $(1-\negl(\lambda))$ probability.

The $\delete$ procedure will operate on the state $\rho_\sk$ as follows. 
For each $\frac{1}{\sqrt{2}}(\ket{0,\bx_{i,0}}+\ket{1,\bx_{i,1}})$, $i \in [k]$, the $\delete$ procedure will turn the state into:
\begin{align*}
& \frac{1}{\sqrt{2}} \sum \ket{0, \cJ(\bx_{i,0})}+\ket{1,\cJ(\bx_{i,1})} \text{ after applying $\cJ(\cdot)$ and uncomputing $\bx_{i,b_i}$ register} \\
& \to \frac{1}{\sqrt{2^{w+2}}} \sum_{\bd_i,b,u} (-1)^{\bd \cdot \cJ(x_{i,b}) \oplus ub} \ket{u} \ket{\bd_i} \text{ after QFT} \\
& = \frac{1}{\sqrt{2^{w}}} \sum_{\bd_i \in \{0,1\}^{w}} (-1)^{\bd_i \cdot \cJ(\bx_{i,0})} \ket{\bd_i\cdot (\cJ(\bx_{i,0}) \oplus \cJ(\bx_{i,1}))} \ket{\bd_i}
\end{align*}
A measurement in the computational basis will give us result ($c_i = \bd_i \cdot (\cJ(\bx_{i,0}) \oplus \cJ(\bx_{i,1})), \bd_i$), for any $i \in [k]$. Correctness of deletion verification thus follows.

\begin{remark}
Note that in our construction, using $\td$ one can also decrypt. It is easy to define a decryption procedure using $\td$ by the $\mathsf{INV(\td,b, \by)}$-procedure  in \Cref{sec:ntcf_from_lwe}. We omit the details here.     
\end{remark}

\else \fi

\subsection{Proof for \Cref{cor:hyb0_extract}}

\begin{proof}
    
Let us denote $\cD_0$ as the distribution used in $\ati$ of Game 0 and $\cD_{2.k}$ as the distribution used in $\ati$ of Game $2.k$. These two distributions are computationally indistinguishable by the security of LWE (see \Cref{claim:game0_1}, \Cref{claim:game1_2}).

By the property of computationally indistinguishable $\ati$ (\Cref{cor:ati_computational_indistinguishable}), we have:
    \begin{align*}
    \Pr\left[ \ati_{\cP,\cD_{2.k}, 1/2+\gamma-3\epsilon}^{\epsilon,\delta}(\rho_\delete) = 1  \vert    
 \ati_{\cP,\cD_{0}, 1/2+\gamma}^{\epsilon,\delta}(\rho_\delete) = 1 \right]     \geq  1-\negl(\lambda)
\end{align*}

By \Cref{thm:last_hybrid_extract} we have:
  \begin{align*}
     \Pr\left[ \begin{array}{cc}
         &\Ext(\rho_\delete, \pk) \to (\bx_{1}, \cdots \bx_{k}):  \\
         &  \bx_i \text{ is the secret in } \bx_i \bA_i + \be_{i}'
    \end{array}     \vert    
    \ati_{\cP,\cD_{2.k}, 1/2+\gamma}^{\epsilon,\delta}(\rho_\delete) = 1 \right]  \geq 1 - \negl(\lambda)
    \end{align*}

Since \Cref{thm:last_hybrid_extract} holds for any inverse polynomial $\gamma$ and $(\gamma-3\epsilon)$ is also inverse-polynomial ($\epsilon = \gamma/100$),
we have in Game 0, for some negligible $\negl(\cdot)$:
  \begin{align*}
    \Pr\left[ \begin{array}{cc}
         &\Ext(\rho_\delete, \pk) \to (\bx_{1}, \cdots \bx_{k}):  \\
         &  \bx_i \text{ is the secret in } \bx_i \bA_i + \be_{i}'
    \end{array}  \vert \ati_{\cP,\cD_0, 1/2+\gamma}^{\epsilon,\delta}(\rho_\delete) = 1 \right]    \geq 1  -\negl(\lambda) 
    \end{align*}
We can deduce the above from the following probability relations: 

Suppose $\Pr[B\vert A] \geq 1-\negl$ and $\Pr[C\vert B] \geq 1-\negl$ and all events $A,B,C$ happen with noticeable probability, then we have $\Pr[B] \geq \Pr[B\cap A] \geq \Pr[A]-\negl$ and 
$\Pr[B\cap C] \geq \Pr[B]-\negl \geq \Pr[A] -\negl$.
Thus 
$$\Pr[A \cap C] \geq \Pr[A \cap B\cap C] = \Pr[A \cap B] - \Pr[A \cap B \cap \overline{C}] \geq \Pr[A] -\Pr[\overline{C} \cap B] -\negl \geq \Pr[A] - \negl$$. Finally, by conditional probability and the fact that $\Pr[A]$ is noticeable, we have $\Pr[C \vert A] \geq 1-\negl$.  

\end{proof}

\subsection{Proof for the Probability Relations in \Cref{sec:reduction_to_ntcf}}

\label{sec:prob_relations}

\begin{claim}
\label{claim:prob_relation}
    Suppose there are three events $A,B,C$ which all happen with inverse polynomial probability, and suppose $\Pr[A \vert B] \geq 1- \negl(\lambda)$ and $\Pr[B \cap C] \geq 1/p$ for some polynomial $p$, we will have $\Pr[B \cap C] \geq 1/p - \negl'(\lambda)$ for some negligible $\negl'(\cdot)$. 
\end{claim}

\begin{proof}
    Since $\Pr[A \vert B] \geq 1- \negl(\lambda)$, we have $\Pr[A \cap B] = \Pr[B]\Pr[A \vert B]\geq \Pr[B] -\negl'(\lambda)$ for some negligible $\negl'(\cdot)$. Therefore, we have $\Pr[B \cap \overline{A}] \leq \negl_1(\lambda)$.
    Then we can deduce:
    \begin{align*}
        \Pr[A\cap B \cap C] & = \Pr[B \cap C] - \Pr[B \cap C \cap \overline{A}] \\
        & \geq 1/p - \Pr[B \cap \overline{A}] = 1/p - \negl'(\lambda)
    \end{align*}
\end{proof}

\section{Quantum Search-to-Decision LWE with Almost-Perfect Extraction(Clean Version)}
\label{sec:lwe_search_to_decision}
In this section, we provide a means to achieve perfect extraction of the LWE secret in a post-quantum LWE search-to-decision reduction, where the quantum LWE distinguisher may come with an auxiliary quantum state. 

The algorithm will be analogous to what we have in \Cref{sec:ful_extraction_analysis} but the analysis will be simler and parameters will be cleaner, because we will feed the algorithm with actual LWE instances versus random instances.

We nevertheless provide it here in an independent section because we believe that the techniques involved will be of independent interest and the statement here may also be used in a black-box way for future works. 

\paragraph{$\gamma$-good LWE distinguisher} Firstly, we define the notion of an $\gamma$-good quantum distinguisher for decisional $\lwe_{n,m,q,\sigma}$ with some auxiliary quantum input. 

The definition is analogous to the $\gamma$-good decryptor in \Cref{def:gamma_good_decryptor}, but with LWE sample/uniform sample as challenges.

\begin{definition}[$\gamma$-good LWE Distinguisher]
\label{def:good_lwe_distinguisher}

A quantum distinguisher can be described by a quantum circuit $U$ and an auxiliary state $\rho$. W.l.o.g, we denote such a distinguisher simply as $\rho$ to represent a quantum state that also contains a classical description of the circuit.
\begin{itemize}
    \item For some arbitrary secret $\bs\in \Z_q^n$ and $\be \gets \cD_{\Z^m, \sigma}$\footnote{For our reduction, we only need such a weak distinguisher that works with respect to a presampled secret $\bs$ and error $\be$.}.
    
       \item Let $\mathcal{P_\lwe} = (P_\lwe, Q_\lwe)$ be the following mixture of projective measurements (in the sense of Definition \ref{def:mixture_of_projective}) acting on a quantum state $\rho$. We denote the following distribution we sample from as $\cD(\lwe)$.

       \begin{itemize}
       \item Sample $\bu_0 = \bs \bA + \be$, where $\bA \gets \Z^{n\times m}_q$.
       \item Compute $\bu_1 \gets 
       \Z_q^m$, a uniform random vector.
       \item Sample a uniform $b \leftarrow \{0,1\}$. 
       \item Run the quantum distinguisher $\rho$ on input $(\bA,\bu_b)$. Check whether the outcome is $b$. If so, output $1$, otherwise output $0$.
      \end{itemize}
       \item Let $\ti_{1/2 + \gamma}(\cP_\lwe)$ be the threshold implementation of $\cP_\lwe$ with threshold value $\frac{1}{2} + \gamma$, as defined in \Cref{def:thres_implement}. Run $\ti_{1/2 + \gamma}(\cP_\lwe)$ on $\rho$, and output the outcome. If the output is $1$, we say that the quantum distinguisher is $\gamma$-good with respect to the secret $\bs$ and error $\be$, otherwise not.
   \end{itemize}

\end{definition}

\begin{theorem}
\label{thm:lwe_serachtodecision_perfect}
   Let $\lambda \in \N$ be the security parameter. Let $n,m \in \N$ be
integers and let $q$ be a prime. If there is a $\gamma$-good quantum distinguisher for decisional $\lwe_{n,m,q,\sigma}$ with some auxiliary quantum input running in time $T$ for some noticeable $\gamma$, there there exists a quantum algorithm that solves search $\lwe_{n,m,q,\sigma}$ with probability $(1-\negl(\lambda))$, running in time $T' = Tnp\cdot \mathsf{poly}(1/\epsilon,1/\log\delta)$ where $p$ is the size of the domain for each entry $\bs_i$ in the secret $\bs$.  and $\epsilon,\delta \in (0,1)$\footnote{In our algorithm, we will set $\epsilon$ to $\gamma/100pn$ and set $\delta$ to be exponentially small.}. 
\end{theorem}





\paragraph{The Search-to-Decision Algorithm} On a search LWE challenge $(\bA, \bu = \bs \bA + \be)$ and a decision LWE distinguisher with state $\rho$ which is $\gamma$-good. We perform the following algorithm:

\begin{itemize}
    \item For coordinate $i \in 1, 2, \cdots , n$: 
\begin{itemize}
    \item For  $g_i \in k,k+1,\cdots, k+p-1$: 
\begin{enumerate}

\item Let $\rho$ be the current state from the quantum distinguisher. 

\item Let $\gamma : = \gamma - 3\epsilon$, where $\epsilon = \frac{\gamma}{100 pn}$

\item Perform $\ati^{\epsilon,\delta}_{, \cP,\cD(g_i), 1/2+\gamma}$ on the state $\rho$ in \Cref{sec:ati} according to the following mixture of measurements $\cP_{\cD(g_i)}$ (we denote the following distribution we sample from as $\cD(g_i)$.):
\begin{itemize}
\item Sample a random $\bc \gets \Z_q^{1 \times m}$, and let matrix $\bC \in \Z_q^{n \times m}$ be a matrix where the $i$-th row is $\bc$ and the rest of rows are $\mathbf{0}$'s.

\item Let $\bA' = \bA + \bC$ and let $\bu_0' = \bu + g_i \cdot \bc$. 

\item Sample another uniformly random vector $\bu_1' \gets \Z_q^{1\times m}$ 

\item Flip a bit $b \gets \{0,1\}$.

\item Run the quantum distinguisher $\rho$ on input $(\bA', \bu_b)$. Check whether the outcome is $b$. If so, output $1$, otherwise output $0$.

\end{itemize}
\item If $\ati^{\epsilon,\delta}_{\cP,\cD(g_i),1/2+\gamma}$ outputs 1 (that $\rho$ is an $\gamma$-good distinguisher), then set $\bs_i = g_i$ and move on to next coodinate: $i: =i+1$. 

\item If $\ati^{\epsilon,\delta}_{ \cP,\cD(g_i),1/2+\gamma}$ outputs 0 (that $\rho$ is not an $\gamma$-good distinguisher), then let $g_i : = g_i+1$ and go to step 1.
\end{enumerate}
\end{itemize}

\end{itemize}


We make the following claims: 
\begin{claim}
\label{claim:equivalent_ti}
For all $i\in [n]$, when the guess $g_i$ is the correct $\bs_i$ in a search LWE challenge $(\bA, \bu = \bs \bA + \be)$, then $\Pr[\ati^{\epsilon, \delta}_{\cP, D(g_i), 1/2+\gamma}\, \rho \to 1] = \Pr[\ati^{\epsilon, \delta}_{\cP, D(\lwe), 1/2+\gamma}\, \rho \to 1]$. 
\end{claim}

\begin{proof}
    When $g_i = \bs_i$, we have $\bA' = \bA + \bC$ where $\bC$ is everywhere 0 except the $i$-th row being uniformly random $\bc$. We also have $\bu_0' = \bs \bA  + \be + g_i \cdot \bc = [(\bs_i(\bA_{1,i}+\bc_{1}) + \sum_{j\neq i}\bs_{j}(\bA_{1,j}+0), \cdots, (\bs_i(\bA_{m,i}+\bc_{m}) + \sum_{j\neq i}\bs_{j}(\bA_{m,j}+0)] + \be^\top = \bs^\top \bA' + \be$, where $\bA_{x,y}$ denotes the entry in $x$-th column and $y$-th row. 

    Therefore, for all $i\in [n]$ and each $g_i \in [k, k+p-1]$, when the guess $g_i$ is the correct $\bs_i$ in the search LWE challenge $(\bA, \bu = \bs \bA + \be)$, then the mixture of measurements $(\cP_{\cD(g_i)}, \cQ_{\cD(g_i)}) $ is the same as the measurements $(\cP_\lwe, \cQ_\lwe)$ defined in \Cref{def:good_lwe_distinguisher}, which test the distinguisher on LWE instances with respect to the same secret $\bs$ and error $\be$. Since they are statistically the same measurements, their $\ati$ will yield the same result.
\end{proof}

\begin{claim}
\label{claim:uniform_ti}
For all $i\in [n]$, when the guess $g_i$ is incorrect, in the above mixture of projections $\cP = (P_{\cD_{g_i}}, Q_{\cD_{g_i}})$, for any noticeable $\gamma$, and any quantum distinguisher $\rho$, the following claims hold:
\begin{enumerate}
    \item   $\Pr[\ati^{\epsilon, \delta}_{\cP, D(g_i), 1/2+\gamma}\, \rho \to 0] = 1- 
    \delta$

    \item Let $\rho'$ be the state after applying $\ati^{\epsilon,\delta}_{\cP, \cD_{g_i},1/2+\gamma}$ on $\rho$, we have that $\lVert \rho-\rho'\rVert_{\Tr} = 0$.
\end{enumerate}

\end{claim}

\begin{proof}
    We first consider the perfect projective threshold implementation $\ti_{1/2+\gamma}(\cD(g_i))$ when $g_i \neq \bs_i$. The vector $\bu_0' = \bs \bA  + \be + g_i \cdot \bc = [(g_i\cdot \bc_{1}+ \sum_{j}\bs_{j}\bA_{1,j}, \cdots, g_i\cdot \bc_{m}+ \sum_{j}\bs_{j}\bA_{m,j}] + \be^\top$. Since $\bc$ is uniformly random, the entire $\bu_0'$ now becomes uniformly random and statistically the same as $\bu_1'$. Therefore no algorithm can have a non-negligible advantage in distinguishing them.
    That is, all possible states will be projected onto the result 0, when one applies the projective implementation $\ti_{1/2+\gamma}(\cP_{cD(g_i)})$ for any noticeable $\gamma$:  $\Pr[\ti_{1/2+\gamma}(\cP_{cD(g_i)})\rho \to 0] = 1$.

    Accordingly, by \Cref{cor:ati_thresimp}, there exists $\epsilon, \delta$ such that $\Pr[\ati^{\epsilon, \delta}_{\cP, D(g_i), 1/2+\gamma}\, \rho \to 0] = 1- \delta(\lambda)$. 

    Property 2 follows directly from \Cref{cor:invariant_ati_same_distr} because the two distributions in $\cD_{g_i}$ are exactly the same and thus we can apply the corollary.

\end{proof}

\paragraph{Correctness}
The correctness follows from the same analysis as in \Cref{lem:extractor_correctness} and we omit the details.

Some simplifications in this "clean LWE" extractor compared to the one in \Cref{sec:ful_extraction_analysis}: our $\ati^{\epsilon,\delta}_{\cP, \cD_{g_i},1/2+\gamma}$ with incorrect $g_i$'s will not incur difference on the trace distance due to the distributions $\cD_{g_i}, \cD_\unif$ being the same; our $\ati^{\epsilon,\delta}_{\cP, \cD_{g_i},1/2+\gamma}$ with correct $g_i$'s will incur error of only $3\delta$ every time after being applied, because distributions $\cD_{g_i}, \cD_\lwe$ are the same
instead of being only statistically close.

Overall the algorithm outputs correct secret with probability $(1-3\delta)^{n} \cdot (1-\delta)^{pn} \geq 1 - pn\cdot \delta$, where $(1-3\delta)^{n}$ is the probability that 
all $\ati^{\epsilon,\delta}_{\cP, \cD_{g_i},1/2+\gamma}$ with correct $g_i$'s output 1 and $(1-\delta)^{pn}$ is the probability that 
all $\ati^{\epsilon,\delta}_{\cP, \cD_{g_i},1/2+\gamma}$ with incorrect $g_i$'s output 0. $ 1 - pn\cdot \delta \geq 1 - \negl(\lambda)$ because our $p$ is at most subsponential and $\delta$ can be an inverse exponential function.

\paragraph{Running Time}
By \Cref{cor:ati_thresimp}, the running time of each $\ati^{\epsilon, \delta}_{\cP, D(g_i), 1/2+\gamma}$ is $T \cdot \mathsf{poly}(1/\epsilon,1/\log\delta)$, so the entire algorithm is $Tnp \cdot \mathsf{poly}(1/\epsilon,1/\log\delta)$.

Note that when we work with small LWE secret, where $p$ is polynomial then our $\epsilon = \frac{\gamma}{100pn}$ is polynomial, so the entire algorithm is polynomial. When $p$ is subexponential, then $\epsilon$ is 
 an inverse subexponential function, so the entire algorithm runs in subexponential time (which we allow when working with subexponential-hardness of LWE).

\paragraph{Conclusion} We thus conclude \Cref{thm:lwe_serachtodecision_perfect}.

\ifllncs
\section{More on Related Works}
\else
\subsection{More on Related Works}
\fi
\label{sec:other_related_works}

We addtionally discuss some other related works.

\paragraph{Quantum Copy Protection} 
\cite{aaronson2009quantum} first built copy-protection for all unlearnable functions based on a quantum oracle.  
 \cite{aaronsonnew} showed a construction for all unlearnable functions based on a classical oracle. \cite{coladangelo2021hidden, liu2022collusion} showed copy protection for signatures, decryption and PRF evaluation in the plain model.
 \cite{coladangelo2020quantum,ananth2022feasibility} constructed copy-protection for point functions and compute-and-compare functions in QROM, the latter improving the security of the former.

Regarding the negative results:
\cite{ananth2020secure} demonstrated that it is impossible to have a copy-protection scheme for all unlearnable circuits in the plain model, assuming LWE and quantum FHE. 

\paragraph{More on Secure Software Leasing} 
\cite{aaronsonnew} observed that under a definition essentially equivalent to infinite-term SSL, namely copy-detection, one could obtain a black-box construction for infinite-term SSL from watermarking and public-key quantum money.
\cite{kitagawa2021secure}
 constructed finite-term SSL for PRFs and compute-and-compare functions from (subexponential) LWE, with similar observations.

\cite{broadbent2021secure,coladangelo2020quantum} constructed secure software leasing for point functions and compute-and-compare functions; \cite{broadbent2021secure} is information-theoretically secure and \cite{coladangelo2020quantum} is secure under QROM. They both used a stronger version of finite-term SSL security: 
while the vendor will honestly check the returned state from the adversary, the adversary can execute the leftover half of its bipartite state maliciously, i.e., not following the instructions in the evaluation algorithm.
SSL security of this stronger finite-term variant is only known for point/compute-and-compare functions up till now.

\paragraph{Unclonable Encryption and Certified Deletion}
Unclonable encryption is studied in \cite{gottesman2001quantum,Broadbent2019UncloneableQE,ananth2021unclonable,ananth2022feasibility}. It is an encryption scheme where the ciphertext is encoded in quantum information, so that the adversary cannot split it into two copies that can both decrypt when given a (classical) decryption key.

Certified deletion of ciphertext is studied in various works \cite{broadbent2020certified,poremba2022certified,hiroka2021certified,bartusek2023certified,bartusek2023certfied-obfuscation}, where the ciphertext is also encoded in a quantum state and given to a server. When asked by the client to delete, the server must provide a proof of deletion of the ciphertext. This bears some similarity to our setting. But note the major difference is that in the certified deletion of ciphertext setting, we need to show that the server no longer has information about the encrypted message(a bit string), while in our setting we need to show the server is deprived of a functionality. Therefore, the results and techniques in secure key leasing and certified deletion of ciphertext are most of the time incomparable.

\else

\section{Acknowledgement}
Orestis Chardouvelis and Aayush Jain were supported by gifts from CyLab of CMU and Google.
Jiahui Liu was supported by DARPA under Agreement No. HR00112020023, NSF CNS-2154149 and a Simons Investigator award. 

We thank an anonymous reviewer 
who points out an issue in the analysis of the LWE search-to-decision extractor analysis in a previous version of this paper, which has now been fixed.

We thank Qipeng Liu, Mark Zhandry and the authors of \cite{ananth2023revocable} for helpful discussions.

\bibliographystyle{alpha}
\bibliography{bib,abbrev3,custom,crypto} 
\appendix

\fi

\end{document}